\documentclass[mnsc]{informs3}

\OneAndAHalfSpacedXII 
\usepackage{graphicx}
\usepackage[margin=1in]{geometry} 
\usepackage{booktabs} 

\usepackage[ruled]{algorithm2e} 

\SetAlFnt{\small}
\SetAlCapFnt{\small}
\SetAlCapNameFnt{\small}
\SetAlCapHSkip{0pt}
\IncMargin{-\parindent}

\usepackage{mathrsfs} 

\usepackage[framemethod=tikz]{mdframed}

\usepackage{nicefrac}

\usepackage{pifont}

\usepackage{cancel}
\definecolor{BrickRed}{rgb}{0.8,0.25,0.33}
\usepackage{booktabs}
\usepackage{multirow}
\usepackage{comment}
\usepackage{dsfont}
\usepackage{bigstrut}
\usepackage{caption}
\usepackage{subcaption}
\usepackage{tcolorbox}
\usepackage[noend]{algorithmic}
\usepackage{todonotes}
\usepackage{amsmath,cases}
\usepackage{float}
\usepackage{comment,tabulary,booktabs,multirow,mathtools}

\usepackage{tikz}
\usetikzlibrary{chains,arrows.meta,quotes,bending}
\usetikzlibrary{decorations.pathreplacing}

\usepackage{xcolor}

\definecolor{BrickRed}{rgb}{0.8,0.25,0.33}
\definecolor{darkgreen}{HTML}{06402B}
\definecolor{navyblue}{HTML}{000080}

\usepackage[pagebackref,colorlinks,citecolor=blue,linkcolor=BrickRed]{hyperref}

\usepackage{mathtools}
\usepackage{cancel}
\usepackage{mwe}
\usepackage[shortlabels]{enumitem}
\usepackage{bbm}
\usepackage{cases}
\usepackage{soul}
\usepackage{tabularx}
\usepackage{amssymb}
\usepackage{booktabs}
\usepackage{url}
\usepackage{hyperref}
\usepackage{thm-restate}
\usepackage{cleveref}

\usepackage{natbib}
 \bibpunct[, ]{(}{)}{,}{a}{}{,}%
 \def\bibfont{\small}%
\newcommand{\qedsymb}{\hfill{\rule{2mm}{2mm}}}

\newcommand{\newreptheorem}[2]{%
\newenvironment{rep#1}[1]{%
 \def\rep@title{#2 \ref{##1}}%
 \begin{rep@theorem}}%
 {\end{rep@theorem}}}
\makeatother
\newreptheorem{theorem}{Theorem}
\newreptheorem{lemma}{Lemma}
\newreptheorem{proposition}{Proposition}
\newreptheorem{corollary}{Corollary}
\newreptheorem{definition}{Definition}
\newreptheorem{fact}{Fact}

\crefname{thm}{Theorem}{theorems}
\crefname{cla}{Claim}{claims}
\crefname{lem}{Lemma}{lemmas}
\crefname{fact}{Fact}{facts}
\crefname{exm}{Example}{examples}

\Crefname{thm}{Theorem}{theorems}
\Crefname{cla}{Claim}{claims}
\Crefname{lem}{Lemma}{lemmas}
\Crefname{fact}{Fact}{facts}
\Crefname{exm}{Example}{examples}

\newcommand{\appendix}{
    \par
    \setcounter{section}{0}
    \setcounter{subsection}{0}
    \renewcommand\thesection{\Alph{section}}
    \renewcommand\thesubsection{\thesection.\arabic{subsection}}
    \renewcommand{\sectionname}{Appendix}
}

\usepackage{natbib}
 \bibpunct[, ]{(}{)}{,}{a}{}{,}%
 \def\bibfont{\small}%
\TheoremsNumberedThrough     
\ECRepeatTheorems

\EquationsNumberedThrough    

\theoremstyle{plain}
\renewenvironment{proof}[1][]{\begin{trivlist} \item[\hspace{\labelsep}{\bf \noindent Proof#1.\/}] }{\qedsymb\end{trivlist}}%
\newcommand{\MyAtop}[2]{\genfrac{}{}{0pt}{}{#1}{#2}}

\newcommand{\bs}[1]{\boldsymbol{#1}}
\newcommand{\pr}[1]{{\rm Pr} \left[ #1 \right]}
\newcommand{\prpar}[1]{{\rm Pr} [ #1 ]}

\newcommand{\ex}[1]{{\mathbb E} \left[ #1 \right]}
\newcommand{\expar}[1]{{\mathbb E} [ #1 ]}

\newcommand{\expartwo}[2]{{\mathbb E}_{#1} [ #2 ]}
\newcommand{\extwo}[2]{{\mathbb E}_{#1} \left[ #2 \right]}

\newcommand{\eps}{\varepsilon}
\newcommand{\bb}{\mathbb}

\newcommand{\nat}{\mathbb{N}}

\newcommand{\ns}{n_f}

\newcommand\restr[2]{{
  \left.\kern-\nulldelimiterspace 
  #1 
  \littletaller 
  \right|_{#2} 
  }}

\makeatletter
\newcommand{\vastmin}{\bBigg@{3}}
\newcommand{\vast}{\bBigg@{4}}
\newcommand{\Vast}{\bBigg@{7}}
\newcommand{\boldvast}{\bBigg@{9}}
\makeatother

\newcommand{\poly}{\mathrm{poly}}

\newcommand{\vv}{\mathcal{S}^{\rm infreq}}
\newcommand{\sss}{\mathcal{S}^{\rm freq}}

\newcounter{mysubequations}

\newcommand{\qedsymbol}{\hfill{\rule{2mm}{2mm}}}
\allowdisplaybreaks

\tcbset{
    theorembox/.style={
        colback=white,
        colframe=black,
        width=\textwidth,
        arc=5mm,
        auto outer arc,
        boxsep=5pt, 
        before skip=10pt, 
        after skip=10pt,  
        left=6pt, 
        right=6pt, 
        top=6pt,   
        bottom=6pt 
    }
}

\usepackage[suppress]{color-edits}
\usepackage{color}         
\usepackage{xcolor}

\definecolor{ForestGreen}{RGB}{0,180,0} 

\addauthor{Ali}{red}
\addauthor{Alireza}{ForestGreen}
\addauthor{Amin}{blue}

\begin{document}
\RUNTITLE{Adaptive Approximation Schemes for Matching Queues}
\TITLE{\fontsize{19}{16}\selectfont Adaptive Approximation Schemes for Matching Queues}

\ARTICLEAUTHORS{%
\AUTHOR{Alireza AmaniHamedani}
\AFF{London Business School, \EMAIL{aamanihamedani@london.edu}} 
\AUTHOR{Ali Aouad\thanks{This work was partially supported by the UKRI Engineering and Physical Sciences Research Council [EP/Y003721/1].}} 
\AFF{Massachusetts Institute of Technology, \EMAIL{maouad@mit.edu}} 
\AUTHOR{Amin Saberi}
\AFF{Stanford University, \EMAIL{saberi@stanford.edu}} 
}

\ABSTRACT{
We study a continuous-time, infinite-horizon dynamic bipartite matching problem. Suppliers arrive according to a Poisson process; while waiting, they may abandon the queue at a uniform rate. Customers on the other hand must be matched upon arrival. The objective is to minimize the expected long-term average cost subject to a throughput constraint on the total match rate.

Previous literature on dynamic matching focuses on  ``static'' policies, where the matching decisions do not depend explicitly on the state of the supplier queues, achieving constant-factor approximations. By contrast, we design ``adaptive'' policies, which leverage queue length information, and obtain near-optimal polynomial-time algorithms for several classes of instances.

First, we develop a bi-criteria fully polynomial-time approximation scheme for dynamic matching on networks with a constant number of queues---that computes a $(1-\eps)$-approximation of the optimal policy in time polynomial in both the input size and $1/\eps$.  A key new technique is a hybrid LP relaxation, which combines static and state-dependent LP approximations of the queue dynamics, after a decomposition of the network. Networks with a constant number of queues are motivated by deceased organ donation schemes, where the supply types can be divided according to blood and tissue~types. 

The above algorithm, combined with a careful cell decomposition gives an efficient polynomial-time approximation scheme for dynamic matching on Euclidean networks of fixed dimension. The Euclidean case is of interest in ride-hailing and spatial service platforms, where the goal is to fulfill as many trips as possible while minimizing driving distances.
}




\maketitle

\section{Introduction}
Matching problems arise in a wide array of market design settings, especially in digital markets, organ donation schemes, and barter systems~\citep{echenique2023online}. Traditionally, the literature in operations research and computer science models such problems as variations of online bipartite matching, introduced by \cite{karp1990optimal}. Offline vertices are available in advance and get matched over time, whereas  online vertices arrive sequentially. There is a very rich literature studying this problem in various settings, such as adversarial arrivals  \citep{mehta2007adwords,devanur2013randomized, huang2020adwords, fahrbach2022edge} or stochastic arrivals \citep{feldman2009online, manshadi2012online, huang2021online, huang2022power,ezra2022prophet}; the state-of-the-art results can be found in a recent survey by \cite{huang2024online}. 



In most real-world markets, however, both sides of the market are dynamic and continuously changing. For instance, new patients are regularly enrolled in organ donation programs, and ride-hailing drivers log in based on their flexible work hours and availability. Another important issue is abandonment.\footnote{The abandonment risk is further compounded by agents often participating in multiple markets simultaneously to enhance their chances of securing a match. For example, in the context of Kidney Paired Donation (KPD), patients may enroll in several KPD registries and subsequently exit a given market if they receive a donation elsewhere. Evidence suggests that existing matching processes account for this type of participant abandonment~\citep{gentry2015best, amanihamedani2023spatial}.} Patients in need of organ transplants can tolerate a limited waiting period, but if a suitable match is not found in time, their condition may deteriorate significantly or they may become ineligible for transplantation. On the other side, deceased donor organs are highly time-sensitive, as they must be transplanted within strict viability windows, ranging from just a few hours to a maximum of a few days.\footnote{The resulting time pressure may cause market inefficiencies. For example, a recent U.S. congressional hearing on organ transplant markets raised a high level of organ waste in the deceased donor matching process, ``as few as one in five potential donor organs are recovered.''
See \href{https://tinyurl.com/organ-waste-congress-hearing}{https://tinyurl.com/organ-waste-congress-hearing}, accessed October 2024. This phenomenon may plausibly contribute to mortality in the long run.} In ride-hailing, if a driver is not matched within a short time window, they may log off and seek alternative work. If riders are not matched promptly, they may switch to other transportation services.

In this context, optimization-based matching algorithms could improve the efficiency of centralized dynamic matching markets. While adversarial models have been studied previously~\citep{huang2018match,huang2019tight,huang2020fully, ashlagi2019edge},  it is natural to represent the market dynamics via a stochastic process with agents' arrivals and abandonment over time~\citep{akbarpour2020thickness,collina2020dynamic,aouad2022dynamic}. This queueing-theoretic model of (edge-weighted) online matching, which we formalize below, is closely related to  parallel-servers and two-sided matching systems, which have been extensively studied in the applied probability literature. {Previous work in this area focuses on large-market analyses, using fluid or heavy-traffic scaling~\citep{ozkan2020dynamic, hurtado2022logarithmic, varma2023dynamic, aveklouris2024matching}, or identifies stability criteria \citep{jonckheere2023generalized,begeot2023stability} and properties of stationary distributions~\citep{gardner2020product,castro2020matching,moyal2023sub}. A notable line of papers constructs near-optimal, low regret algorithms for dynamic matching systems~\citep{nazari2019reward,kerimov2023optimality, kerimov2024dynamic, wei2023constant}. Importantly, however, these results all require the agents to be sufficiently patient for a thick market to form, which is incompatible with uncontrolled  departures. Recently, \cite{kohlenberg2024cost} developed a sharp analysis of the impact of abandonment on a two-sided queue.



Less is known about the design of efficient matching algorithms under abandonment for general networks of queues.  The corresponding sequential optimization problem, termed {\em dynamic (or stationary) matching}, has been introduced and studied in a recent line of research~\citep{collina2020dynamic,aouad2022dynamic, kessel2022stationary, patel2024combinatorial,kohlenberg2024greedy}. Existing results rely on relatively simple linear programming (LP) relaxations {with state-independent variables,} to develop {\em static} policies.\footnote{{See \cite{ashlagi2022price} for an analysis of a related waiting list setting where the queue lengths function as state-dependent (adaptive) prices.}} These policies, informally speaking, are state-oblivious: they do not incorporate real-time information about the system's state (e.g., the number of agents in the queues) when making matching decisions in real time, {contrary to {\em adaptive} policies}. Static policies provide constant-factor approximations but suffer from an inherent constant-factor optimality gap,  as we illustrate in \Cref{ssec:adaptivity}. Consequently, adaptive policies are valuable in settings where near-optimal performance is desirable.

The main contribution of this work is the development of efficient {\em dynamic} approximation schemes for matching queues across a broad class of inputs. 
We resolve, in part, open questions raised by \cite{patel2024combinatorial}.\footnote{\citet{patel2024combinatorial} study the stationary matching problem, and its extension to combinatorial allocation. They list as an open question to identify classes of inputs for which there exist (F)PTASes.} 

A significant challenge in analyzing adaptive policies is computing the stationary distribution induced by these policies. Indeed, optimal policies correspond to the solution of an infinite-dimensional dynamic program. To overcome this, we propose a new hybrid LP relaxation technique. This method bridges online LP-rounding concepts with drift analysis, enabling effective prioritization of queues and capturing the differences in timescales between newly defined notions of ``short'' and ``long'' queues.

\subsubsection*{{Problem formulation.}} We are given an edge-weighted bipartite network $G = ({\cal S}, {\cal C}, E)$. We use {\em suppliers} $i\in {\cal S}$ and {\em customers} $j\in {\cal C}$ to refer to the two sides of the markets, respectively.\footnote{Note that this terminology may not be the most natural in certain applications such as organ donation {since donors who \emph{supply} organs are impatient whereas the organ \emph{customers} are more patient.}} Suppliers of type $i$ arrive according to a Poisson process with rate $\lambda_i$ and independently abandon the market after an exponentially distributed duration of rate $\mu$, if they are not matched before. Without loss of generality, we fix $\mu=1$ unless specified otherwise. Customers of type $j$ arrive with rate $\gamma_j$, upon which they can be matched to an available supplier, or immediately leave the system. Matching a supplier of type $i$ to a customer of type $j$ incurs a cost $c_{i,j}\geq 0$, measuring the compatibility between $i$ and~$j$.


A policy is a time-adapted process that can match an arriving customer to any available supplier in the queueing network.
For any policy $\pi$, we define the {\em throughput rate} $\tau(\pi)$ as the expected long-term average rate of realized matches, and the {\em cost rate} $c(\pi)$ as the expected long-term average cost of those matches. That is, $\tau(\pi) = \liminf_{t \to \infty} \frac{\ex{T^\pi(t)}}{t}$ and  $ c(\pi) = \limsup_{t \to \infty} \frac{\ex{C^\pi(t)}}{t}$, where $M^\pi(t)$ is the cumulative number of matches until time $t$ and $C^\pi(t)$ is the cumulative cost of those matches.


Given a {\em cost-throughput target} $(c^*, \tau^*)$, our goal is to find a policy that satisfies $c(\pi) \leq c^*$ and $\tau(\pi) \geq \tau^*$. A policy $\pi$ is $\alpha$-approximate for some $\alpha\in (0,1)$, with respect to the target $(c^*,\tau^*)$, if it achieves  $c(\pi) \leq \frac{1}{\alpha}\cdot c^*$ and $\tau(\pi) \geq \alpha \tau^*$. We treat the bi-criteria target $(c^*, \tau^*)$ as an input to the dynamic matching problem, but our subsequent algorithms can easily approximate the Pareto frontier of achievable cost-throughput rates, without the need to specify a target as input. 

It is important to note that this bi-criteria optimization setting encompasses and refines the reward maximization studied in previous literature \citep{aouad2022dynamic, kessel2022stationary, amanihamedani2024improved}, i.e., a polynomial-time approximation scheme for our problem can be efficiently converted to an approximation scheme for the environments where the objective is to maximize the expected long-term average rewards. However, our dual objective is more pertinent in scenarios where abandonment is highly costly, such as patient mortality and wasted organs, but the match compatibility is also essential.

\subsection{Static versus Adaptive Policies}\label{ssec:adaptivity}

\begin{figure}[ht]
    \centering
    \caption{Performance of static versus adaptive policies and the effect of $\boldsymbol{\tau^*}$ and $\boldsymbol{\mu}$ on the adaptivity gap}
    \label{fig:adaptivity}
    \vspace{0.4cm}   
    \begin{minipage}{0.49\textwidth}
        \includegraphics[width=\linewidth]{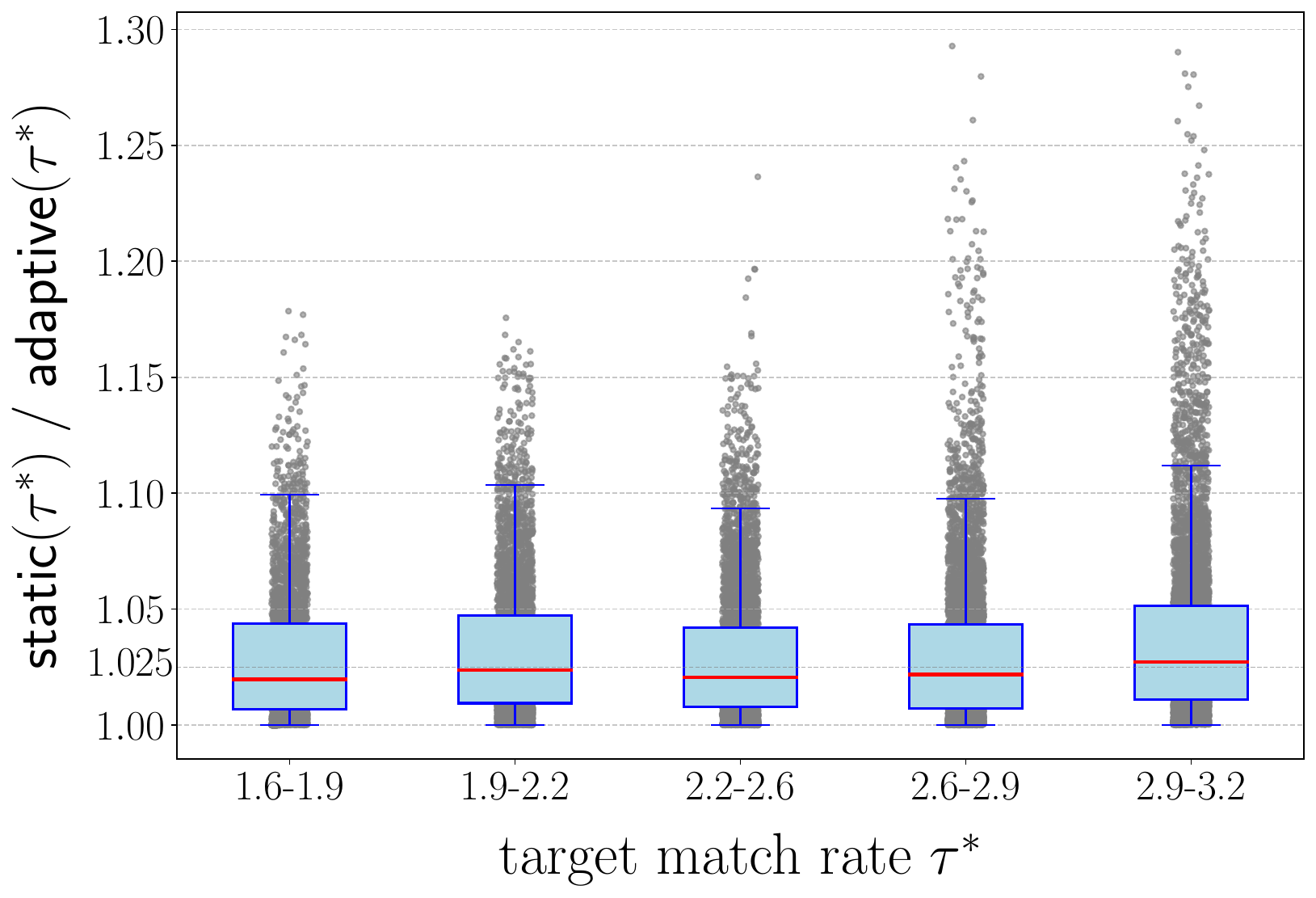}
        \captionsetup{labelformat=empty}
        \caption*{\footnotesize (a) The adaptivity gap as a function of $\boldsymbol{\tau^*}$}
    \end{minipage}
    \hfill
    \begin{minipage}{0.49\textwidth}
        \vspace{-0.11cm}
        \includegraphics[width=\linewidth]{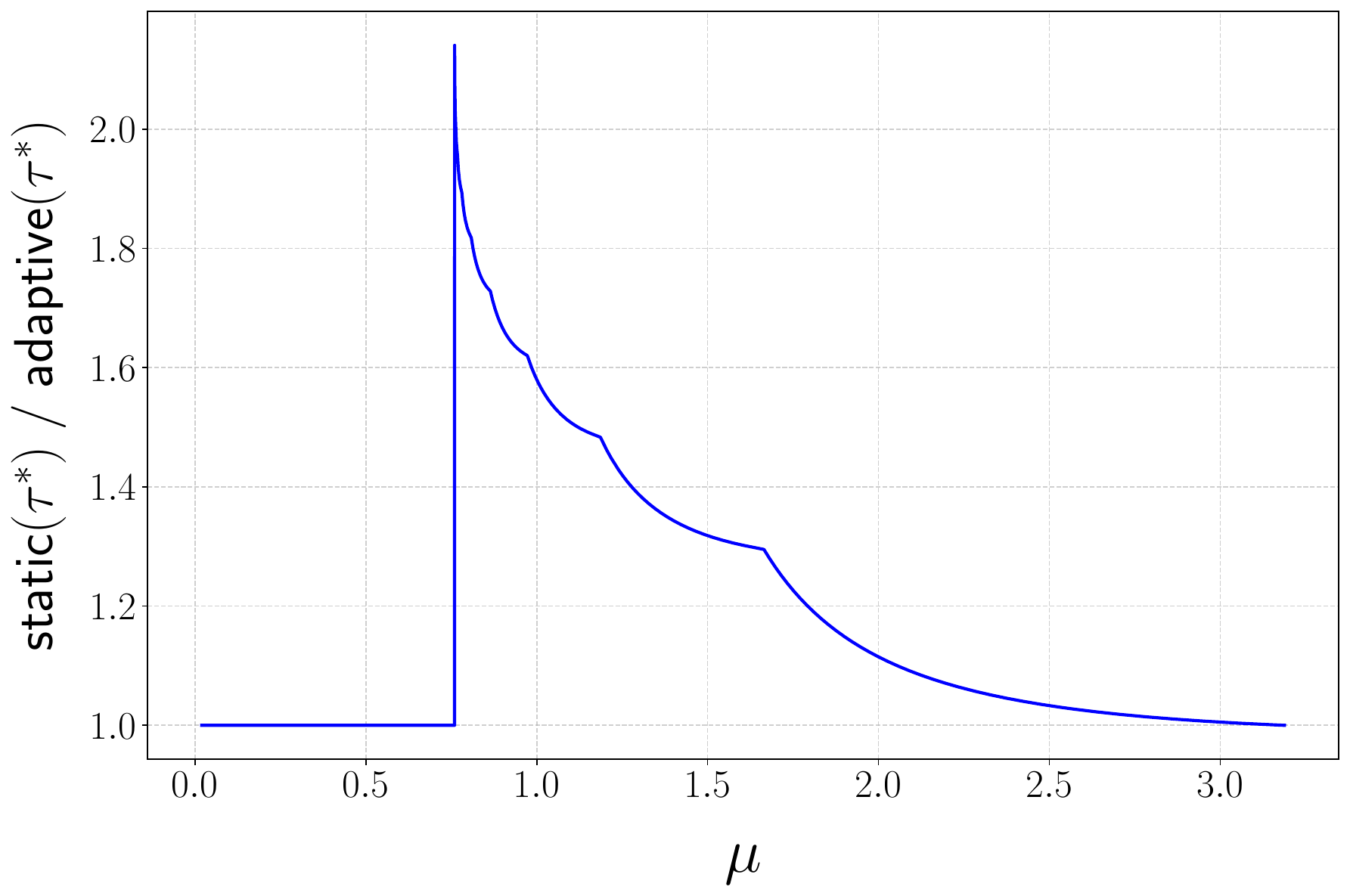}
        \captionsetup{labelformat=empty}
        \caption*{\footnotesize (b) The adaptivity gap as a function of $\boldsymbol{\mu}$}
    \end{minipage}     
        \begin{minipage}{1\textwidth}
    \raggedright
    \vspace{0.3cm}
    \footnotesize
    On the left, the dashed lines within each box represent the 25th, 50th (median), and 75th percentiles, respectively. On the right, we vary $\mu$ from 0 to 3 and report the adaptivity gap of the hard instance.
 \end{minipage}
\end{figure}

    {We illustrate the gap between static and adaptive policies in simple examples.} Consider the setting with a single supplier type, arriving with rate $\lambda=4$, and three customer types with arrival rates $\boldsymbol{\gamma} = \{\gamma_1, \gamma_2, \gamma_3\}$ and their corresponding matching costs $\boldsymbol{c} = \{c_1, c_2, c_3\}$ with $c_1 \leq c_2 \leq c_3$. We generate 1000 instances with random $\boldsymbol{\gamma}$-s and $\boldsymbol{c}$-s.\footnote{The arrival rates $\boldsymbol{\gamma}$ follow a random arithmetic sequence with the first value drawn uniformly from $[1, 2]$ and the consecutive differences are drawn independently and uniformly from $[0, 2]$. The costs $\boldsymbol{c}$ are independently and uniformly chosen in $[0, 2]$. The values are then ordered so that $\gamma_1 \leq \gamma_2 \leq \gamma_3$ and $c_1 \leq c_2 \leq c_3$. Recall that $\mu =1$.} Given a throughput target $\tau^*$, we compare the minimum cost rates achieved by static and adaptive policies. 

{A matching policy chooses whether to serve each incoming customer depending on their type. A policy is said to be static if it serves customers incurring costs below a certain fixed threshold, with potential random tie-breaking.\footnote{It is easy to see that with a single supplier type, our focus on threshold static policies is without loss of optimality.} Concretely, for any given throughput target, the optimal static policy selects \(k \in \{1, 2, 3\}\) and \(p \in [0, 1]\), serving every arriving type-\(i\) customer with probability 1 if \(i < k\), with probability \(p\) if \(i = k\), and not serving if \(i > k\). In contrast, adaptive policies can dynamically adjust the threshold based on the length of the queue; for example, when the number of agents waiting is large, the policy may choose to serve more customer types by lowering its threshold.}
    
{Consequently, the {\em adaptivity gap} of an instance ${\cal I}$ with throughput target $\tau^*$ is defined~as
    \[
        \textsf{Gap}_{\cal I}(\tau^*) = \frac{\textsf{static}(\tau^*)}{\textsf{adaptive}(\tau^*)} \ ,
    \]
where $\textsf{static}(\tau^*)$ and $\textsf{adaptive}(\tau^*)$ represent the minimum achievable costs under the respective policy classes, subject to ensuring a throughput rate of at least $\tau^*$.\footnote{{See \cite{dean2005adaptivity} for the definition of adaptivity gap in stochastic packing problems.}} For each of the 1000 randomly generated instances, we find the optimal static policy by enumerating all combinations of $k$ and $p$ introduced above. The optimal adaptive policy corresponds to the solution of a dynamic program, presented in \Cref{sec:dlp1}.}
   

Figure~\ref{fig:adaptivity} (a) visualizes the empirical distribution of adaptivity gaps across the 1000 instances, as we vary the throughput target $\tau^*$. Static policies incur 3.2\% more cost on average relative to the optimal adaptive policy, with the gaps larger than 5\% for about 25\% of the random instances and exceeding 40\% for the worst instance. 
    

In fact, by adjusting the parameters of our generative setting, we can find instances where the adaptivity gap can be as large as 208\%. Consider a market with $c_1 = c_2 = 0, c_3 = 1$, $\boldsymbol{\gamma} = \{2.4,2.4, 7.2\}$, $\tau^* = 3$. Figure~\ref{fig:adaptivity} (b) shows that the gap between static and adaptive policies depends on the relative level of abandonment $\mu$. In extreme cases where $\mu$ is small (e.g. $\mu < 0.75$) or large (e.g., $\mu \geq 2.5$), the adaptivity gap is negligible. Intuitively, when suppliers are very patient, a static policy mimicking fluid-optimal decisions is near-optimal: in this instance, both the adaptive and static policies can satisfy the throughput constraint by only serving type-1 and type-2 customers, with zero cost. When suppliers are highly impatient, no inventory builds up in the market, and thus, myopic (static) policies are also near-optimal. In contrast, interim values of $\mu$ allow adaptive policies to significantly outperform static policies. There exists $\mu^* \approx 0.76$ such that if $\mu = \mu^*$, both adaptive and static policy must always serve type 1 and 2 customers to have a throughput of $\tau^*$, with a cost rate of 0. As $\mu$ exceeds $\mu^*$, both policies begin serving type-3 customers as well, at which point the adaptive policy gains an advantage over static policies. In this case, the adaptivity gap exceeds 208\%.

    Our focus in the remainder of the paper is on designing adaptive approximation schemes for such matching queues.

\color{black}


\subsection{Preview of Our Main Results and Techniques} 
We develop efficient approximation schemes in two regimes of interest. In the first result, we focus on instances of the dynamic matching problem with a constant number of queues. In the following statements, $|{\cal I}|$ represents the size of the input. 

\begin{tcolorbox}[theorembox]
    \begin{restatable}{theorem}{constantptas}\label{thm:constant_ptas}
        There exists a Fully Polynomial-time Approximation Scheme (FPTAS) for the bi-criteria dynamic matching problem in networks with a constant number of queues, i.e., $n = O(1)$. Specifically, for any attainable cost-throughput target $(c^*,\tau^*)$, and for any accuracy level $\eps \in (0,1)$, our algorithm computes a $(1-\eps)$-approximate policy in time $\poly( (\frac{n}{\eps})^{n^2} \cdot  m^n \cdot |{\cal I}|)$. 
    \end{restatable}
\end{tcolorbox}

A corollary of Theorem~\ref{thm:constant_ptas} is an approximation scheme on Euclidean graphs in fixed dimension $d = O(1)$. Here, each customer and supplier has a location in $[0,1]^d$ drawn from known distribution and the matching cost is the Euclidean distance between them. The model is formalized in Section~\ref{sec:euclidean}.

\begin{tcolorbox}[theorembox]
\begin{restatable}{corollary}{euclidptas}\label{thm:ptas-line}
    There exists an efficient polynomial-time approximation scheme for the bi-criteria dynamic matching problem in $d$-dimensional Euclidean networks with $d = O(1)$. Specifically, for any attainable cost-throughput target $(c^*,\tau^*)$, and for any accuracy level $\eps \in (0,1)$, our algorithm computes a $(1-\eps)$-approximate policy in time  $\poly(g(\eps, d) \cdot |{\cal I}|)$ where $g(\eps,d) = (\frac{2d}{\eps})^{(\frac{2d}{\eps})^{4d+2}}$.\end{restatable}
\end{tcolorbox}

Proving Theorem~\ref{thm:constant_ptas} and Corollary~\ref{thm:ptas-line} requires a host of new technical ideas, combining online algorithm design and stability analysis for stochastic systems. The crux of our approach is a  linear programming framework that efficiently approximates near-optimal adaptive policies for both classes of inputs. As a starting point, we consider the case of a single queue ($n=1$). {Even in this setting, computing optimal policies is challenging and the best-known polytime algorithm is a constant-factor approximation \citep{kessel2022stationary}.}


\paragraph{\bf Adaptive policies for $\boldsymbol{n=1}$ via the  Dynamic LP.}  We develop an exact LP formulation, called the {\em Dynamic LP (DLP)}, initially focusing on the case $n=1$. This ``configuration LP'' represents the state-action occupancy measure of the system in coordinates $x^\ell_M$, where $\ell\in {\bb N}$ is the current number of suppliers in the queue and $M \subseteq {\cal C}$ is the (random) subset of customer types that can be served in that state. Linear constraints on the birth and death rates capture the queueing dynamics.

Although at first glance this infinite-dimensional LP is highly intractable, we devise our FPTAS via an efficient approximation of $(DLP)$. 
By studying a variant of the Bellman equation obtained from duality, we give an intuitive description of optimal adaptive policies. We establish a nested, threshold-concave property of the binding matching sets $M$ in the dual LP, reducing the number of subsets $M$. Bounding the queue lengths $\ell \in {\bb N}$ is more challenging, as the optimal policy's stationary distribution may be ``heavy-tailed''. Simple truncation or rounding ideas alone give a pseudo-polynomial running time. Instead, using the structural properties of the optimal dual solution, probabilistic coupling ideas, and alterations of the birth and death rates, we construct $(1-O(\eps))$-optimal policies whose queue lengths span only a polynomial range. 
Combining these observations in \Cref{sec:dlp1} we obtain an FPTAS for a single queue. 


\paragraph{\bf Key technical ideas: FPTAS for constant-size networks.} As multivariate birth-death processes are not time-reversible, our probabilistic analysis for bounding queue lengths does not extend to $n>1$. The network case requires a new synthesis of techniques, combining online LP-rounding ideas and Lyapunov drift analysis. At the core of our FPTAS, we devise a ``hybrid'' LP relaxation in \Cref{sec:constant_ptas}. Informally, we divide the network into ``thin'' and ``thick'' market types and use different algorithmic tools for each component, keeping the overall approach computationally tractable.

We define {\em short} queues as those whose probability of being depleted is $\Omega(\eps)$; the remaining queues are {\em long}. Intuitively, short queues are at risk of being depleted, so adaptive decisions matter. Long queues, by contrast, are rarely depleted, so static matching suffices. To capture this intuition, we devise the {\em Network  LP $(NLP)$}, which uses a lifted version of $(DLP)$ for the short queues and  static decision variables for the long queues. 

We introduce our policy \emph{Priority Rounding} for an online rounding of $(NLP)$. Upon each customer, we sample matching decisions based on the LP solution. Our hybrid LP decision variables may create ``contention'' between matching a supplier from a short queue and a supplier from a long queue. In case of contention, we always prioritize short queues, and the random match drawn for long queues is dropped. For some customer types, said to be {\em non-contentious}, this strict prioritization does not affect much the cumulative throughput. Nonetheless, deprioritizing long queues is costly for certain {\em contentious} customer types. To make up for the resulting loss, we create a {\em virtual buffer}, where we {\em schedule} each dropped match between a supplier from a long queue and a contentious customer, to fulfill it later. The virtual buffer is depleted whenever there is a surplus of unmatched customers. By using a drift analysis and bounding the bursts of virtual matches, we show that the virtual buffer is bounded in expectation, implying that our policy accurately tracks $(NLP)$'s match rates in the long run.

\paragraph{\bf Application to constant-dimensional Euclidean networks.} We devise in \Cref{sec:euclidean} a careful reduction from $d$-Euclidean networks with $d=O(1)$ to independently solved constant-size networks. We decompose the space of locations $[0,1]^d$ into multiple {\em cells} and approximate the dynamic matching problem locally within each cell using our previous FPTAS. A near-optimal allocation of the throughput target across cells is computed via a fractional min-{knapsack} problem, and the supplier and customer locations in each cell are clustered into a constant number of types. 

\paragraph{\bf Related literature.} 
The best-known guarantee for online stationary matching under reward maximization is a 0.656-approximation when $n=1$, obtained by \cite{kessel2022stationary} using a static threshold policy. We improve on this result by providing an FPTAS for networks with $n=O(1)$ queues, where we recall that the reward-maximization problem can be approximated by our bi-criteria optimization setting.\footnote{We note in passing that \citet{patel2024combinatorial} proposed a competitive algorithm achieving a ratio of $1 - 1/\sqrt{e}$ against the offline optimum. This bound was later improved by \citet{amanihamedani2024improved} to $1 - 1/\sqrt{e} + \eta$, for a universal constant $\eta$, which is currently the best-known competitive ratio.}


A central contribution of our work is to combine ideas from drift analysis and approximation algorithms. Recent literature has established the existence of low-regret algorithms for stationary matching problems {\em without abandonment} \citep{wei2023constant,kerimov2023optimality, kerimov2024dynamic,gupta2024greedy}. These results use stochastic network optimization techniques~\citep{neely2022stochastic}, bounding random deviations from the fluid LP relaxation, which we analogously employ in our analysis. However, these results crucially rely on a ``thick market'', where the queues can grow sufficiently large, in the absence of abandonment---this amounts to having only long queues in our terminology. Our algorithm achieves near-optimal performance with both short and long queues, as imposed by the abandonment patterns. {To do so, we  consider much tighter LP relaxations that represent short-term market dynamics.}

Near-optimal algorithms have previously been devised for classical online linear programming and resource allocation problems mainly under the assumptions of large budget and random arrivals, when the resources are available throughout the horizon and do not abandon \citep{devanur2009adwords, feldman2010online, devanur2011near, agrawal2014fast, kesselheim2014primal, gupta2016experts, arlotto2019uniformly}. Efficient approximation schemes are relatively less frequent in the online matching literature; e.g., \cite{segev2024near} develops a QPTAS for a related finite-horizon model with a single customer type and heterogeneous supplier types; \cite{anari2019nearly} devise a PTAS for Bayesian online selection with laminar constraints. FPTASes are known for the prophet secretary setting with a single-unit resource~\citep{dutting2023prophet}.\footnote{{For the prophet secretary problem, \cite{arnosti2023tight} establish tight constant-factor guarantees (against the offline optimal) for static threshold policies.} 
} For online stochastic matching on metric graphs, the best-known competitive ratio is $O((\log \log \log n)^2)$~\citep{gupta2019stochastic}. In spatial settings, \cite{kanoria2022dynamic} shows that a hierarchical-greedy algorithm is asymptotically optimal in large Euclidean networks with balanced demand/supply arrivals.

\subsection{Dynamic Matching Applications}\label{ssec:applications}



Our constant-size network result ({\bf \Cref{thm:constant_ptas}}) may be relevant to applications such as cadaveric organ allocation. In this setting, patients may wait several years for a transplant, whereas deceased donor organs must be allocated within hours. Accordingly, our model treats patients as suppliers and recovered organs as customers.

The feasibility of a transplant is primarily determined by blood type and tissue compatibility. Patients are often classified based on primary characteristics, such as blood and tissue types, and secondary factors, such as waiting time and age group. Similarly, organs can be categorized into corresponding types. Using this classification, incompatibility between a patient-organ pair $(i,j)$ can be quantified by a cost $c_{i,j}$, and thus a key objective is to maximize the number of transplants while minimizing incompatibility costs. This dual objective aligns with the cost-throughput bi-criteria framework of our model. Due to the limited number of biological markers for compatibility, our constant-size network model may offer an accurate yet tractable representation of the organ allocation process. Interestingly, our policy prescribes a simple classification and prioritization of the queues into ``short'' and ``long'' types (see \Cref{ssec:NLP}).

This classification recovers basic intuition about different patient types. For instance, type-AB patients, who arrive infrequently, are likely associated with short queues, whereas more common blood types (e.g. type O) are likely to form long queues. Our prescribed policy carefully tracks short queues (e.g., type-AB patients), leveraging real-time information for decision-making, while depleting long queues (e.g., type-O patients) using a static matching rule.\footnote{Although our model does not specify intra-class prioritization, giving priority to patients with longer waiting times would be important in practice. See \cite{agarwal2021equilibrium} for an in-depth study of deceased-donor organ allocation.} 

Beyond organ allocation, several other markets feature spatial frictions in addition to their dynamic characteristics. In ride-hailing systems, for instance, passengers (corresponding to customers) are relatively impatient while drivers (corresponding to suppliers) wait longer before matched.\footnote{Using the Uber Houston data analyzed by \cite{castillo2023benefits}, a rough estimate for a typical market suggests that passengers {are matched within a few minutes, while drivers experience an average per-ride waiting time of 27 minutes.} \citet[Fig. 14]{castillo2023benefits} indicates that ``full-time'' drivers work approximately 40 hours per week and complete an average of 40 trips during that time. Furthermore, Section 2 of the paper notes that drivers are \emph{idle} for 45.2\% of their working hours. By combining these points, {we derived the drivers' average waiting time per ride}.} In this case, the distance between a passenger and a driver is the most important factor that informs platforms' dispatch decisions. Hence, ride-hailing markets can be represented as a Euclidean network of dimension $d=2$, capturing spatial pick-up distances  \citep{kanoria2022dynamic}. Our Euclidean matching policy for $d = O(1)$ ({\bf \Cref{thm:ptas-line}}) may be applicable. This policy partitions the space into smaller cells to make \emph{local} matches---meaning that a passenger would be assigned to a driver nearby. Our policy avoid making matches over long distances, which aligns with the often-used {\em maximum dispatch radius} strategy in ride-hailing operations~\citep{wang2024demand}. 
Within each cell, however, our policy solves a dynamic program that accounts for future market dynamics \citep{ozkan2020dynamic,castillo2024matching} {and aims to fulfill a certain local throughput rate.} 


Lastly, our modeling frameworks can be applied to other dynamic matching systems with spatial features, such as food banks, blood banks, and emergency first-response. In such cases, we may leverage our result for $d$-dimensional Euclidean networks---with small values of $d$---to model physical locations as well as other relevant attributes (e.g., urgency, match compatibility, etc).

\subsection{Additional Notation and Terminology}
We use the shorthand ${\cal S} = [n]$ and ${\cal C} = [m]$. Since no policy can achieve a throughput above $\sum_{j \in {\cal C}} \gamma_j$, we denote by $\tau_{\max} = \sum_{j \in {\cal C}} \gamma_j$ an upper bound on the maximum achievable throughput. We also use the notation $c_{\max} = \max_{i,j} c_{i,j}.$ 

For any variable $a$, we let $\{a\}^+ = \max\{0, a\}$. We use ${\bb N}$ to refer to non-negative integers and ${\bb N}^+$ refers to positive integers. Throughout, bold variables refer to vectors, e.g., ${\boldsymbol{\ell}} = \left(\ell_1, \ldots, \ell_k \right)$ for some $k \geq 1$. For any $\ell \in {\bb N}^+$, the notation $[\ell]$ stands for the interval $\{1,\ldots, \ell\}$ and by extension, $[\ell]_0 = [\ell] \cup \{0\}$. {Moreover, for any subset of suppliers $S\subseteq {\cal S}$, we define a \textit{partial assignment} as a collection of disjoint subsets of customers assigned to suppliers in $S$. That is, a partial assignment $\bs{M} =\{M_i\}_{i\in S} \in \mathcal{D}(S)$ consists of subsets $M_i \subseteq {\cal C}$ such that $M_i \cap M_{i'} = \emptyset$ for all $i \neq i' \in S$.} Then, $j \in {\bs{M}}$ means that $j \in M_i$ for some $i \in {S}$. We use $M(j)=i$ to denote the unique $i\in S$ whose matching set covers $j$, i.e., $j\in M_i$ and $M(j) = \perp$ if $j$ is not covered, i.e., $j\notin \cup_{i\in S}M_i$.
Plus, for every set $M \subseteq [m]$, we use the shorthand $\gamma(M) = \sum_{j \in M} \gamma_j$ and $\lambda(M) = \sum_{i \in M} \lambda_i$. Lastly, if $LP$ is the name of a linear program, $LP^*$ denotes its optimal value.

\section{New Linear Programming Relaxation: The Case of a Single Queue} \label{sec:dlp1}


We develop a new LP relaxation for the dynamic matching problem that approximates the optimum with any desired accuracy, significantly tightening LPs studied in previous literature~\citep{collina2020dynamic,aouad2022dynamic,kessel2022stationary}. As a stepping stone for networks with $n=O(1)$, it is instructive to start with the case of a single queue (i.e., $n=1$) in this section. This setting reveals the essence of the new LP and provides intuition on the structure of optimal policies.  Since $n=1$, we omit the supplier index $i$ in our notation throughout this section. 




At a high level, we construct a polytope that represents all feasible stationary distributions, corresponding to adaptive policies that decide on the matching---whether or not to serve an incoming customer---based on how many suppliers are waiting. 


A notable difficulty is that the stationary distributions of birth-death processes, induced by stationary policies, are nonlinear in the input parameters and state-dependent matching decisions. This motivates us to use a configuration LP, which ``linearizes'' the stationary distribution by considering {\em all} states and actions. Specifically, we introduce an LP decision variable  $x_{M}^\ell$ for each subset of customers $M \subseteq [m]$ and queue length $\ell \in {\bb N}$. Next, for each finite $\bar{\ell} \in {\bb N}$, we define ${\cal B}(\bar{\ell})$ as the set of $\boldsymbol{x}$ vectors that meet the conditions:
\begin{align}
    \lambda  \sum_{M \subseteq [m]} x_{M}^{{\ell} -1} &= \sum_{M \subseteq [m]} x_{ M}^{\ell}  (\gamma(M) + \ell) \ , && 0 \leq \ell \leq \bar{\ell} \label{eq:single_flow_balance} \\ 
 \sum_{M \subseteq [m], \ell \in \mathbb{N}} x_{M}^\ell &= 1 \ ,  \label{eq:total_prob} \\
    x_{M}^\ell &\geq 0 \ . && \forall M \subseteq [m], 0 \leq \ell \leq \bar{\ell} \notag
\end{align} 
Interpret $x_{M}^\ell$ as the (unconditional) steady-state probability that the number of suppliers in the queue is equal to $\ell$ and the policy currently commits to serving only customers from types $M$. We say that $M$ is the policy's current {\em matching set}; intuitively, any stationary policy (deterministic or randomized) is fully characterized by the distribution of matching sets it commits to, before seeing the next arrival, in each state. {For example, in the context of spatial markets, we can view $M$ as a supplier's ``catchment area'', {drawn each time a new customer arrives}.}


Naturally, when $\ell = 0$, the only feasible matching set is $M=\emptyset$. When $\ell = \bar{\ell}$,  it is implicit in ${\cal B}(\bar{\ell})$ that the queue cannot grow larger. That is, policies must discard new arriving suppliers if state $\bar{\ell}$ is reached. We call such policies {\em $\bar{\ell}$-bounded}.

With the restriction to $\bar{\ell}$-bounded policies, the dynamic matching problem admits the following new LP formulation, which we refer to as the {\em Dynamic LP}: 
\begin{align}
(DLP(\bar{\ell})) \quad &&\min_{\boldsymbol{x}} &&&\sum_{\ell \leq \bar{\ell}} \sum_{M \subseteq [m], j \in M}  \gamma_j  c_{j}  x_{M}^{\ell} \notag \\
&&\text{s.t.}  &&& \left(x_{M}^\ell\right)_{M, \ell}
\in {\cal B}(\bar{\ell}) \ ,  \label{eq:single_dlp_flow_balance} \\ 
&&  &&&\sum_{1 \leq \ell \leq \bar{\ell}}\sum_{M \subseteq [m]} \gamma(M)  x_{M}^\ell \geq \tau^* \ . &&& \label{eq:single_dlp_throughput}   
\end{align} 
The cost expression of $DLP(\bar{\ell})$ follows from the intuitive PASTA property \citep{wolff1982poisson}. It states that, under our interpretation of the decision variables $ x_{M}^\ell$-s, the match rate between suppliers and type-$j$ customers in state $\ell$ must be exactly $\sum_{M \subseteq [m], j \in M} \gamma_j  x_{M}^\ell$.
Cost minimization is counterbalanced by constraint~\eqref{eq:single_dlp_throughput}, which guarantees that the throughput target $\tau^*$ is met. 

In general, policies do not place any cap on queue lengths, which may be unbounded. By extending our LP formulations to any $\ell \in {\bb N}$, we obtain the infinite-dimensional polytope ${\cal B}(+\infty)$ and the resulting dynamic linear program $DLP(+\infty)$, whose optimum is $DLP^*$. The next lemma shows that $DLP(+\infty)$ provides an exact LP formulation of the dynamic matching problem.

\begin{lemma}\label{prop:dlp1_benchmark}
For every policy $\pi$ for the single-queue instance, $\tau(\pi) \geq \tau^*$ implies $DLP^* \leq c({\pi})$. Moreover, there exists a policy $\pi^*$ such that $\tau(\pi^*) = \tau^*$ and $c(\pi^*) = DLP^*$.
\end{lemma}
 The proof of \Cref{prop:dlp1_benchmark} appears in Appendix \ref{app:dlp_proof}, and formalizes the preceding discussion.
Any policy $\pi$ admits, at each time $t \in [0, \infty)$, a subset  $M^\pi(t) \subseteq [m]$ of customer types that $\pi$ commits to serve upon arrival. Denoting by ${\cal L}^\pi(t)$ the queue length at time $t$, the distribution of matching sets $(x^\ell_M)_M$ corresponds to the time-average frequency of $M^\pi(t)$ in state ${\cal L}^\pi(t)=\ell$. Since the queue is replenished at rate $\lambda$ and depleted at rate $\mathcal{L}^\pi(t) + \gamma(M^\pi(t))$, the stationary distribution satisfies our LP constraints. 



With a slight abuse of notation, we henceforth use $DLP$ and $DLP(+\infty)$ interchangeably. Although $DLP$ is {an exact formulation}, it is not clear how one can solve this LP due to the large number of variables: the infinite dimension $\ell \in {\bb N}$ and the exponential number of matching sets $M \subseteq [m]$. We establish that $DLP$ can be efficiently approximated in polynomial time. 
\begin{restatable}{proposition}{singlefptas}\label{prop:single_fptas}
    There exists an FPTAS for the bi-criteria dynamic matching problem in single-queue instances (i.e. $n=1$). Specifically, for any attainable cost-throughput target $(c^*,\tau^*) \in (0,c_{\max})\times (0,\tau_{\max})$, and for any accuracy level $\eps \in (0,1)$, our algorithm computes, in time polynomial in the input size and $\frac{1}{\eps}$, a $(1+\eps)$-approximate policy $\pi$  such that $c(\pi) \leq c^*$ and $\tau(\pi) \geq (1-\eps) \tau^*$.
\end{restatable}
What happens if a cost-throughput target $(c^*, \tau^*)$ is not attainable? Our algorithm either outputs a policy $\pi$ that satisfies $\tau(\pi) \geq (1-\eps) \tau^*$ and an estimate of its corresponding cost rate $c(\pi)$, or it does not return any feasible policy. In both cases, we either find a suitable policy $(1+\eps)$-close to the Pareto frontier, or we have a certificate that $(c^*, \tau^*)$ is not achievable. 

The remainder of this section establishes Proposition~\ref{prop:single_fptas}, which follows from approximating $DLP(+\infty)$ and then implementing its solution as a policy. Solving $DLP(+\infty)$ has two challenges: the unbounded queue length and the exponential number of matching sets. We show in Section~\ref{ssec:ptas-glp-proof} that the dimensionality can be reduced to polynomially-sized LPs with a small loss. Beforehand, we take an excursion into duality and develop intuition on the optimal policy's structure.

 
\subsection{Dual Formulation and Properties of Optimal Policies} \label{ssec:dual}

Because $DLP(+\infty)$ is infinite-dimensional, writing its dual requires some care. 
First, we relax the equality in equation \eqref{eq:single_flow_balance} and substitute instead the inequality
\begin{align}
    \lambda  \sum_{M \subseteq [m]} x_{M}^{{\ell} -1} \geq \sum_{M \subseteq [m]} x_{M}^{\ell}  (\gamma(M) + \ell) \ . && 0 \leq \ell \leq \bar{\ell} \label{ineq:flow_balance_prime} \tag{\getrefnumber{eq:single_flow_balance}'}
\end{align}
 \Cref{clm:dlp_relaxation} in Appendix \ref{app:dlp_relaxation} shows that this relaxation {is without loss and leaves the optimum unchanged}. From this equivalent primal LP, we use standard rules to derive the dual:
\begin{align}
    \max_{{\alpha}, \boldsymbol{\delta}, \theta} \quad & \alpha  + \theta   \tau^* &&   \notag \\
    \text{s.t.} \quad 
    & -\lambda  \delta^{\ell+1} + 
    \Bigl(\ell + \sum_{j \in M} \gamma_j\Bigr)  \delta^{\ell} + \alpha  \leq \sum_{j \in M} \gamma_j   \left(c_{j} - \theta\right) \ , && \forall M \subseteq [m], \forall \ell \in \mathbb{N}^+ \label{ineq:single_dual_bellman} \\ & -\lambda  \delta^1 + \alpha \leq 0 \ , \label{ineq:single_dual_qone} \\
    & \delta^\ell \leq 0 \ , && \ell \in {\bb N}^+  \label{ineq:single_delta_negativity}  \\ 
    & \theta \geq 0 \ .  \notag
\end{align}
Here, $\alpha$ is the dual variable to the constraint \eqref{eq:total_prob}, $\delta^{\ell}$ is associated with detailed balance \eqref{eq:single_dlp_flow_balance}, and $\theta$ corresponds to the throughput target constraint \eqref{eq:single_dlp_throughput}.  It is easy to see that any optimal solution of this dual LP satisfies $\delta^\ell \to 0$ as $\ell \to \infty$. Thus, the so-called {\em transversality condition} holds, which in turn guarantees strong duality and complementary slackness \citep{romeijn1992duality}.\footnote{Informally, the transversality condition says that the dual prices are asymptotically zero {as $\ell \rightarrow +\infty$}; see also \cite{romeijn1998shadow}.} Note that, by contrast, strong duality does not hold for the ``natural dual'' of $DLP(+\infty)$.

We can interpret the dual variable $\theta$ as a discount on costs, yielding a modified instance where $c_{j}$ is replaced with the reduced cost $c'_{j} = c_{j} - \theta$. The goal in this modified instance is still to minimize the cost rate $\sum_{\ell \in \mathbb{N}} \sum_{M \subseteq [m], j \in M} \gamma_j c'_{j} x_{M}^\ell$. {While there is no constraint on throughput,} $c'_{j}$ may be negative, meaning it is now profitable to match such (negative reduced-cost) customer types.

On closer inspection, the dual constraints \eqref{ineq:single_dual_bellman} constitute a {time-differenced} version of the Bellman equations for the average cost dynamic program in the modified instance. That is, $\alpha$ is the optimal average cost and $\delta^\ell$ is related to the notion of {\em bias}---the difference between the value function in the current state and the time-average cost (see Appendix~\ref{app:bellmans} showing that $\delta^\ell$ is the difference of bias for consecutive states). Having this dynamic programming formulation in mind, the optimal policy for this modified instance is a state-dependent threshold policy, where $\delta^\ell$ is the maximum admissible cost in each state $\ell$. The following lemma formalizes this statement (proof in Appendix~\ref{app:dual-threshold-increasing}).


\begin{lemma}\label{lem:dual-threshold-increasing}
     Consider a pair of relaxed primal $DLP(+\infty)$ and its dual optimal solutions $\boldsymbol{x}, {\alpha}, \boldsymbol{\delta}$, $\theta$. Then, ${x^\ell_{M}} > 0$ implies that $\hat{M}^\ell \subseteq M \subseteq \hat{M}^\ell \cup \tilde{M}^\ell$,
    where we define \begin{align}\label{eq:dual-matching-sets}
        \hat{M}^\ell = \left\{j \in [m]: c_{j} - \theta < {\delta}^{\ell} \right\} \quad \text{and} \quad \tilde{M}^\ell = \left\{j \in [m]: c_{j}  - \theta = {\delta}^{\ell} \right\} \ . \notag   \end{align} Moreover, $\boldsymbol{\delta}$ is concave and increasing, i.e., for every $\ell \geq 1$, $0 \leq \delta^{\ell+2} - \delta^{\ell+1} \leq \delta^{\ell+1} - \delta^\ell$.
\end{lemma}
As may be expected, Lemma~\ref{lem:dual-threshold-increasing} shows that the optimal thresholds are roughly non-decreasing in the queue lengths. The more suppliers we have accumulated in the queue, the more willing we are to forego larger matching costs. 
As for the concavity, it is intuitive that the marginal value from adding one supplier to the queue, in the modified instance, decreases as we consider larger queues.

A direct algorithmic implication of Lemma~\ref{lem:dual-threshold-increasing} is that we need not consider all subsets $M \subseteq [m]$ in the primal, but only focus on the family of $m$ nested subsets $M_1, \cdots, M_m$ without loss, where $M^u = \{j\in [m]: c_j\leq c_u\}$ (assuming $c_1 < \cdots < c_m$).\footnote{{Considering strict inequalities is without loss of generality since the Poisson superposition property implies that we can merge customer types with the same costs.}} Nonetheless, the primal and dual LPs are still intractable, as the number of decision variables, corresponding to varying queue lengths, is unbounded. By exploiting the structure of the stationary distribution, we devise our efficient approximation scheme in the next section.


\subsection{Proof of \Cref{prop:single_fptas}: Bounded Policies}\label{ssec:ptas-glp-proof}



Is it possible to approximate any policy $\pi$, up to a factor of $(1+\eps)$, with a corresponding $\bar{\ell}$-bounded policy? Recall that an $\bar{\ell}$-bounded policy discards new suppliers who arrive when the queue length is already at $\bar{\ell}$. The answer is obviously `yes' for a sufficiently large value of $\bar{\ell}$. The tail distribution of the birth-death queuing process decays geometrically in $\ell$ from $\ell\geq \bar{\ell}= \Theta({\lambda})$, as the death rate then exceeds the birth rate. However, this naive bound only gives a {\em pseudo-polynomial} dependency in the input ($\lambda$). In fact, there may be instances in which the optimal policy's stationary distribution is ``heavy'' tailed (i.e., super-polynomial queue lengths occur with a constant probability). To derive a truly polynomial bound, we restrict attention to {\em monotone} policies. 
\begin{definition}[Monotone policies] \label{def:monotone}
    A stationary policy $\pi$ is monotone if, for every $j \in [m]$, the type-$j$ conditional match rate $\gamma_j^{(\ell)} = \sum_{M: j \in M} \gamma_j \cdot \prpar{M^\pi_\ell = M | {\cal L}^\pi = \ell}$ is non-decreasing in $\ell\in {\bb N}$. 
\end{definition}

Monotone policies are natural: as the queue length increases, we expect to be less picky in serving customers, and thus, the conditional match rate of each customer type must be non-decreasing. Lemma~\ref{lem:dual-threshold-increasing} already showed that the optimal policy is monotone, as the threshold $\delta^\ell$ increases in $\ell$.

The crux of the proof of Proposition~\ref{prop:single_fptas} resides in showing that monotone policies can be uniformly approximated by polynomially bounded ones.

\begin{lemma}\label{lem:truncate_distrib}
For every monotone policy $\pi$ and $\eps \in (0,1)$, there exists a $K$-bounded randomized policy $\tilde{\pi}$ such that $\tau_j({\tilde{\pi}}) \leq \tau_j({\pi})$ for all $j \in [m]$, $\tau({\tilde{\pi}}) \geq (1-\eps)  \tau({\pi})$, and  $K = {O}( \frac{1}{\eps} \cdot (\log \frac{\tau_{\max}}{\tau(\pi)} + \log \frac{1}{\eps}))$. 
\end{lemma}
Comparing $K = \tilde{O}(\frac{1}{\eps}\log(\frac{\tau_{\max}}{\tau(\pi)}))$ in Lemma~\ref{lem:truncate_distrib} to the naive bound $\bar{\ell} = \Theta(\lambda)$, we achieve an exponential-order improvement. The proof involves lengthy technical details, and thus, we defer it to Appendix \ref{app:state-space-collapse}. We combine various probabilistic couplings to devise an approximate policy with a ``light tailed'' stationary distribution. Key to our construction is the fact that monotone policies induce a unimodal stationary distribution for the queue length, and thus, we ``compress'' the distribution into a logarithmic span of queue lengths by altering the birth and death rates.

Combining Lemma~\ref{lem:dual-threshold-increasing} and Lemma~\ref{lem:truncate_distrib} completes the proof of \Cref{prop:single_fptas}. Suppose that we are given a target $(c^*,\tau^*)$  attainable by a reference policy $\pi^{\rm ref}$. We solve the primal version of $DLP(K)$ with the relaxed throughput target $(1-\eps)\tau^*$, including only decision variables $x^\ell_{M}$ for the matching sets $M\in \{M^u\}_{u\in [1,m]}$ in the nested family $M^u= \{j\in [m]: c_j \leq c_u\}$ and $M^0 = \emptyset$. This restriction is without loss by Lemma~\ref{lem:dual-threshold-increasing}. Our upper bound on $K$ immediately implies polynomial running time. Hence, the primal solution describes a $K$-bounded randomized policy $\hat{\pi}$ that achieves 
\[ 
c(\hat{\pi})\leq c(\tilde{\pi}) = \sum_{j\in m} c_j \tau_j(\tilde{\pi}) \leq \sum_{j\in m} c_j \tau_j({\pi}^{\rm ref}) \leq c^* \ ,
\]
where $\tilde{\pi}$ is the $K$-bounded policy in Lemma~\ref{lem:truncate_distrib} with respect to $\pi = \pi^{\rm ref}$. The first inequality follows from the optimality of $\hat{\pi}$ in $DLP(K)$, and the second inequality follows from Lemma~\ref{lem:truncate_distrib}. At the same time, we guarantee a throughput rate $\tau(\hat{\pi}) \geq (1-\eps) \tau^*$ from constraint~\eqref{eq:single_dlp_throughput}.

\section{FPTAS for a Constant Number of Queues}\label{sec:constant_ptas}

We turn our attention to networks with a constant number of supplier types.{ A natural generalization of $(DLP)$ introduces decision variables of the form $x^\ell_{i,M}$ for each type $i\in {\cal S}$, subject to $\bs{x}_{i}\in {\cal B}_i(\bar{\ell})$ and other capacity constraints (see Appendix \ref{subsec:tentative} for a complete description). While this approach gives a valid  LP relaxation, it can be shown that the resulting LP has a constant-factor gap with the online optimum for the family of instances in \citet[App.~C]{kessel2022stationary}.} Extending our FPTAS from a single queue to networks with $n>1$, requires a host of new techniques beyond the Dynamic LP, which we develop in this section. We restate our theorem below. 


\constantptas* 

The FPTAS develops from an important insight: ``thin'' and ``thick'' markets operate on different timescales. Specifically, in Section~\ref{ssec:NLP}, we introduce a distinction between  {\em short}  and {\em long} queues, depending on their probability of being empty, under an assumption on their depletion rates. This distinction is leveraged to develop a  hybrid LP relaxation and a priority-based LP-rounding policy. {Our Network LP employs state-adapted decision variables to model short queues more precisely,  while simplifying the treatment of long queues using static decision variables to achieve a polynomial running time.} Consequently, we introduce a matching policy, termed {\em Priority Rounding}, which performs an online rounding of the LP solution, at different timescales, with $O(\epsilon)$ loss. Priority rounding is designed in Section \ref{ssec:pm_policy} and analyzed in Section \ref{ssec:pm_analysis}, completing the proof of \Cref{thm:constant_ptas}. Throughout this section, we assume that the cost-throughput target $(c^*, \tau^*)$ is attainable. 


\subsection{Network LP}\label{ssec:NLP}

\paragraph{Short versus long queues.} If a queue has a high arrival intensity relative to its abandonment rate, and its service ``load'' is not excessive, it is unlikely to be depleted. This is because the queue evolves approximately as a symmetric random walk when the queue length is small, {as} formalized next.\footnote{This claim follows from the drift method; see \Cref{clm:exp_queue_ub} in Appendix \ref{app:exp_queue_ub}.}
\begin{claim} \label{obs:depleted}
Consider a birth-death process \( L_i(\cdot) \) with birth rate \( \lambda_i \) and death rate \( \ell + \lambda_i \) at state \( \ell \). If \( \lambda_i \geq 1/\varepsilon \), then the steady-state probability of being in state 0 satisfies \[ \pr{L_i(t) = 0}\leq \sqrt{\varepsilon} \ . \]
\end{claim}

This motivates us to distinguish between short and long queues.\footnote{{Claim~\ref{obs:depleted} serves only to motivate the definition of short and long queues; our analysis will require a stronger property.}}  For carefully chosen cutoff values $0 \leq \underline{\lambda} \leq {\bar{\lambda}}$, we define  ${\cal S}^{\rm short} = \{ i \in [n]: \lambda_i \leq \underline{\lambda} \}$ and ${\cal S}^{\rm long} = \{ i \in [n]: \lambda_i \geq \bar{\lambda} \}$. Ideally, we would like the two sets to form a partition of all supplier types. Let $\delta = \frac{\eps^2}{n}$. {By the pigeonhole principle, there exists $\kappa~\in~\{1,\cdots, \min \{1/\eps, n\} + 1 \}$ such that if we exclude every supplier type with arrival rate in $\left(\frac{1}{\delta^\kappa}, \frac{1}{\delta^{\kappa + 1}}\right)$, there exists a policy $\pi$ with $\tau(\pi) \geq (1 - \epsilon)\tau^*$ and $c(\pi) \leq c^*$. While the exact value of $\kappa$ may not be known a priori, it can be found through guessing and enumeration.} 

Consequently, we set $\underline{\lambda} = \frac{1}{\delta^\kappa}$, $\bar{\lambda} = \frac{1}{\delta^{\kappa+1}}${ and, by a slight abuse of notation, we refer to the remaining supplier types as ${\cal S} = {\cal S}^{\rm short} \sqcup {\cal S}^{\rm long}$, after removing supplier types that have arrival rate in $(\frac{1}{\delta^\kappa}, \frac{1}{\delta^{\kappa+1}})$. Furthermore to avoid carrying around a factor $(1+\eps)$, we assume that $(c^*, \tau^*)$ is still attainable after the removal of these supplier types.} Henceforth, a short supplier refers to a type $i \in {\cal S}^{\rm short}$, and long suppliers are defined analogously. 

It turns out that these two types of suppliers require different algorithmic tools. There is less concern from abandonment for long queues. Thus, it may be less important to track the precise state of the queues. Short queues, by contrast, require a more careful allocation because such suppliers may be scarce resources, due to the higher risk of depletion.

\paragraph{A hybrid linear programming approach.} 
For each type $i\in {\cal S}^{\rm long}$,  we introduce static decision variables $y_{i,j}$ that capture the expected conditional probability of matching a supplier of type $i$ given that a customer of type $j \in {\cal C}$ arrives. The contribution to the cost rate is simply $\sum_{i\in {\cal S}^{\rm long}} \gamma_j c_{i,j} y_{i,j}$ and that to the throughput rate is $\sum_{i\in {\cal S}^{\rm long}} \gamma_j y_{i,j}$. Matching constraints on the supplier side ensure that the match rates respect the arrival rate of each type, i.e., $\sum_{j\in {\cal C}}\gamma_j y_{i,j}\leq \lambda_i$ for each $i\in {\cal S}^{\rm long}$.

Now, for each type $i\in {\cal S}^{\rm short}$, we  utilize state-dependent decision variables in the spirit of $(DLP)$ in Section~\ref{sec:dlp1}. However, as noted previously, a straightforward extension of $(DLP)$ was too lossy to achieve near-optimal performance guarantees (see Appendix~\ref{subsec:tentative}).


To tighten that LP, we start by formulating a polytope that exactly represents all feasible multivariate birth-death distributions, induced by matching policies. Let us consider $\bar{\ell}$-bounded policies for some $\bar{\ell}\geq 0$ and a subset of queues $S\subseteq {\cal S}$. The state of the system is a vector $\bs{\ell} \in [\bar{\ell}]_0^S$ of queue lengths. Recall that ${\bs{M}}\in {\cal D}(S)$ refers to a (partial) assignment between the supplier types in $S$ and customer types---i.e., meaning that the subsets of customer types $\{M_i:i\in S\}$ are disjoint. Any feasible distribution $\bs{x} = (x^{\bs{\ell}}_{{\bs{M}}})_{\bs{\ell}\in [\bar{l}]_0,{\bs{M}}\in {\cal D}(S)}$ over states and assignments must satisfy: 
\begin{align}
     ({\cal B}(S, \bar{\ell})) \hspace{0.5cm}  & \sum_{{\bs{M}} \in {\cal D}(S)} \left( \sum_{\substack{i \in S:\\ \ell_i \geq 1}} \lambda_i   x_{\bs{M}}^{{\boldsymbol{\ell}} -e_i} + \sum_{i \in S} (\gamma(M_i) +  (\ell_i+1))  x_{\bs{M}}^{{\boldsymbol{\ell}} + e_i} \right) \notag \\ & \quad = \sum_{{\bs{M}} \in {\cal D}(S)} x_{\bs{M}}^{\boldsymbol{\ell}}  \left( \sum_{i \in S} \sum_{j \in M_i} \gamma_j \cdot \mathbb{I}[\ell_i \geq 1] + \lVert{\boldsymbol{\ell}}\rVert_1 + \sum_{i \in S} \lambda_i \cdot {\bb I}[\ell_i < \bar{\ell}] \right) \ , && \forall  {\boldsymbol{\ell}} \in [\bar{\ell}]_0^{|S|} \label{eq:multivariate_flow_balance} \ .
\end{align}
Here, $x_{\bs{M}}^{\boldsymbol{\ell}}$ is the stationary probability that, for every $i \in S$, the policy matches an arriving customer from set $M_i$ to a supplier of type $i$ when the state of the system (i.e., the joint queue lengths vector) is exactly~$\boldsymbol{\ell}$.   Constraints~\eqref{eq:multivariate_flow_balance} represent flow balance system for a {multivariate} birth-death process, where the left-hand side is the inflow rate to state $\bs{\ell}$ and the right-hand side is the outflow rate from state $\bs{\ell}$. Consequently, we denote by ${\cal B}(S,\bar{\ell})$ the polytope of distributions $x_{\bs{M}}^{\boldsymbol{\ell}}$ that satisfy constraints~\eqref{eq:multivariate_flow_balance}.

The  reader may notice that, when $S = \{i\}$ is a single queue, ${\cal B}(S,\bar{\ell})$ does not exactly recover ${\cal B}(\bar{\ell})$. Constraints \eqref{eq:multivariate_flow_balance} are the {\em global balance equations}, whereas constraints \eqref{eq:single_dlp_flow_balance} express a simpler condition, known as {\em detailed balance}---the transition rate from state $\ell - 1$ to $\ell$ is equal to that from $\ell$ to $\ell-1$. There is a key difference between a single queue and a multivariate birth-death process; the former has a reversible steady-state distribution that satisfies detail balance.\footnote{We refer the interested reader to Kolmogorov's critera, which provide a necessary and sufficient condition for reversibility of a Markov chain; see \citet[Sec.~1.5]{kelly2011reversibility}.}

By piecing together these two approaches, we formulate our {\em Network LP} $(NLP(\bar{\ell}))$. We model short queues through the polytope ${\cal B}({\cal S}^{\rm short},\bar{\ell})$ and long queues using the static variables $(y_{i,j})_{i\in {\cal S}^{\rm long},j\in {\cal C}}$, thereby obtaining
\begin{align}
 &&(NLP(\bar{\ell})) \qquad \min_{\boldsymbol{x},\boldsymbol{y}\geq 0}\qquad  &\sum_{\MyAtop{i \in {\cal S}^{\rm short}}{{\bs{M}} \in {\cal D}({\cal S}^{\rm short})}}  \sum_{\substack{j \in M_i\\{\boldsymbol{\ell}: \ell_i \geq 1}}}\gamma_j  c_{i, j}  x_{\bs{M}}^{\boldsymbol{\ell}} + \sum_{i\in {\cal S}^{\rm long}} \sum_{j \in [m]} \gamma_j  c_{i,j}  y_{i,j}   \nonumber \\
&& \text{s.t.}  \qquad & \ \ \ \bs{x}\in {\cal B}({\cal S}^{\rm short},\bar{\ell}) \ ,  \label{eqn:NLP_polytope} \\
&&  &\quad \ \sum_{j \in {\cal C}} \gamma_j y_{i,j} \leq \lambda_i \ , \hspace{5cm} \forall i\in {\cal S}^{\rm long} \label{ineq:abundant_capacity}  \\ 
&&  & \sum_{\MyAtop{i \in {\cal S}^{\rm short}}{{\bs{M}} \in {\cal D}({\cal S}^{\rm short})}}  \sum_{\substack{j \in M_i\\{\boldsymbol{\ell}: \ell_i \geq 1}}} x_{\bs{M}}^{\boldsymbol{\ell}} + \sum_{i\in {\cal S}^{\rm long}} y_{i,j} \leq 1\ , \hspace{1.9cm}  \forall j \in [m] \label{ineq:NLP-contention-constraint} \\
&&  &\sum_{\MyAtop{i \in {\cal S}^{\rm short}}{{\bs{M}} \in {\cal D}({\cal S}^{\rm short})}}  \sum_{\substack{j \in M_i\\{\boldsymbol{\ell}: \ell_i \geq 1}}}\gamma(M_i) \cdot x_{\bs{M}}^{\boldsymbol{\ell}} + \sum_{i \in {\cal S}^{\rm long}} \sum_{j \in {\cal C}}\gamma_{j}  y_{i, j} \geq (1-\eps) \tau^* \ .  \label{ineq:NLP_target_attenuated}  
\end{align}
 Constraints \eqref{ineq:NLP-contention-constraint} couple the two types of decision variables, state-dependent and static, to ensure that their combination does not exceed the customer capacity. 
 Constraint \eqref{ineq:NLP_target_attenuated} ensures that we achieve a large enough throughput rate $(1-\eps)\tau^*$. Note that $NLP(\bar{\ell})$ is neither a tighter formulation than $\widetilde{DLP}$ in Appendix~\ref{subsec:tentative}, nor vice versa. This is because $NLP(\bar{\ell})$ is a hybrid between a dynamic programming formulation for short queues and a static LP approximation for long queues. 


We establish that $NLP(\bar{\ell})$ is a valid relaxation of the dynamic matching problem over policies that do not let short queues grow longer than $\bar{\ell}$. Moreover, we can restrict attention to a polynomially bounded value of $\bar{\ell}$ with only a small loss of throughput.
The proof is in Appendix \ref{prf:NLP_feasible}.

\begin{lemma}\label{lem:NLP_feasible}
    For each $\eps \in (0,1)$, there exists $\bar{\ell} \in \poly(\frac{ n^{\kappa}}{\eps^{\kappa}})$ such that for every policy $\pi$ with $\tau(\pi) \geq \tau^*$, we have $NLP(\bar{\ell})^* \leq c(\pi)$. Moreover, $(NLP(\bar{\ell}))$ is solvable in time $\poly((\frac{n^\kappa}{\eps^\kappa})^n \cdot m^n  \cdot |{\cal I}|)$, where $|{\cal I}|$ is size of the input. 
 \end{lemma}

Now, the crucial question is how to convert a solution of $NLP(\bar{\ell})$ into a matching policy with near-optimal performance guarantees.

\subsection{Priority Rounding}\label{ssec:pm_policy}
Here, we introduce our matching policy, $\pi^{\rm pr}$, based on a randomized LP-rounding (\Cref{alg:priority_matching}). 
Our algorithm takes as input a 
 solution $\{x_{\bs{M}}^{\boldsymbol{\ell}}\}_{{\bs{M}}, \boldsymbol{\ell}}, \{y_{i,j}\}_{i,j}$ of $(NLP(\bar{\ell}))$ with $\bar{\ell}$ specified as in \Cref{lem:NLP_feasible}. {We assume that this solution is non-degenerate, meaning that  $x_{\boldsymbol{M}}^{\boldsymbol{\ell}} > 0$ for $j\in M_i$ implies $\ell_i\geq1$, which can be easily enforced.} Upon the arrival of each new customer, $\pi^{\rm pr}$  samples matching decisions based on the LP solution. If there is ``contention'' between matching a supplier from a short queue and a supplier from a long queue, it is resolved by always prioritizing the former. {In certain cases, de-prioritized matches with long queues are {\em scheduled} on a {\em virtual buffer} to be fulfilled later.} 


\paragraph{Contention and priority rule.} Suppose that a customer $\mathfrak{c}$ of type $j$ arrives. 
A natural LP rounding strategy is to draw an assignment ${\bs{M}}^{\rm short}$ according to the distribution  $\frac{x^{\bs{\ell}}_{{\bs{M}}}}{\sum_{{\bs{M}}' \in {\cal D}({\bs S}^{\rm short})} x^{\bs{\ell}}_{{\bs{M}}'}}$ over ${\bs{M}}\in {\cal D}({\cal S}^{\rm short})$, and then, match $\mathfrak{c}$ with a supplier of type $i^{\rm short} = {\bs{M}}^{\rm short}(j)$ if $j$ is covered by ${\bs{M}}^{\rm short}$ or otherwise, do not match $\mathfrak{c}$ if ${\bs{M}}^{\rm short} (j) = \perp$.  Here, we choose $\bs{\ell} \in [\bar{\ell}]_0^{n_s}$ to be the current state of the short queues, where $n_s= |{\cal S}^{\rm short}|$. 
The challenge, however, is that $NLP(\bar{\ell})$ also ``promises'' to match $\mathfrak{c}$ with the long queues $i\in {\cal S}^{\rm long}$ at rate $\gamma_j y_{i,j}$. It is unclear at face value how to simultaneously fulfil both match rates. Our LP  only guarantees that the combined match rates do not exceed the capacity $\gamma_j$ via  constraint~\eqref{ineq:NLP-contention-constraint}. 

To resolve the contention between short and long queues, we distinguish between two types of customer types ${\cal C} = {\cal C}^{\rm ct} \sqcup \overline{{\cal C}^{\rm ct}}$. We classify $j\in {\cal C}^{\rm ct}$ as a {\em contentious} customer type if
\begin{eqnarray} \label{ineq:contentious}
    \sum_{\MyAtop{{\bs{M}} \in {\cal D}({\cal S}^{\rm short}):}{j\in {\bs{M}}}} \; \sum_{{{\boldsymbol{\ell}} \in [\bar{\ell}]_0^{n_s}}} x_{\bs{M}}^{\boldsymbol{\ell}} \geq \eps \ .
\end{eqnarray}
Intuitively, these are customer types $j \in {\cal C}^{\rm ct}$ with a significant matching proportion from short queues.  Our policy $\pi^{\rm pr}$ will prioritize matching these customer types $j \in {\cal C}^{\rm ct}$  with short queues, and postpone their matches with long queues to a later stage, when there is a surplus of unmatched type-$j$ customers. If  condition~\eqref{ineq:contentious} is reversed, then we say that $j\in \overline{{\cal C}^{\rm ct}}$ is {\em not} a contentious customer type. In this case, the matches with suppliers from short queues are marginal, so we still prioritize short queues but simply ignore the contention with long queues.


\paragraph{Matching policy.}  \Cref{fig:pr} visualizes the Priority Rounding policy. Upon the arrival of each new customer, $\pi^{\rm pr}$ has three phases: {\em sampling}, {\em matching \& scheduling}, and {\em surplus matching}. 
In the first phase, we draw a random assignment ${\bs{M}}^{\rm short}$ of customers, as per Line \ref{lin:s_def}, to determine $i^{\rm short} = {\bs{M}}^{\rm short}(j)$. Simultaneously, we draw on Line~\ref{lin:abundant_scheduling} a random long queue $i^{\rm long}=i$ with probability $(1-{\eps})y_{i,j}$ for each $i\in {\cal S}^{\rm long}$ and $i^{\rm long} = \perp$ with the residual probability $1-\sum_{i\in {\cal S}^{\rm long}}(1-\eps)y_{i,j}$. 

Next, in the second phase, we give a higher priority to matching $\mathfrak{c}$ with a supplier from a short queue. If $i^{\rm short}\neq \bot$, then we match $\mathfrak{c}$ to a supplier in queue $i^{\rm short}$ (Line \ref{lin:short_matching}). In this case, if $\mathfrak{c}$ is of a contentious type, we also schedule a delayed match with $i^{\rm long}$ by incrementing the virtual buffer ${\cal V}_{i^{\rm long}, j}$ on Line \ref{lin:task_increment}. Alternatively, if $\mathfrak{c}$ is not matched to a short queue---because $i^{\rm short} = \bot$---we proceed to a low-priority match with a supplier from a long queue. That is, we match $\mathfrak{c}$ with a supplier in queue $i^{\rm long}$ if $i^{\rm long}\neq \bot$ and $\ell_{i^{\rm long}}>0$ (Line~\ref{lin:abundant_non-contentious_match}).


At the end of this process, if $\mathfrak{c}$ is still unmatched, we use the surplus of type $j\in {\cal C}^{\rm ct}$ to match with the virtual buffers, corresponding to scheduled matches that have not yet been fulfilled. To this end, we draw a random long queue $i^\dagger \in {\cal S}^{\rm long}$ with probability  $\prpar{i^\dagger = i }$ proportional to $y_{i,j}$. If a scheduled match is pending (i.e. ${\cal V}_{i^\dagger,j}>0$) and the queue is non-empty ($\ell_{i^\dagger}>0$), we match $\mathfrak{c}$ with a supplier in queue~$i^\dagger$. 
\begin{figure}[ht]
    \centering
    \includegraphics[scale=0.37]{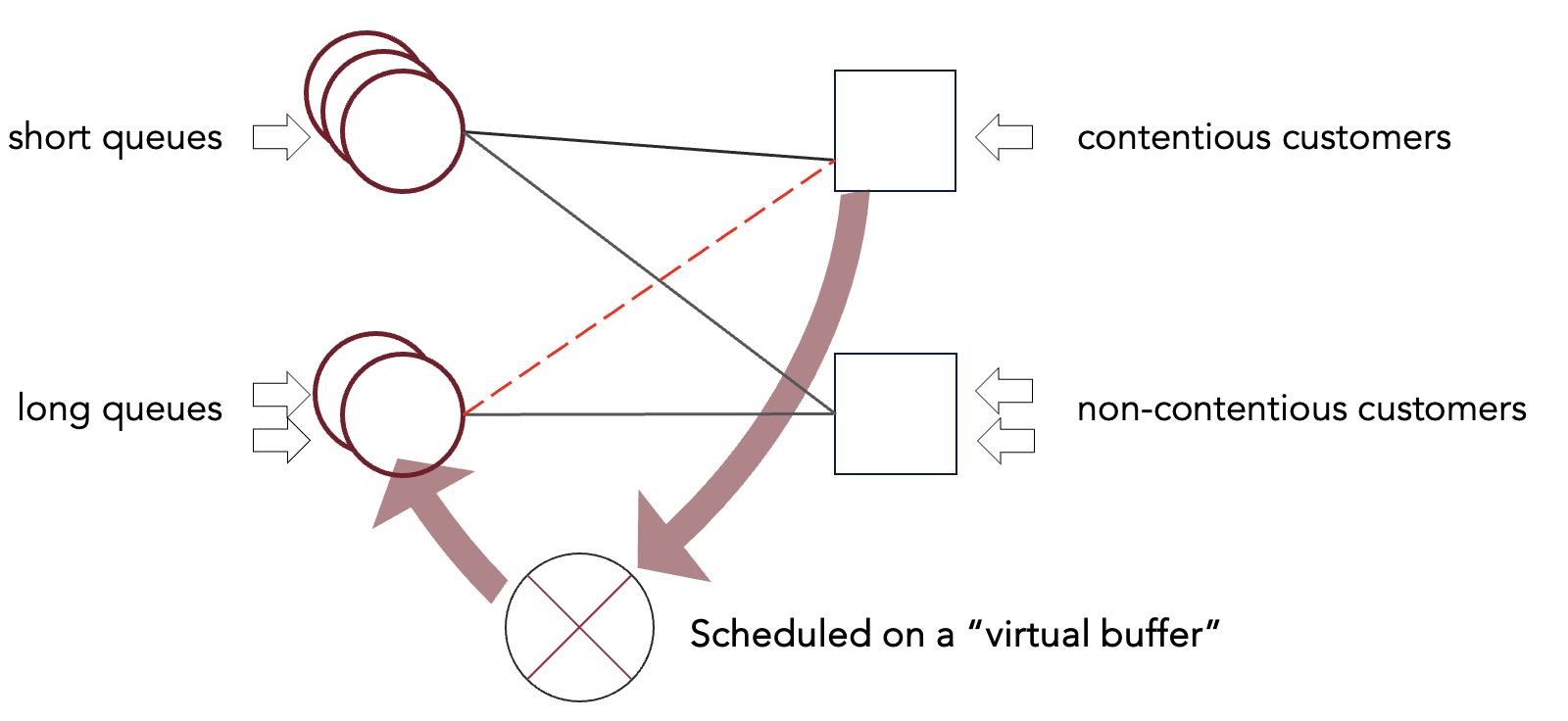}
    \caption{A schematic visualization of ${\pi^{\rm pr}}$}
    \label{fig:pr}
        \raggedright
    \footnotesize The solid black line represent the prioritized matches. The dashed red line represents the de-prioritization of long queues for contentious customers. The thick red arrows represent scheduled  and delayed matches. Upon the arrival of each new customer, $\pi^{\rm pr}$  samples matching decisions based on the LP solution. If there is contention between matching a supplier from a short queue and a supplier from a long queue, it is resolved by always prioritizing the former. In certain cases, deprioritized matches with long queues are scheduled on a virtual buffer to be fulfilled later.
\end{figure}

\begin{algorithm}[ht]
    \caption{Priority Rounding, $\pi^{\rm pr}$}
    \label{alg:priority_matching}
    \begin{algorithmic}[1]
    {\small
        \STATE Let $\{x_{\bs{M}}^{\boldsymbol{\ell}}\}_{{\bs{M}}, \boldsymbol{\ell}}$ and $\{y_{i,j}\}_{i,j}$ be a solution of $NLP(\bar{\ell})$ as per \Cref{lem:NLP_feasible} 
        \STATE For all $i \in {\cal S}^{\rm long}$ and $j \in \overline{{\cal C}^{\rm ct}}$, initialize the virtual buffer ${\cal V}_{i,j} \gets 0$
        \FORALL{arrival times} \label{line:loop-start}
            \STATE Let $\boldsymbol{\ell} \in [\bar{\ell}]_0^{n_s}$ be the state of short queues, i.e., $\ell_i$ is queue $i$'s length for $i \in {\cal S}^{\rm short}$
            \IF{a supplier of type $i \in {\cal S}^{\rm short} $ has arrived}
                \IF{$\ell_i = \bar{\ell}$}
                    \STATE Discard that supplier
                \ENDIF
            \ELSIF{a customer $\mathfrak{c}$ of type $j$ has arrived}
\vspace*{0.2cm}
            \STATE \CommentSty{\color{blue} /* sampling */}

            \STATE Draw  ${\bs{M}}^{\rm short} \in {\cal D}(S^{\rm short})$ with $\prpar{{\bs{M}}^{\rm short} ={\bs{M}}}=\frac{x_{\bs{M}}^{\boldsymbol{\ell}}}{\sum_{{\bs{M'}} \in {\cal D}({\cal S}^{\rm short})} x_{\bs{M'}}^{\boldsymbol{\ell}}}$ \label{lin:s_def} and define $i^{\rm short} = {\bs{M}}^{\rm short}(j)$

            \STATE Draw $i^{\rm long}\in {\cal S}^{\rm long}\cup \{\bot\}$ with $\prpar{i^{\rm long} = i} = (1-\eps)y_{i,j}$ and $\prpar{i^{\rm long} = \bot} = 1-\sum_{i \in {\cal S}^{\rm long}} (1-\eps)y_{i,j}$ \label{lin:abundant_scheduling}

            \vspace*{0.2cm}
            \STATE \CommentSty{\color{blue} /* matching \& scheduling */}
            \IF{ $i^{\rm short}\neq \bot$} 
            \STATE Match $\mathfrak{c}$ to a supplier in queue $i^{\rm short}$ \hspace{0.7cm} \CommentSty{\color{red} /* high-priority match with a short queue */}\label{lin:short_matching}
                            
            \IF{$j \in {\cal C}^{\rm ct}$ and $i^{\rm long} \neq \bot$} \label{lin:abundant_contentious_case}
                    \STATE ${\cal V}_{i^{\rm long},j} \gets {\cal V}_{i^{\rm long},j} + 1$ \hspace{1.15cm} \CommentSty{\color{red} /* scheduling a delayed match on a virtual buffer */} \label{lin:task_increment} 
                \ENDIF 
            \ELSIF{$i^{\rm long} \neq \bot$ and $\ell_{i^{\rm long}} > 0$} \label{lin:low_priority_matching}            
                    \STATE Match $\mathfrak{c}$ to a supplier in queue $i^{\rm long}$  \label{lin:abundant_non-contentious_match} \hspace{0.6cm} \CommentSty{\color{red} /* low-priority match with a long queue */}    
            \ENDIF
            \vspace*{0.2cm}
            \STATE \CommentSty{\color{blue} /* surplus matching  */}
            \IF{$\mathfrak{c}$ is unmatched and $j \in {\cal C}^{\rm ct}$ }
                \STATE Draw  $i^\dagger \in {\cal S}^{\rm long}$ with $\prpar{i^\dagger=i}= \frac{y_{i, j}}{\sum_{i'\in {\cal S}^{\rm long}} y_{i',j}} $  
                \IF{${\cal V}_{i^\dagger,j} > 0$ and $\ell_i^\dagger > 0$}
                        \STATE Match $\mathfrak{c}$ to a supplier in queue $i^\dagger$ \hspace{0.6cm} \CommentSty{\color{red} /* delayed match with a virtual buffer */}  \label{lin:abundant_contentious_match}
                \ENDIF         
                \STATE ${\cal V}_{i^\dagger,j} \gets \max\{0,{\cal V}_{i^\dagger,j} - 1\}$ 
                \label{lin:abundant_contentious_decrement}
                \ENDIF 
\ENDIF
        \ENDFOR }
    \end{algorithmic}
\end{algorithm}

\subsection{Proof Outline of \Cref{thm:constant_ptas}}\label{ssec:pm_analysis}


In this section, we analyze our policy and show that it satisfies $\tau(\pi^{\rm pr}) \geq (1-O(\eps))\tau^*$ and $c(\pi^{\rm pr}) \leq c^*$. The analysis proceeds by proving that $\pi^{\rm pr}$ tracks the fractional match rates described by $NLP(\bar{\ell})$ up to an $O(\eps)$-fraction.  First, we focus on short queues and show a strong convergence property, which follows from the structure of $\pi^{\rm pr}$. Analyzing long queues is more difficult. While we do not precisely characterize their steady-state, we establish crucial structural properties for performance analysis. Finally, we put these pieces together and derive our performance guarantees.  Unless specified otherwise, we analyze the system under $\pi^{\rm pr}$, which is sometimes omitted to lighten the notation. {When analyzing the transient system, we assume by convention  that the queues are empty at $t=0$.} {To simplify the exposition, we analyze a simplified policy where if Line \ref{lin:low_priority_matching} is reached and $\ell_{i^{\rm long}} = 0$, we do \textit{not} proceed to surplus matching.} {The analysis easily extends to the original policy $\pi^{\rm pr}$ as we prove that these events are infrequent.} 





\paragraph{Easy case: Short queues.} Denote by  ${\cal L}_{\rm short}^{\pi}(t) \in [\bar{\ell}]_0^{n_s}$ be the state of the short queues at time $t$ under policy $\pi$, i.e., the vector of queue lengths.
Since short queues always have a higher priority, it is intuitive that their evolution must track exactly the match rates of the LP solution.
As a result, ${\cal L}_{\rm short}^{\pi^{\rm pr}}(t)$ converges, as $t \to \infty$, to the stationary distribution described by $NLP(\bar{\ell})$:
\begin{lemma}\label{lem:short_convergence}
    For every $\boldsymbol{\ell} \in [\bar{\ell}]_0^{n_s}$, we have
    $\lim_{t \to \infty} \prpar{{\cal L}^{\pi^{\rm pr}}_{\rm short}(t) 
    = \boldsymbol{\ell}} = \sum_{{\bs{M}} \in {\cal D}({\cal S}^{\rm short})} x_{\bs{M}}^{\boldsymbol{\ell}}\ $.
\end{lemma}
The proof of this lemma appears in Appendix \ref{prf:short_convergence}. To simplify the exposition, we recast \Cref{lem:short_convergence} as stating that, for every time $t \geq 0$,
\begin{align}
    \pr{{\cal L}^{\pi^{\rm pr}}_{\rm short}(t) = \boldsymbol{\ell}} = \sum_{{\bs{M}} \in {\cal D}({\cal S}^{\rm short})} x_{\bs{M}}^{\boldsymbol{\ell}}  \ .\label{eq:short_stationary_dist}
\end{align} 
This is without loss of generality as we focus on the long-term cost and throughput rates. 

\paragraph{Difficult case: Structural properties for long queues.} 
Analyzing long queues is more complex since the LP uses fluid variables. Moreover, suppliers in long queues have a lower priority and their corresponding randomized matches may be postponed via the virtual buffers. Therefore, if we wish to argue that $\pi^{\rm pr}$ achieves match rates  approximately equal to $\gamma_j y_{i,j}$ for all $(i,j)\in {\cal S}^{\rm long} \times {\cal C}$, we need three important properties {that bound the losses relative to $NLP(\bar{\ell})$}: (i) long queues must be rarely empty, (ii) virtual buffers (for contentious types) must be bounded in expectation, and (iii) non-contentious types must receive enough matches.

Property (i) checks that our definition of long queues  is internally consistent. As per Claim~\ref{obs:depleted}, our hybrid LP  anticipates the long queues to be rarely depleted. This is verified by the next lemma, {whose proof in Appendix \ref{prf:abundant_empty_prob} uses a stronger version of Claim~\ref{obs:depleted}.}
\begin{lemma}\label{lem:abundant_empty_prob}
For all $i \in {\cal S}^{\rm long}$, $t \geq 0$, ${\boldsymbol{\ell}} \in [\bar{\ell}]_0^{n_s}$, we have $\prpar{{\cal L}^{\pi^{\rm pr}}_i(t) = 0 \; \left| \; {\cal L}^{\pi^{\rm pr}}_{\rm short}(t) = \boldsymbol{\ell}\right.} \leq {\eps}$.
\end{lemma}

Property (ii) ensures that we eventually get on with all the scheduled matches between long queues $i\in {\cal S}^{\rm long}$ and contentious types $j\in {\cal C}^{\rm ct}$. The virtual buffer ${\cal V}_{i, j}(t)$ at time $t$ represents previously scheduled matches, according to the LP rounding, that are pending in the real process. The next lemma shows that this inventory does not build up over time. {The proof in Appendix~\ref{prf:scheduler_bounded} establishes the stability and bounded expectation of ${\cal V}_{i,j}$ as a consequence of its negative expected drift.} 
\begin{lemma}\label{lem:scheduler_bounded}
    For all $i \in {\cal S}^{\rm long}$ and $j \in {\cal C}^{\rm ct}$, there exists a constant $L_{i,j} \geq 0$ (independent of $t$) such that $\ex{{\cal V}_{i, j}(t)} \leq L_{i,j}$ for every $t \geq 0$. 
\end{lemma}
Finally, Property (iii) facilitates the analysis of matches between long queues $i\in {\cal S}^{\rm long}$ and non-contentious types $j\in \overline{{\cal C}^{\rm cr}}$. Recall that  $\pi^{\rm pr}$ prioritizes matching with the short queue $i^{\rm short}$ on Line~\ref{lin:short_matching} over matching with the long queue $i^{\rm long}$ on Line~\ref{lin:abundant_non-contentious_match}; if so, the corresponding match is permanently ``cancelled''. The next lemma (proof in Appendix \ref{prf:lem_non_cont}) shows that, for each $j\in \overline{{\cal C}^{\rm cr}}$, a fraction $(1-O(\eps))$ of the match rate $\gamma_j y_{i,j}$ is not cancelled. Let $j(t) \in {\cal C} $ denote the customer type that arrives at time $t$, if any, and let $i^{\rm short}(t) \in {\cal S}^{\rm short} \cup \{\bot\}$ be the random draw on Line~\ref{lin:s_def}.

\begin{lemma} \label{lem:ineq:line_oncu_reach}
For all $t\geq 0$, we have $\prpar{i^{\rm short}(t) = \bot \; | \; j(t)\in \overline{{\cal C}^{\rm ct}}} \geq 1-{\eps} $.
\end{lemma}

\paragraph{Completing the performance analysis.} \Cref{lem:short_convergence} immediately implies that supplier from  short queues achieve exactly the same match rates as those in our fractional solution of $NLP(\bar{\ell})$. To keep this proof concise, we focus on match rates between each supplier type $i \in {\cal S}^{\rm long}$ and each customer type $j \in {\cal C}$, and prove it is $(1-O(\eps))$-close to that in the $NLP(\bar{\ell})$ solution. More concretely, we denote  by $\tau_{i,j}(\pi^{\rm pr})$  the expected average match rate between suppliers of type $i \in {\cal S}^{\rm long}$ and customers of type $j \in {\cal C}$. Our objective is to prove that $(1-O(\eps)) \cdot \gamma_j y_{i,j} \leq \tau_{i,j}(\pi^{\rm pr}) \leq \gamma_j y_{i,j}$. Combined with the match rates of short queues, these inequalities imply  $c(\pi^{\rm pr}) \leq c^*$ and $\tau(\pi^{\rm pr}) \geq (1-O(\eps))\tau^*$. In what follows, we use $M_{i,j}(t_1, t_2)$ to denote the (random) number of matches between type $i$ and type $j$ in the time interval $[t_1,t_2)$. Furthermore, $A_j(t_1, t_2)$ stands for the random number of type-$j$ customer arrivals in that interval.

Fix $j \in \overline{{\cal C}^{\rm ct}}$. 
Here, we discretize time in small steps $\Delta t$ and argue that
\begin{align}
    \tau_{i,j}(\pi^{\rm pr}) &= \liminf\limits_{T \to \infty} \frac{1}{T} \cdot \ex{M_{i,j}(0, T)} \notag \\ &\geq \liminf\limits_{k \to \infty} \frac{1}{k \Delta t} \cdot \left(  \ex{M_{i,j}(0, k\Delta t)} - \gamma_j \Delta t \right) \ , \label{ineq:T_vs_kdelta}
\end{align}
where setting $k = \lceil T / \Delta t \rceil$, we have $\ex{M_{i,j}(0, k \Delta t)}  \leq \ex{M_{i,j}(0, T)} + \gamma_j \Delta t$, since $M_{i,j}(t_1, t_2) \leq A_j(t_1, t_2)$ and $A_j(t_1, t_2)$ is a Poisson process with rate $\gamma_j$. The right hand side of \eqref{ineq:T_vs_kdelta} is equal to
\begin{align}
    &\liminf\limits_{k \to \infty} \frac{1}{k \Delta t} \cdot \ex{M_{i,j}(0, k\Delta t)} \notag \\
    & \; = \liminf\limits_{k \to \infty} \frac{1}{k \Delta t} \cdot {\sum_{r = 0}^{k-1} \ex{M_{i,j}(r\Delta t, (r+1)\Delta t)}} \notag
    \\ & \; \geq \liminf\limits_{k \to \infty} \frac{1}{k \Delta t} \cdot \sum_{r = 0}^{k-1} \pr{A_j(r\Delta t, (r+1) \Delta t) = 1} \ex{\left. M_{i,j}(r\Delta t, (r+1)\Delta t)\;  \right | \; A_j(r\Delta t, (r+1) \Delta t) = 1} \ , \label{ineq:one_arrival_bound}
\end{align}
where the inequality considers, as a lower bound, the case where exactly one customer arrives in each step. Observe that, conditioned on exactly one type-$j$ customer arriving within the interval $[r\Delta t,, (r+1)\Delta t)$, this customer is matched to a type-$i$ supplier provided that upon its arrival, we have (i) $i^{\rm short} = \bot$ in Line \eqref{lin:s_def}, (ii) $i^{\rm long} = i$ in Line \ref{lin:abundant_scheduling}, and (iii) queue $i$ is nonempty at that time. \ Let $t_{\mathfrak{c}}$ be the arrival time of that customer, denoted by $\mathfrak{c}$. Since $i^{\rm long}$ is independent of $i^{\rm short}$ and ${\cal L}_i(t_{\mathfrak{c}})$, the RHS of \eqref{ineq:one_arrival_bound} is at least 
\begin{align}
    & \liminf\limits_{k \to \infty} \frac{1}{k \Delta t} \cdot \sum_{r = 0}^{k-1} \pr{A_j(r\Delta t, (r+1) \Delta t) = 1}  \cdot (1-\eps)y_{i,j} \cdot \pr{i^{\rm short} = \perp, {\cal L}^{\pi^{\rm pr}}_i(t_{\mathfrak{c}}) > 0}  \notag
    \\ & \; = \liminf\limits_{k \to \infty} \frac{1}{k \Delta t} \cdot \sum_{r = 0}^{k-1} e^{-\gamma_j \Delta t}\gamma_j\Delta t \cdot (1-\eps)y_{i,j} \cdot \pr{i^{\rm short} = \perp} \cdot \pr{{\cal L}^{\pi^{\rm pr}}_i(t_{\mathfrak{c}}) > 0 \; | \; i^{\rm short} = \perp } \notag
    \\ & \; \geq \liminf\limits_{k \to \infty} \frac{1}{k \Delta t} \cdot \sum_{r = 0}^{k-1} e^{-\gamma_j \Delta t}\gamma_j\Delta t \cdot (1-\eps)y_{i,j} \cdot (1-\eps) \cdot (1-\eps) \notag
    \\ & \; = \liminf\limits_{k \to \infty} \frac{1}{k} \cdot \sum_{r = 0}^{k-1} e^{-\gamma_j \Delta t}\gamma_j \cdot (1-\eps)^3 y_{i,j} \notag \ ,
\end{align}
where the inequality follows from \Cref{lem:abundant_empty_prob} and  \Cref{lem:ineq:line_oncu_reach}. Since $\Delta t$ was arbitrary, we can take it to 0, in which case the above bound converges to $(1-\eps)^3 \gamma_j y_{i,j}$ which is the desired lower bound. As for the upper bound, note that a necessary condition for a match between a type-$j$ customer and a type-$i$ supplier is that the draw in Line \eqref{lin:abundant_scheduling} is $i^{\rm long} = i$, upon the customer's arrival. Since the probability of this draw is independent of the past and exactly equal to $(1-\eps)y_{i,j}$, we obtain $\tau_{i,j}(\pi^{\rm pr}) \leq \gamma_j (1-\eps)y_{i,j}$, which implies that $\tau_{i,j}(\pi^{\rm pr}) = (1-O(\eps)) \cdot \gamma_j y_{i,j}$, as desired. 

    We now turn to the match rate between contentious customer types $j \in {\cal C}^{\rm ct}$ and supplier type $i \in {\cal S}^{\rm long}$. Such matches  occur either in the {\textit{matching} \& \textit{scheduling}} phase or in the \textit{surplus matching} one. 
    An argument nearly identical to that used for non-contentious types shows that the average rate of matches during the {\textit{matching} \& \textit{scheduling}} phase is 
    \begin{eqnarray} \label{ineq:match_sched}
        \left(1 - O(\eps)\right) \cdot \gamma_j y_{i,j} \cdot (1 - x_{{\rm short}, j}) \ \ \ \text{where} \ \ \ x_{{\rm short}, j} = \sum_{\boldsymbol{\ell} \in [\bar{\ell}]_0^{n_s}}  {\sum\limits_{\substack{{\bs{M}} \in {\cal D}({\cal S}^{\rm short}) : \\ j \in {\bs{M}}}} x_{\bs{M}}^{\boldsymbol{\ell}}} \ .
    \end{eqnarray}

    Focusing on {\it surplus} matches, the average rate of such matches can be determined for by counting the number of time that Line \ref{lin:abundant_contentious_decrement} is reached with $i^\dagger = i$ and ${\cal L}^{\pi^{\rm pr}}_i(t) > 0$. Let ${\cal V}^{\rm inc}_{i,j}(t_1, t_2)$ be the number of times Line \ref{lin:task_increment} is executed with $j$ and $i^{\rm long} = i$ in the interval $[t_1, t_2)$. By equation \eqref{eq:short_stationary_dist}, the probability that this line is conditional on a type $j$ arrival is exactly $x_{{\rm short}, j}$. Therefore, we have
    \begin{align}
        \liminf\limits_{T \to \infty} \frac{1}{T} \cdot \ex{{\cal V}^{\rm inc}_{i,j}(0, T)} = (1-\eps) \gamma_j y_{i,j} x_{{\rm short},j} \ . \label{eq:average_inc_v}
    \end{align}
    Similarly, define ${\cal V}^{\rm dec}_{i,j}(t_1, t_2)$ to be the number of times Line \ref{lin:abundant_contentious_decrement} is executed with $j$, $i^{\rm long} = i$, and ${\cal V}_{i,j} > 0$ in the interval $[t_1, t_2)$. Whenever ${\cal V}_{i,j}$ is decremented with ${\cal L}^{\pi^{\rm pr}}_i > 0$, a match is made. A decrement happens each time we draw $i^{\rm short} = i^{\rm long} = \bot$. Here, we use the simplification that if Line \ref{lin:low_priority_matching} is reached and $\ell_{i^{\rm long}} = 0$, we do {not} proceed to surplus matching.  Since the state of queue $i$  is independent of $i^{\rm long}$, the probability that ${\cal L}^{\pi^{\rm pr}}_i > 0$ conditional on decrementing ${\cal V}_i$ is
    \[
        \pr{{\cal L}^{\pi^{\rm pr}}_i(t) > 0 \; | \; i^{\rm short} = \bot} \geq 1-\eps \ ,
    \] by \Cref{lem:abundant_empty_prob}. Therefore, letting $M_{i,j}^s(t_1, t_2)$ be the number of surplus matches between queue $i$ and type-$j$ customers, the linearity of expectation implies
    \begin{align*}
        \ex{M^s_{i,j}(t_1, t_2)} \geq (1-\eps) \cdot \ex{{\cal V}_{i,j}^{\rm dec}(t_1, t_2)} \ . 
    \end{align*} 
    Let $\tau_{i,j}^s(\pi^{\rm pr})$ be the long-term average rate of surplus matches between queue $i$ and customer type $j$. Consequently, we remark that
    \begin{align}
        \tau_{i,j}^s(\pi^{\rm pr}) &= \liminf\limits_{T \to \infty} \frac{1}{T} \cdot \ex{M^s_{i,j}(0, T)} \notag \\ & \geq \liminf\limits_{T \to \infty} \frac{1}{T} (1-\eps) \cdot \ex{{\cal V}^{\rm dec}_{i,j}(0, T)} \label{ineq:match_dec_lb} \ . 
    \end{align}
    Since we have ${\cal V}_{i,j}(t) = {\cal V}^{\rm inc}_{i,j}(0, t) - {\cal V}^{\rm dec}_{i,j}(0, t)$ and \Cref{lem:scheduler_bounded} implies $\lim_{T \to \infty}\ex{{\cal V}_{i,j}(T)}/T = 0$, the lower bound in \eqref{ineq:match_dec_lb} entails
    \begin{align}
        \tau_{i,j}^s(\pi^{\rm pr}) \geq \liminf\limits_{T \to \infty} \frac{1}{T} (1-\eps) \cdot \ex{{\cal V}^{\rm inc}_{i,j}(0, T)} \geq (1-\eps)^2 \gamma_j y_{i,j} x_{{\rm short},j} \ , \label{ineq:surplus_lb}
    \end{align}
    where we used equality \eqref{eq:average_inc_v}.  Since we also have $\tau_{i,j}^s(\pi^{\rm pr}) \leq \liminf\limits_{T \to \infty} \ex{{\cal V}_{i,j}^{\rm inc}(0, T)} / T$, it holds that $\tau_{i,j}^s(\pi^{\rm pr}) \leq \gamma_j y_{i,j} x_{{\rm short}, j}$ by inequality \eqref{eq:average_inc_v}. {By combining this observation with~\eqref{ineq:surplus_lb} and~\eqref{ineq:match_sched}, we infer that $\tau_{i,j}(\pi^{\rm pr}) = (1-O(\eps)) \cdot \gamma_j y_{i,j} $, which completes our proof.}
    \hfill \qedsymbol

\section{Approximation Scheme for Euclidean Networks}\label{sec:euclidean}
In this section, we present our result for the Euclidean setting, where matching a customer of type $j$ to a supplier of type $i$ incurs a cost of $c_{i,j}$, derived from a Euclidean embedding. As previously, the goal is to find a policy that achieves a certain throughput rate $\tau^*$ subject to an upper limit on the cost rate $c^*$.  Similarly to \Cref{sec:constant_ptas}, our approach can detect infeasible cost-throughput targets.
\paragraph{Euclidean matching.} Let $d \in {\bb N}^+$ be the dimension of the Euclidean network. Suppliers of type $i \in \mathcal{S}$ arrive at location $\boldsymbol{l}^{i, S} \in [0, 1]^d$ according to a Poisson process with rate $\lambda_i$ and customers of type $j \in \mathcal{C}$ are located at $\boldsymbol{l}^{j,C} \in [0,1]^d$ and arrive according to a Poisson process with rate $\gamma_j$. Moreover, suppliers abandon the system with uniform rate of 1. Our model posits $c_{i,j} = \lVert \boldsymbol{l}^{i, S} - \boldsymbol{l}^{j, C} \rVert_{d}$ where $\lVert \cdot \rVert_d$ is the $d$-dimensional Euclidean norm. For any location vector $\boldsymbol{l} \in [0,1]^d$ and coordinate $k \in [d]$, ${l}_k$ denotes the $k$-th coordinate of $\boldsymbol{l}$. 

The input consists of the arrival rate and location of each of the $n = |{\cal S}|$ supplier types and $m = |{\cal C}|$ customer types. Nonetheless, all our results extend if we consider an infinite number of types that form a mixture of point masses and piece-wise uniform distribution. 
Our main result is an efficient approximation scheme for Euclidean networks assuming a fixed dimension $d=O(1)$. Our algorithm calls the FPTAS for constant-size networks of \Cref{sec:constant_ptas} as a subroutine, to make ``local'' matching decisions in different neighborhoods, with careful choices of the cost-throughput targets.
\euclidptas*
\paragraph{Outline of algorithm and analysis.}  At a high level, our policy develops an approximate reduction to the constant-size network setting of Section \ref{sec:constant_ptas}. To this end, we begin by decomposing the space of locations $[0,1]^d$ into multiple {\em cells} and by approximately solving the dynamic matching problem locally within each cell. As a preliminary step, we argue that for an appropriate cell decomposition, there exists a $(1+\eps)$-approximate policy with the restriction that suppliers and customers can only be matched within a cell; we say that such policies are {\em non-crossing}. Subsequently, finding a near-optimal non-crossing policy can be formulated as a min-knapsack linear program that ``glues'' together policies which are obtained by solving each local-cell instance. Using the solution of this LP and the convexity of the minimum achievable cost for any attainable throughput target, we determine a throughput target for each cell. We cluster the supplier types within each cell to obtain a constant-size network. A careful clustering ensures a small increase in our matching cost rate and allows us to adopt local priority rounding policies, as in \Cref{sec:constant_ptas}.


To keep the paper concise, we defer a formal description to Appendix~\ref{app:euclidean}. We introduce the notion of non-crossing policies in Appendix~\ref{ssec:non-crossing}. The decomposition into local-cell instances appears in Appendix~\ref{ssec:decomposition_lp} and the clustering is specified in Appendix~\ref{ssec:clustering}. Putting these pieces together, we obtain our Euclidean matching policy and prove \Cref{thm:ptas-line}.


\section{Conclusion}

This paper studies the dynamic matching problem where agents arrive and depart over time in structured networks of queues. Our main takeaway is that efficient adaptive policies, obtained from dynamic LP relaxations, offer strong near-optimal performance guarantees across diverse network configurations. Our findings suggest several directions for exploring broader applications and extensions of this approach.


\paragraph{Extension to other settings.}
Our main technical innovation lies in the hybrid approximation framework, which combines the network LP with priority-based matching policies, leveraging different queue timescales. This two-timescale algorithmic design may have further applications in non-bipartite graphs and multi-way matchings. While the classification of queues and customers extends naturally to non-bipartite graphs, our multi-variate birth-death formulation does not characterize the evolution of these graphs in general. Exploring the potential of two-timescale designs in other stochastic control problems may provide new valuable algorithmic results.


\paragraph{Heterogeneous supplier abandonment rates.} Another important generalization is settings where different supplier types may have different abandonment rates. Expanding our results to account for different supplier types with heterogeneous abandonment rates presents significant challenges. Although the network LP can be extended straightforwardly, our rounding policy relies heavily on the assumption of uniform abandonment rates. Considering heterogeneous abandonment rates is likely to be a difficult problem in general. For example, for a single customer type and a finite-horizon setting without supplier replenishment, the best-known algorithmic result is a quasi-PTAS \citep{segev2024near}.


\paragraph{Simplified adaptive policies.}
While our adaptive approximation schemes are efficient and near-optimal, they may require a high degree of adaptivity. It would be interesting to explore intermediary levels of adaptivity, in a continuum from static to fully adaptive policies. Identifying  simpler adaptive designs with strong performance guarantees seems a valuable and practical direction for future research.

\bibliographystyle{apalike}
\bibliography{refs}

\newpage

\begin{APPENDICES}

\section{Additional Proofs from Section \ref{sec:dlp1}}

\subsection{Proof of \Cref{prop:dlp1_benchmark}}
\phantomsection
\label{app:dlp_proof} 
\paragraph{Commitment to matching sets.} Let ${\cal F}_t$ be the canonical filtration generated by the arrival and abandonment processes of customers and suppliers. For some time $t \geq 0$, an admissible (randomized) matching policy $\pi_t: {\bb N} \times [m] \to [0,1]$ is an ${\cal F}_t$-measurable function that maps the state of the system (i.e., the queue length) and the type of the customer arriving at time $t$ to a the probability of matching the customer. We argue that by anticipating the type of the arriving custom, this decision can be made prior to observing the arrival. Specifically, we construct policy $\tilde{\pi}_t$ that mimics $\pi_t$ and achieves the same match rates. At the beginning of time $t$ and prior to observing the potential customer arrival, $\tilde{\pi}_t$ makes $m$ calls to $\pi_t$ with the queue length at time $t$ and every type of the arriving customer. As a result, it obtains a list of probabilities $p_1, \cdots, p_m$. Then, if a type-$j$ customer arrives, $\tilde{\pi}_t$ matches it with probability~$p_j$. Since under both policies an arriving customer is matched with the same probability, the match rates are identical. 

Next, we show that $\tilde{\pi}_t$ can implement the matching probabilities $p_1, \cdots, p_m$ by committing to a random matching set $M \subseteq [m]$, which is a convex combination of all subsets of $[m]$. This fact can be seen by observing that $\bs{p} = (p_1, \cdots, p_m)$ lies in the convex polytope $[0,1]^m$ and thus, by Carathéodory's theorem, $\bs{p}$ can be represented by a convex combination of vertices of $[0,1]^m$ which are exactly the points $\{0,1\}^m$, corresponding to subsets of $[m]$. Therefore, the action of a (randomized) policy, at time $t$ and queue length $\ell$, is a probability distribution $\boldsymbol{\rho}_t \in \Delta(\{0,1\}^m)$, where $\Delta(\{0,1\}^m)$ is the probability simplex over $\{0,1\}^m$. At the beginning of time $t$, our policy selects a matching set $M_t$ according to $\bs{\rho}_t$, and serves the customer arriving at that time if and only if their type is included in $M_t$.

\paragraph{Markov decision process.} Under the above formulation, the queue governed by the policy $(\tilde{\pi}_t)_t$ evolves as a continuous-time Markov decision process, with state space $\mathbb{N}$. As discussed before, the action at time $t$ is a probability distribution $\bs{\rho}_t \in \Delta(\{0,1\}^m)$. Hence, at time $t$, the transition rate from state $\ell$ to $\ell - 1$ is given by ${\cal L}^{\tilde{\pi}}(t) + \sum_{M \subseteq [m]} \bs{\rho}_t(M) \cdot \gamma(M)$, where ${\cal L}^{\tilde{\pi}}(t)$ is the queue length at time $t$, and the transition rate from $\ell$ to $\ell + 1$ is $\lambda$. Consequently, we can frame our problem as minimizing the cost rate, subject to a lower bound of $\tau^*$ on the throughput. Technically speaking, we have a constrained continuous-time MDP with average~criteria.

\paragraph{Proof of \Cref{prop:dlp1_benchmark}.} Since there exists an optimal policy that is stationary (see Appendix \ref{app:opt_stationary}), consider an optimal stationary policy $\pi^*$ with $\tau(\pi^*) \geq \tau^*$. It suffices to show that the marginals $(x_{M}^\ell)$ induced by $\pi^*$ are feasible for $(DLP)$. Recall the interpretation of stationary policies in defining ${\cal B}(\cdot)$ (expression \eqref{eq:single_flow_balance}): every stationary policy $\pi$ can be characterized by the probability of committing to a matching set $M \subseteq[m]$ at each state $\ell \geq 1$. 

Now, let $x_M^\ell$ be the (unconditional) probability that the queue length is $\ell$ and  $\pi^*$ commits to a matching set $M$; clearly, the definition implies $\sum_{M \subseteq [m], \ell \in {\bb N}} x_M^\ell = 1$. To prove $DLP^* \leq c(\pi^*)$, it is enough to show that the intensity matrix of ${\cal L}(t) \in {\bb N}$ (under policy $\pi^*$), between states $\ell \neq \ell'$, is:
    \begin{align}\label{eq:intensity-matrix}
        \tilde{{\cal Q}}_{\ell, \ell'} = \begin{cases}
            \lambda & \text{if } \ell'= \ell+1, \\ \frac{\sum\limits_{M \subseteq {[m]}} x_{M}^\ell \cdot(\gamma(M) + \ell)}{\sum\limits_{M \subseteq {[m]}} x_{M}^\ell } & \text{if } \ell' = \ell-1, \\ 0 & \text{otherwise,}
        \end{cases}
    \end{align}
    and $\tilde{{\cal Q}}_{\ell,\ell} = -\sum_{\ell' \neq \ell} \tilde{{\cal Q}}_{\ell,\ell'}$ for every $\ell \in {\bb N}$. 
    The transition rate to $\ell' = \ell + 1$ is clear since suppliers arrive with rate $\lambda$. For $\ell' = \ell-1$, however, we must consider the matches as well as the abandonments. The abandonments occur independently with aggregate rate of $\ell$. However, for the matches, there is a match between a supplier and a customer of type $j$ if there is an arrival of type $j$ and also our policy commits to a matching set $M$ with $j \in M$. The latter happens, independently from everything else, with probability $\frac{\sum_{M\subseteq[m]:j\in M}x_{M}^\ell}{\sum_{M\subseteq{[m]}}x_{M}^\ell}$. Thus, by the PASTA and Poisson thinning properties, the transition rate from $\ell$ to $\ell-1$ is \begin{align*}\label{expr:transition-rate}
        \ell + \sum_{j \in {[m]}} \gamma_j \cdot \frac{\sum_{M \subseteq [m]:j\in M}x_{M}^\ell}{\sum_{M\subseteq{[m]}}x_{M}^\ell} = \frac{\sum\limits_{M \subseteq {[m]}} x_{M}^\ell \cdot(\gamma(M) + \ell)}{\sum\limits_{M \subseteq {[m]}} x_{M}^\ell } \ .
    \end{align*} Since the measure $\vartheta(\ell) = \sum_{M \subseteq[m]} x_{M}^\ell$ is a probability distribution, it satisfies transition rates \eqref{eq:single_flow_balance}, and only a unique probability measure has this property, the stationary distribution is equal to the one described by $(DLP)$:
        \[ \pr{{\cal L}(\infty) = {\ell}} = \sum_{M\subseteq{[m]}}x_{M}^\ell \ , \] for every $\ell \in {\bb N}$. Moreover, the PASTA property immediately implies that the value of $DLP$, under $(x_M^\ell){M, \ell}$, is equal to $c(\pi^*)$. Since $\pi^*$ is optimal, we infer that $DLP^* \leq c(\pi)$ for any feasible policy $\pi$, as desired. The converse statement that $DLP^* = c(\pi^*)$ can be proved similarly by devising a policy that chooses its matching sets according to $DLP$'s optimal solution. Since the argument is similar to the one above, we omit the~details.

\subsection{Existence of an Optimal Stationary Policy}\label{app:opt_stationary}

We refer the reader to \citet[Ch.~12]{zhang2008constrained}, where conditions for the existence of an optimal stationary policy in constrained continuous-time Markov decision processes (CT-MDPs) are identified. Here, we provide an alternative, self-contained proof specifically tailored to the structure of our problem.

Consider an arbitrary positive integer \( \ell > 0 \). Define a truncated instance \( {\cal I}(\ell) \), which is identical to our original dynamic matching instance \( {\cal I} \), except that no new supplier arrivals occur when the queue length reaches \( \ell \). For any policy \( \pi \) defined on the original instance \( {\cal I} \), we can implement \( \pi \) in the truncated instance \( {\cal I}(\ell) \) by introducing fake supplier arrivals whenever the queue length is \( \ell \). Importantly, only matches involving actual suppliers (i.e., non-fake arrivals) are considered when calculating the associated cost and throughput rates. We denote by \( c^{(\ell)}(\pi) \) and \( \tau^{(\ell)}(\pi) \) the cost and throughput rates, respectively, of policy \( \pi \) when applied to the truncated instance \( {\cal I}(\ell) \). {Note that $c^{(\ell)}(\pi)$ and $\tau^{(\ell)}(\pi)$ are well-defined since the augmented Markov chain corresponding to the number of fake and non-fake suppliers is time-homogeneous, and thus, the empirical match rates converge to their steady-state values.} 


Observe that for every fixed \( \ell > 0 \), the state space of \( {\cal I}(\ell) \) is finite, and thus all transition rates are uniformly bounded. Consequently, we can apply the uniformization method to convert \( {\cal I}(\ell) \) into an equivalent discrete-time constrained Markov decision process (DT-CMDP), where the goal is again to minimize the cost rate subject to a lower bound on the throughput rate. Any finite-state, unichain DT-CMDP admits an optimal stationary policy \citep[Thm.~4.1]{altman2021constrained}. Since our underlying Markov chain is  irreducible under any policy, the truncated instance \( {\cal I}(\ell) \) admits an optimal stationary policy~\( \pi_{\ell}^* \).\footnote{If the throughput target $\tau^*$ is not feasible in ${\cal I}(\ell)$, the cost of the optimal policy is defined to be $+\infty$.}

Moreover, since $c^{(\ell)}(\pi_{\ell}^*)$ is decreasing in $\ell$, by the Monotone Convergence Theorem, the limit \( \bar{c} = \lim_{\ell \to \infty} c^{(\ell)}(\pi_{\ell}^*) \) exists. By construction, for every policy \( \pi \) for \( {\cal I} \), we have the inequality \( c^{(\ell)}(\pi_{\ell}^*) \leq c^{(\ell)}(\pi) \). Additionally, due to the presence of abandonments, large queue lengths have vanishingly small stationary probabilities, which implies that \( c^{(\ell)}(\pi) \to c(\pi) \) and \( \tau^{(\ell)}(\pi) \to \tau(\pi) \) as \( \ell \to \infty \). Thus, recalling the inequality \( c^{(\ell)}(\pi_{\ell}^*) \leq c^{(\ell)}(\pi) \), we infer that \( c(\pi) \geq \bar{c} \) for every policy~$\pi$. 


We now show the existence of a stationary policy \( \pi^* \) for the original instance \( {\cal I} \) achieving \( c(\pi^*) = \bar{c} \), thereby proving its optimality. To establish this, we first argue the existence of a convergent subsequence \( (\pi_{\ell_k}^*)_{k \in \mathbb{N}} \subseteq (\pi_{\ell}^*)_{\ell} \). Indeed, this follows directly from \Cref{lem:dual-threshold-increasing}, which readily extends to truncated instances. Specifically, the optimal stationary policy \( \pi_{\ell}^* \) is characterized by \( (\alpha_{\ell}^*, \theta_{\ell}^*) \), the solution of the dual of $DLP(\ell)$, and $m$ tie-breaking probabilities $(p^*_1,\cdots, p_m^*)$; given \( (\alpha_{\ell}^*, \theta_{\ell}^*) \), the dual variables $\bs{\delta}$ are uniquely determined to satisfy the complementary slackness of constraints \eqref{ineq:single_dual_bellman}. Consequently, at state $k$, policy $\pi^*_\ell$ serves an arriving type-$j$ customer with probability 1 if $c_j - \theta < \delta^k$, and with probability $p^*_j$ if $c_j - \theta = \delta^k$, where we note that for every $j$, we can have $c_j - \theta = \delta^k$ for at most one value of $k$.\footnote{This follows from the fact the optimal solution $\delta^k$ is strictly increasing in $k$. Otherwise, given \Cref{lem:dual-threshold-increasing}, we must have $\delta^k = 0$ for every sufficiently large $k$, which is infeasible.} Furthermore, if the throughput target $\tau^*$ is feasible for ${\cal I}(\ell)$, the dual LP has a bounded value, which implies that $\alpha^*_\ell, \theta_\ell^*$ are bounded. Thus, by the Bolzano-Weierstrass theorem, the sequence \( (\alpha_{\ell}^*, \theta_{\ell}^*)_{\ell}, (p_1^*, \cdots, p_m^*) \) admits a convergent subsequence, implying the existence of a convergent subsequence of policies \( (\pi_{\ell_k}^*)_{k \in \mathbb{N}} \). Now, define \( \pi^* = \lim_{k \to \infty} \pi_{\ell_k}^* \).

Finally, continuity of the cost rate functional \( c(\cdot) \) ensures
\[c(\pi^*) = c\left(\lim_{k \to \infty} \pi_{\ell_k}^*\right) = \lim_{k \to \infty} c\left(\pi_{\ell_k}^*\right) = \lim_{k \to \infty} c^{(\ell_k)}\left(\pi_{\ell_k}^*\right) = \bar{c} \ . \] 
Therefore, \( \pi^* \) is an optimal stationary policy for the original instance \( {\cal I} \). This completes the proof.

\subsection{Tight Relaxation of $\bs{DLP}$}\phantomsection\label{app:dlp_relaxation}
\begin{claim}\label{clm:dlp_relaxation}
    The relaxed $DLP$ has the same optimal value as $DLP$. 
\end{claim}
\begin{proof}[ sketch]
    We prove the claim by contradiction. Consider $\boldsymbol{x}$, an optimal solution of $(RLP)$ for which there exists some $\ell' \in {\bb N}$ with
    \[ \lambda  \sum_{M \subseteq [m]} x_{M}^{{\ell}' -1} > \sum_{M \subseteq [m]} x_{ M}^{\ell'}  (\gamma(M) + \ell') \ . \] Let $M' \subseteq [m]$ be such that $x_{M'}^{\ell'-1} > 0$. Consequently, for a small $\delta > 0$, we define $\boldsymbol{\tilde{x}}$ as follows:
    \begin{align*}
        \tilde{x}_{M}^\ell = \begin{cases}
            x_{M}^\ell - \delta & \text{ if } M = M', \ell = \ell' - 1 \ , \\ x_{M}^\ell + \delta  & \text{ if } M = M', \ell = \ell' \ , \\ x_{M}^\ell  & \text{ otherwise} \ .
        \end{cases}
    \end{align*}
    It is easy to see that if $\delta$ is sufficiently small, $\boldsymbol{\tilde{x}}$ is feasible for $(RLP)$, and yields the same objective value. Repeating this procedure leads to a feasible and optimal solution for $(DLP)$. Since we trivially have $RLP^* \leq DLP^*$, the proof is now complete. 
\end{proof}
\subsection{Proof of Lemma \ref{lem:dual-threshold-increasing}}\phantomsection\label{app:dual-threshold-increasing}
    First, for a fixed $\ell \in {\bb N}$, we find the tightest $M$ for the constraint \eqref{ineq:single_dual_bellman} by finding 
\begin{align}
 M^\ell &= \text{argmin}_{M}  \sum_{j \in M} c_{j} \gamma_j + \sum\limits_{j \in M} \beta_j + \lambda {\delta^{\ell+1}} - (\gamma(M) + \ell) {\delta}^{\ell} - \alpha - \theta \gamma(M)
 \notag \\ &=  \text{argmin}_M \sum_{j \in M} c_{j} \gamma_j + \sum\limits_{j \in M} \beta_j - \gamma(M) {\delta^{\ell}} - \theta  \gamma(M) \notag \\ & = \text{argmin}_M \sum_{j \in M} \gamma_j  \left(c_{j} + \frac{\beta_j}{\gamma_j} - {\delta^{\ell}} - \theta \right) \ . \label{eq:dual-glp-argmin-s}
\end{align}
It is then clear that if there exists some $j \in M$ with $c_{j} + \frac{\beta_j}{\gamma_j} - \theta > {\delta}^{\ell}$ (or $j \notin M$ with $c_{j} + \frac{\beta_j}{\lambda_j} < {\delta}^{\ell}$), constraint \eqref{ineq:single_dual_bellman} for $M$ is loose since including (excluding) $j$ would relax the constraint. Then, complementary slackness implies that $x_{M}^\ell = 0$ for such $M$. Therefore, having $x_M^\ell > 0$ implies $\hat{M}^\ell \subseteq M \subseteq \hat{M}^\ell \cup \tilde{M}^\ell$. Moreover, note that we must have $x_{M}^\ell > 0$ for at least one $M$ since otherwise, we must have $x_{M}^{\ell'} = 0$ for every $M \subseteq [m], \ell' \in {\bb N}$, which is not feasible. Then, complementary slackness implies that the constraint \eqref{ineq:single_dual_bellman} for $M = M^\ell$ is tight, which leads to the equalities
\begin{align}\label{eq:alpha-cstr}
    \lambda \delta^{\ell+1} =  {\delta}^{\ell} \cdot \ell + \alpha - \sum_{j \in {M}^\ell}  \gamma_j \left(c_{j} + \frac{\beta_j}{\gamma_j} - \theta - {\delta^{\ell}}\right)  \ , 
\end{align} for every $\ell \geq 1$, and $\alpha = \lambda \delta^1$ that corresponds to the constraint for variable $x_{\emptyset}^0$.  

We turn to proving $\delta^\ell \leq \delta^{\ell+1} \leq 0$ for every $\ell \geq 1$. Suppose ad absurdum that $\delta^k > \delta^{k+1}$ for some $k \geq 1$. We now claim that the modification $\delta^\ell = \delta^k$ for every $\ell \geq k+1$ retains the feasibility of $\boldsymbol{\delta}$. It is clear that constraint \eqref{ineq:single_dual_bellman} with $\ell = k$ remains valid for every $M \subseteq [m]$. In fact, all of these constraints for different subsets $M \subseteq [m]$ would be loose. Then, for every $\ell > k$ and $M \subseteq [m]$, we have \begin{align*}
    &-\lambda \delta^{k} + (\gamma(M) + \ell) \delta^{k} + \alpha + \theta  \gamma(M) - \sum_{j \in M}\beta_j \\ & \; \leq -\lambda \delta^{k} + (\gamma(M) + k) \delta^{k} + \alpha + \theta \gamma(M) - \sum_{j \in M}\beta_j \\ & \; < \sum_{j \in M} \gamma_j c_{j} \ ,
\end{align*} where the first inequality follows from $\delta^k \leq 0$ and the second one is from the looseness of constraints for $\ell=k$, argued above. Therefore, we have obtained a new feasible and optimal solution. Nevertheless, having loose constraints for every subset $M$ violates equality \eqref{eq:alpha-cstr} which is true for every dual-optimal solution. The contradiction establishes that $\boldsymbol{\delta}$ is (weakly) increasing. 

To prove the concavity, first note that monotonicity and non-positivity of $\boldsymbol{\delta}$ implies that we have $\delta^\ell \to 0$ as $\ell \to \infty$. Now, suppose that there exists $k > 1$ such that $\delta^{k+1} - \delta^{k} > \delta^{k} - \delta^{k-1}$. Then, by equality \eqref{eq:alpha-cstr}, we have
\begin{align*}
    \frac{\delta^{k+2} - \delta^{k+1}}{\lambda_i} &=  (k+1) \delta^{k+1} - \sum_{j \in M^{k+1}}  \gamma_j \left(c_{j} + \frac{\beta_j}{\gamma_j} - \theta - {\delta^{k+1}}\right) - k \delta^k + \sum_{j \in M^k}  \gamma_j \left(c_{j} + \frac{\beta_j}{\gamma_j} - \theta - {\delta^{k}}\right) \\ & \geq  k (\delta^{k} - \delta^{k-1}) + \mu \delta^k - \sum_{j \in {M}^{k}}  \gamma_j  \left(c_{j} + \frac{\beta_j}{\gamma_j} - \theta - {\delta^{k+1}}\right) + \sum_{j \in {M}^k}  \gamma_j  \left(c_{j} + \frac{\beta_j}{\gamma_j} - \theta - {\delta^{k}}\right) \\ & = k (\delta^{k} - \delta^{k-1}) + \mu \delta^k - \sum_{j \in {M}^{k}}  \gamma_j  \left(c_{j} + \frac{\beta_j}{\gamma_j} - \theta - {\delta^{k}}\right) + \sum_{j \in {M}^k}  \gamma_j  \left(c_{j} + \frac{\beta_j}{\gamma_j} - \theta - {\delta^{k-1}}\right) \\ & \quad + \sum_{j \in M^k} \gamma_j  \left(\delta^{k+1} - \delta^{k}\right) + \sum_{j \in M^k} \gamma_j \left(\delta^{k} - \delta^{k-1}\right) \\ & \geq (k-1) (\delta^{k} - \delta^{k-1}) +  \delta^k - \sum_{j \in {M}^{k}}  \gamma_j  \left(c_{j} + \frac{\beta_j}{\gamma_j} - \theta - {\delta^{k}}\right) + \sum_{j \in {M}^{k-1}}  \gamma_j  \left(c_{j} + \frac{\beta_j}{\gamma_j} - \theta - {\delta^{k-1}}\right) \\ & = \frac{\delta^{k+1} - \delta^{k}}{\lambda} \ ,
\end{align*}
where both inequalities use equation \eqref{eq:dual-glp-argmin-s}, monotonicity of $\boldsymbol{\delta}$, and the assumption on $k$. Now, this result entails that if  the second order difference of $\delta^\ell$ is positive for some $\ell$, it is positive for all larger values. Nevertheless, this is impossible in light of the convergence $\delta^\ell \to 0$. Hence, we infer $$\delta^{\ell+1} - \delta^{\ell} \leq \delta^{\ell} - \delta^{\ell-1} \ ,$$ for every $\ell > 1$. The proof of the lemma is now complete. 

\subsection{Interpretation of the dual formulation vis-a-vis Bellman equations }\label{app:bellmans}
The following proposition formalizes the Bellman equations in the modified instance, using the general approach in continuous-time MDPs \citep{bertsimas1997introduction}. 
\begin{proposition}[c.f. \cite{bertsimas1997introduction}, Prop. 5.3.1]
    If a scalar $\alpha$ and vector $\boldsymbol{\varphi}(\cdot)$ satisfy
\begin{align}
    & \varphi(\ell) = \notag \\ & \; \min_{M \subseteq [m]} \left\{ \sum_{j \in M} \frac{\gamma_j}{\gamma(M)} \cdot (c_j - \theta) - \alpha \cdot \frac{1}{\ell + \lambda +  \gamma(M)}   + \frac{\lambda}{\ell + \lambda +  \gamma(M)} \cdot \varphi(\ell+1) + \frac{\ell + \gamma(M)}{\ell + \lambda +  \gamma(M)} \cdot \varphi(\ell-1) \right\} \label{eq:phi_bellman}
\end{align}
for every $\ell \in {\bb N}^+$, then $\alpha$ is the optimal average cost per stage. Furthermore, if a policy attains the minimum in \eqref{eq:phi_bellman}---by choosing the minimizing matching set $M$---at every state $\ell$, it is optimal. 
\end{proposition}
Equality \eqref{eq:phi_bellman} can equivalently be expressed as
\begin{align*}
    \varphi(\ell) \leq \sum_{j \in M} \frac{\gamma_j}{\gamma(M)} \cdot (c_j - \theta) - \alpha \cdot \frac{1}{\ell + \lambda +  \gamma(M)}   + \frac{\lambda}{\ell + \lambda +  \gamma(M)} \cdot \varphi(\ell+1) + \frac{\ell + \gamma(M)}{\ell + \lambda +  \gamma(M)} \cdot \varphi(\ell-1)  \\ \forall M \subseteq [m] 
\end{align*}
which after a rearrangement, leads to
\begin{align}
     -\lambda \cdot (\varphi(\ell+1) - \varphi(\ell)) +  (\gamma(M) + \ell) \cdot (\varphi(\ell) - \varphi(\ell - 1)) + \alpha \leq \sum_{j \in M} \gamma_j (c_j - \theta) \ . && \forall M \subseteq [m]  \label{ineq:modified_phi}
\end{align}
It is now immediate to see that \eqref{ineq:modified_phi} is invariant under a shift of $\boldsymbol{\varphi}$. Thus, letting $\delta^\ell = \varphi(\ell + 1) - \varphi(\ell)$ shows that the dual constraint \eqref{ineq:single_dual_bellman} is equivalent to the Bellman equation for the modified instance. In fact, had we used global balance equations in our $DLP$ formulation instead of detailed balance ones, we would have recovered \eqref{ineq:modified_phi}. In conclusion, $\delta^\ell$ can be interpreted as the difference in bias of consecutive states.

\subsection{Proof of Lemma \ref{lem:truncate_distrib}}\label{app:state-space-collapse}
\paragraph{Preliminary notation and definitions.} We can write $\rho_\ell = (1 + \sum_{\ell'\geq 1} \prod_{q=1}^{\ell'} a_{q})^{-1}\cdot \prod_{q=1}^{\ell} a_{q} $ where $a_q = \frac{\lambda}{\sum_j \gamma^{(\ell)}_j + \mu q}$ is the birth-death ratio at $q \geq 1$. We know that $\rho_\ell$ is unimodal and use $\ell^{\rm pk}$ to denote its peak. Throughout this section, we use an auxiliary parameter $\eps' = 1-(1-\eps)^{1/4}$. Define $\ell_1 = \min \{\ell \geq 1: a_\ell \leq 1+\eps'\}$ and $\ell_2 = \min \{\ell \geq 1: a_\ell \leq \frac{1}{1+\eps'}\}$. $\ell_1, \ell_2$ exist since the queue is stable from abandonments and $\ell_1 \leq \ell^{\rm pk} \leq \ell_2$. 

\paragraph{Distribution design: elementary operations.} We will make use of three basic alterations of a policy and the problem instance, which will be useful in designing and analyzing our $\bar{\ell}$-bounded policy. Here, we define these operations and state their properties:
\begin{enumerate}
\item {\bf Inflating/deflating arrival rates.} Suppose that we construct an alternative instance where suppliers arrive at rate $ (1+\eps_\ell) \lambda$  at each queue length $\ell$ for some $\eps_\ell\in (-1,+\infty)$, and all else is unchanged. We denote by $\tilde{c}(\pi),\tilde{\tau}(\pi)$ the expected cost and throughput rates for policy $\pi$ in this modified instance with inflated/deflated arrival rates.
\begin{claim} \label{clm:inflation}
Given any policy $\pi$, there exists a policy ${\pi}'$ such that 
$ {\tau}_j({{\pi}'}) \leq \frac{1}{1+\inf_q \{\eps_q\}^+} \cdot \tilde{\tau}_j(\pi)$ for all $j \in [m]$ and $ {\tau}({{\pi}'}) \geq \frac{1}{1+\sup_\ell \{\eps_\ell\}^+} \cdot\tilde{\tau}(\pi)$.
\end{claim}
\begin{proof}
    Our goal is to construct a policy $\tilde{\pi}$ with the mentioned properties for its performance in the original instance. 	
 In each state $q$, we scale the arrivals of type-$i$ servers either by a simulating an independent steam of fake servers with rate $\{\eps_q\}^+ \lambda_i$ or by immediately discarding a fraction $-\{\eps_q\}^-$ of the arriving servers. Fake servers are treated in the queue as real ones. Then, policy $\tilde{\pi}$ mimics $\pi$ in each state, and treats fake servers as real ones---so the queue length now includes fake servers. The resulting policy $\tilde{\pi}$ induces the same stationary distribution with respect to the original instance as policy $\pi$ with respect to the modified instance, i.e.,  ${\rho}^{\tilde{\pi}}_q = \tilde{\rho}^{\pi}_q$. Now, this implies that the cost rate for each customer type $j$ is 
\begin{eqnarray*}
{c}_j({\tilde{\pi}}) &=& \sum_{q \geq0} {\rho}^{\tilde{\pi}}_q \sum_{S: j\in S} \theta_S \cdot  \frac{1}{1+ \{\eps_q\}^+} \cdot c_{i,j} \gamma_j \\
&\leq & \frac{1}{(1+\inf_q \{\eps_q\}^+) } \cdot \sum_{q\geq 0} \tilde{\rho}^\pi_q \sum_{S: j\in S} \theta_S \cdot c_{i,j} \gamma_j\\
& =&  \frac{1}{(1+\inf_q \{\eps_q\}^+) } \cdot \tilde{c}_j(\pi) \ ,
\end{eqnarray*} 
where in the first inequality,  we use the PASTA property, and we observe that the conditional probability that the matched type-$i$ server is real (non-fake) is no more than $\frac{1}{(1+\inf_q \{\eps_q\}^+) }$ upon each match. We similarly analyze the throughput rate, noting that the conditional probability that each matched server is real is at least $\frac{1}{(1+\sup_q \{\eps_q\}^+) }$ upon each match.
\end{proof}

\item {\bf Upper truncation.} Denote by $\pi^{\downarrow(\ell)}$ a policy that mimics $\pi$ until $\ell$, then starts discarding every arriving server. Indeed, this operation induces a stationary distribution that is supported on $[\ell]_0$. We note that if $\pi$ is monotone then $\pi^{\downarrow(\ell)}$ is also monotone.
\begin{claim} \label{clm:upper_trunc}
For every monotone policy $\pi$, state $\ell \geq 0$, and parameter $\eps' > 0$, if $\sum_{q \geq \ell} \mu_q^{\pi} \leq \eps' \cdot \frac{\tau({\pi})}{\tau_{\max}}$, then $\tau_j(\pi^{\downarrow(\ell)}) \leq \tau_j(\pi)$ for all $j\in [m]$ and $\tau(\pi^{\downarrow(\ell)}) \geq (1-\eps') \tau({\pi})$.
\end{claim}
\begin{proof}[ sketch] The truncation shifts the stationary distribution to the left, so expected cost rate may only decrease for a nested policy. With regard to throughput, the contribution of states below $\ell$ may only increase. The contribution of states above $\ell$ to $\pi$'s throughput rate is at most $\sum_{q \geq \ell} \rho_{q} \tau_{\max} \leq \eps' \tau(\pi)$. Both observations imply that $\tau(\pi) - \tau({\pi^{\downarrow(\ell)}}) \leq \eps' \tau(\pi)$, as desired. 
\end{proof}
\item {\bf Left shift.} Denote by $\pi^{\uparrow(\ell)}$ a policy that in each state $q$, for every $q \geq 0$, mimics $\pi$ by following the same matches as $\pi$ in state $q + \ell$ and by further discarding suppliers at rate $\mu \ell$ (i.e., the total abandonment rate would be $\mu(q+\ell)$). We note that if $\pi$ is monotone then $\pi^{\uparrow(\ell)}$  is still monotone.
\begin{claim} \label{clm:left_shift}
For every monotone policy $\pi$, state $\ell\geq 0$, and parameter $\eps' > 0$, if $\sum_{q < \ell} \rho_q \leq \eps' $, then $\tau_j({\pi^{\uparrow(\ell)}}) \leq (1-\eps')^{-1} \tau_j(\pi)$ for all $j\in [m]$ and $\tau({\pi^{\uparrow(\ell)}}) \geq (1-\eps') \tau({\pi})$.
\end{claim}
\begin{proof}[ sketch] We can interpret the policy change as ``eliminating''  expected contributions to cost and throughput from states below $\ell$ and by uniformly scaling up those contributions above $\ell$ to ``renormalize'' the distribution. Since the renormalization factor is at most $(1-\eps')^{-1}$ from the fact that $\sum_{q < \ell} \rho_q \leq \eps'$, the effect on total expected cost from each type $j$ is such that  $\tau_j({\pi^{\uparrow(\ell)}}) \leq (1-\eps')^{-1} \tau_j(\pi)$. Now, we note that throughput rate per unit of time is larger in $\pi$ for queue lengths $q \geq \ell$ than $q < \ell$ from the monotone property, i.e., $\gamma_j^{(\ell)} $ is non-decreasing in $\ell$. It follows that the loss in throughput rate is at most a factor $\eps'$, noting that the eliminated states are $q < \ell$ and their combined probability is $\sum_{q < \ell} \rho_q \leq \eps'$.
\end{proof}
\end{enumerate}

\paragraph{Constructing the alternative policy ${\tilde{\pi}}$.} Having the elementary operations at our disposal, our goal now is to construct a randomized policy $\tilde{\pi}$ such that it is polynomially bounded and approximates $\pi$ with our desired level of accuracy. We proceed in three steps. First, we operate an upper truncation of $\pi$ and define $\pi^{(1)} = \pi^{\downarrow(\ell)}$ with $\ell = \ell_2 + \lceil \log_{1+\eps'}(\frac{\tau_{\max}}{\eps'^2\tau(\pi)}) \rceil$. Note that 
\[
\sum_{q > \ell} \rho_\ell \leq \rho_{\ell_2} \cdot \frac{1}{(1+\eps')^{{\log_{1+\eps'}}\left(\frac{\tau_{\max}}{\eps'^2\tau(\pi)}\right)}} \cdot \left(\sum_{q > \ell}  (1+\eps')^{\ell-q} \right)\leq \frac{\eps'^2 \tau(\pi)}{\tau_{\max} } \cdot \frac{1}{\eps'} \leq \eps' \cdot \frac{\tau^\pi}{\tau_{\max} } \ .
\]
It follows from Claim~\ref{clm:upper_trunc} that $\tau_j({\pi^{(1)}}) \leq \tau_j(\pi)$ for all $j\in [m]$ and $\tau({\pi^{(1)}}) \geq (1-\eps') \tau({\pi})$.

Next, we operate a left shift of $\pi^{(1)}$ and obtain $\pi^{(2)} = (\pi^{(1)})^{\uparrow(\ell)}$ with $\ell = \max\{0,\ell_1-\lceil \log_{(1+\eps')}(\frac{1}{\eps'^2})\rceil \}$. If $\ell = 0$, then clearly $\pi^{(1)} = \pi^{(2)}$. Denoting by $(\rho^{\pi'}_\ell)_{\ell}$ the stationary distribution induced by policy $\pi'$, we observe  
\begin{eqnarray*}
\sum_{q < \ell} \rho^{\pi^{(1)}}_q \leq  \rho^{\pi^{(1)}}_{\ell_1} \cdot (1+\eps')^{-\log_{(1+\eps')}(\frac{1}{\eps'^2})} \cdot \left(\sum_{q = 1}^{\ell+1} (1+\eps')^{-q}\right) \leq \eps' \ .
\end{eqnarray*}
By Claim~\ref{clm:left_shift}, we infer that $\tau_j({\pi^{(2)}}) \leq (1-\eps')^{-1} \tau_j(\pi)$ for all $j\in [m]$ and $\tau({\pi^{(2)}}) \geq (1-\eps')^2 \tau({\pi})$. 

Third, we specify our policy $\tilde{\pi}$. If $\ell_2 - \ell_1 \leq K$ with  $K=O(\log\frac{\tau_{\max}}{\tau(\pi)})$, then $\pi^{(2)}$ is $O(\log \frac{\tau_{\max}}{\tau(\pi)})$-bounded and Lemma~\ref{lem:truncate_distrib} straightforwardly follows from setting $\tilde{\pi} = \pi^{(2)}$. The rest of the proof considers the more difficult case where $\ell_2 - \ell_1 \geq K$. In particular, we use $K = \max\{(\log_{1+\eps'} (\frac{\tau_{\max}}{\tau(\pi)}) + \log_{1+\eps'}(\eps' \lceil \log_{(1+\eps')}(\frac{1}{\eps'^2}) \rceil) + 1) \cdot (4+\eps'), \frac{2}{\eps'}\}$; the exact formula is indeed complicated and insigificant, however, it notably leads to a polynomially bounded policy. Consistently with our left shift operation, we define $\tilde{\ell}_1 = \ell_1 - \max\{0,\ell_1-\lceil \log_{(1+\eps')}(\frac{1}{\eps'^2})\rceil\}$ and $\tilde{\ell}_2 = \ell_2 - \max\{0,\ell_1-\lceil \log_{(1+\eps')}(\frac{1}{\eps'^2})\rceil\}$. 

 At a high level, we construct a policy $\tilde{\pi}$ that imitates $\pi^{(2)}$ cost-wise and throughput-wise, but reduces the number of states in the peak region, between $\tilde{\ell}_1$ and $\tilde{\ell}_2$, to be at most $K$. To this end, we first introduce a target stationary distribution for the ``left'',``peak', and ``right'' region of the state space, effectively reducing the number of states in the peak region. Second, we use our elementary operations for distribution design to argue that this target stationary distribution approximates $(\rho_\ell)_{\ell \in {\bb N}}$ well enough and can be achieved by a randomized policy $\tilde{\pi}$, which is $K$-bounded by the definition of our target distribution.
 
 \paragraph{Step 1: Target distribution.} We define regions ${\rm left} = [0, \tilde{\ell}_1 -1]$, ${\rm peak} = [\tilde{\ell}_1, \tilde{\ell}_2-1]$, and ${\rm right} = [\tilde{\ell}_2, \tilde{\ell}_2 + \lceil \log_{{1+\eps'}}(\frac{\tau_{\max}}{\eps'^2 \tau(\pi)}\rceil]$ and let $\rho_{\rm left} = \sum_{\ell < \tilde{\ell}_1} \rho_\ell$, $\rho_{\rm peak} = \sum_{\ell = \tilde{\ell}_1}^{\tilde{\ell}_2-1} \rho_\ell$, and $\rho_{\rm right} = \sum_{\ell \geq \tilde{\ell}_2} \rho_\ell$. By a slight abuse of notation, for any distribution $\xi$, we define $a^\xi_\ell = \frac{\xi_\ell}{\xi_{\ell-1}}$. Motivated by our elementary operations and our goal of having a condenssed peak region, we define the notion of \emph{simple} distributions:
 \begin{definition}{\it 
     A distribution $\xi$ is $K_0$-simple, if it satisfies }
     \begin{eqnarray} \label{eq:birth-death-target}
 {a}^{\xi}_\ell &\in \begin{dcases}  \left\{a_{\ell} \right\} & \text{if } \ell \leq \tilde{\ell}_1 \vee \ell \in  \left[(\tilde{\ell}_1 + K_0+1), (\tilde{\ell}_1 + K_0 + 1) + \left\lceil \log_{{1+\eps'}}\left(\frac{\tau_{\max}}{\eps'^2 \tau(\pi)}\right)\right\rceil\right] \ ,\\
  \left\{({1+\eps'})^{-1}, 1, 1 + \eps'\right\}  & \text{if } \ell \in  [\tilde{\ell}_1,\tilde{\ell}_1+K_0] \ .
 \end{dcases}
\end{eqnarray}
 \end{definition}
\noindent In other words, $\xi$ has the same ratio of consecutive stationary probabilities in the left and right region, however, in the peak region, where the actual ratio is in $[\frac{1}{1+\eps'}, 1+\eps']$, $\xi$ has a ratio that is either $\frac{1}{1+\eps'}, 1,$ or $1+\eps'$. Here, we use the term region for $\xi$ with respect to the boundary values $\tilde{\ell}_1$ and $\tilde{\ell}_1 + K_0$, i.e., its left region is $[0, \tilde{\ell}_1]$, right region is $[(\tilde{\ell}_1 + K_0 + 1), (\tilde{\ell}_1 + K_0 + 1) + \lceil \log_{{1+\eps'}}(\frac{\tau_{\max}}{\eps'^2 \tau(\pi)})\rceil]$, and the peak region is $[\tilde{\ell}_1, \tilde{\ell}_1 + K_0]$. For every region $v$, we define $\xi_v$, similar to $\rho_v$, as the aggregate probability in that region. The next claim shows that there exists a {$O(\log\frac{\tau_{\max}}{\tau(\pi)})$}-simple distribution that gives an accurate approximation of the corresponding stationary probabilities of $\pi^{(2)}$. Indeed, this simple distribution serves as our \emph{target distribution}. In the following, we drop the superscript and with an abuse of notation, $\rho$ refers to the stationary distribution under policy ${\pi^{(2)}}$.
\begin{claim} \label{clm:calibration} 
There exists a $K_0$-simple distribution ${\rho}^{\rm target}$ such that 
\[(1-\eps')\left({\rho}_{v} - \eps' \frac{\tau(\pi)}{\tau_{\max}}\right) \leq {\rho}^{\rm target}_{v}  \leq \frac{{\rho}_{v} + \eps' \frac{\tau(\pi)}{\tau_{\max}}}{1-\eps'}\]
for each region $v \in \{{\rm left},{\rm right}, {\rm peak}\}$, where $K_0 = O(\log\frac{\tau_{\max}}{\tau(\pi)})$ and $\tau_{\max} = \sum_{j\in [m]}\gamma_j$. Specifically, we have $K_0~\leq~(\log_{1+\eps'} (\frac{\tau_{\max}}{\tau(\pi)}) + \log_{1+\eps'}(\eps' \lceil \log_{(1+\eps')}(\frac{1}{\eps'^2}) \rceil) + 1) \cdot (4+\eps')$. 
\end{claim}
The proof follows an involved constructive argument and we defer it to \Cref{prf:calibration}. Subsequently, we use this claim to devise our desired policy $\tilde{\pi}$.  

\paragraph{Step 2: Policy ${\tilde{\pi}}$.} Define $\theta^{\rm peak}_M$ to be the probability that $\pi^{(2)}$ selects $M$ as its matching set (recall the interpretation of stationary policies given for ${\cal B}(+\infty)$), conditional on the current state being in in the peak segment $\ell \in [\tilde{\ell}_1, \tilde{\ell}_2]$. Note that this conditional probability is identical to that of the original policy $\pi$ over the peak segment $\ell \in [\ell_1,\ell_2]$ because it is invariant to the previous left shift and upper truncation operations. Consequently, we devise a penultimate policy $\tilde{\pi}'$ as follows:
\begin{itemize}
\item In every state $\ell\in [0,\tilde{\ell}_1]$ of the left segment, policy $\tilde{\pi}'$ follows the same matching decisions as $\pi^{(2)}$ in state $\ell$. 
\item In every state $\ell \in [\tilde{\ell}_1,\tilde{\ell}_1 +K_0]$ of the peak segment, policy $\tilde{\pi}'$ randomizes over matching sets $M$ with the  state-independent distribution $(\theta^{\rm peak}_M)_M$.
\item In every state $\ell \in [(\tilde{\ell}_1 +K_0+1) ,(\tilde{\ell}_1 + K_0 + 1) + \lceil \log_{{1+\eps'}}(\frac{\tau_{\max}}{\eps'^2 \tau(\pi)})\rceil]$ of the right segment, policy $\tilde{\pi}'$ follows the same matching decisions as $\pi^{(2)}$ in state $\tilde{\ell}_2 + \ell - (\tilde{\ell}_1 +K_0+1)$.
\end{itemize}
Furthermore, $\tilde{\pi}'$ also implements the same discarding of servers that $\pi^{(2)}$ adopts due to the left shift operation. 
The logic behind $\tilde{\pi}'$ is that in a carefully constructed modified instance, it achieves the target stationary distribution $\rho^{\rm target}$, which then allows us to show that its cost and throughput are $(1+O(\eps'))$-competitive against $\pi^{(2)}$, as stated in Claim~\ref{clm:achieve-target}. Specifically, we define $\Delta = \max\{0,\ell_1-\lceil \log_{(1+\eps')}(\frac{1}{\eps'^2})\rceil\}$ and consider the modified instance with inflated/deflated arrival rates $(1+\eps_\ell)\lambda$ for suppliers, where in any none-peak state $\ell$, we have $\eps_\ell = 0$, and in any peak state $\ell \in  [\tilde{\ell}_1,\tilde{\ell}_1 +K_0]$, we choose $\eps_\ell \in (-1,+\infty)$ such that $(1+\eps_\ell)\lambda = (1+\eps)^{t} (\sum_{M} \theta^{\rm peak}_M \gamma(M) + \mu (\ell + \Delta))$ for some $t \in \{-1, 0, 1\}$, as prescribed by \Cref{clm:calibration}. Denote by $\tilde{\rho}^{\tilde{\pi}'}$ the stationary distribution induced by $\tilde{\pi}'$ in this modified instance. The next claim summarizes the main property of our construction: policy $\tilde{\pi}'$ attains the target distribution on the modified instance and therefore it yields the desired cost-throughput rates up to an $O(\eps')$ factor.
\begin{claim} \label{clm:achieve-target}
We have $\tilde{\rho}^{\tilde{\pi}'} = \rho^{\rm target}$, and it follows that $\tilde{\tau}({\tilde{\pi}'}) \geq (1-\eps') \tau({\pi^{(2)}}) - (1-\eps')\eps' \tau(\pi)$.
\end{claim}
\begin{proof}[ sketch]
	We calibrated the birth-death rate ratios to match those defined in equation~\eqref{eq:birth-death-target} for $\rho^{\rm target}$, thus $\tilde{\rho}^{\tilde{\pi}'}_\ell = \rho^{\rm target}_\ell$ for every state $\ell \in [0, (\tilde{\ell}_1 + K_0 + 1) + \lceil \log_{{1+\eps'}}(\frac{\tau_{\max}}{\eps'^2 \tau(\pi)})\rceil]]$. To analyze the throughput, we use the property of Claim~\ref{clm:calibration} and observe that we can recover that level of throughput contributions from each segment, noting that the arrival rate of customers is at most $\tau_{\max}$.
\end{proof}
As a final step, it remains to convert $\tilde{\pi}'$ into our policy $\tilde{\pi}$ which achieves nearly the same performance with respect to the original instance, rather than the modified one. Furthermore, we must ensure that the match rate of customer type $j$ is not more than $\tau_j(\pi)$ for any $j \in [m]$. We invoke Claim~\ref{clm:inflation} that facilitates a sensitivity analysis for the effect of inflating/deflating the arrival rates of suppliers. Observe that since ${\pi}^{(2)}$ is nested and makes the same decisions as $\pi$ in the peak region, we have 
\begin{eqnarray} \label{ineq:boundnested}
\left(\gamma^{\pi^{(2)}}_{\tilde{\ell}_1} + \mu \ell_1 \right) \leq \sum_{M} \theta^{\rm peak}_M \gamma(M) + \mu (\ell + \Delta) \leq \left(\gamma^{\pi^{(2)}}_{\tilde{\ell}_2-1} + \mu ({\ell}_2-1) \right) \ ,
\end{eqnarray} for every $\tilde{\ell}_1 \leq \ell < \tilde{\ell}_2$.
Moreover, due to the left shift and right truncation operations, $\gamma^{\pi^{(2)}}_{\tilde{\ell}_1} + \mu {\ell}_1 = \gamma^{\pi^{(1)}}_{\ell_1} + \mu \ell_1 = \gamma^{\pi}_{\ell_1} + \mu \ell_1  \geq (1 +\eps')^{-1}\lambda$, where the last inequality follows from the definition of $\ell_1$. Similarly, we observe that $\gamma^{\pi^{(2)}}_{\tilde{\ell}_2-1} + \mu ({\ell}_2-1) = \gamma^{\pi^{(1)}}_{\ell_2-1} + \mu (\ell_2-1) = \gamma^{\pi}_{\ell_2-1} + \mu (\ell_2-1)  \leq (1 +\eps')\lambda$, where the last inequality follows from the definition of $\ell_2$. Combining with inequality~\eqref{ineq:boundnested}, we obtain for every ${\ell}_1 \leq \ell < {\ell}_2$ that
\begin{eqnarray*} 
(1+\eps')^{-1} \lambda \leq \sum_{M} \theta^{\rm peak}_M \gamma(M) + \mu  (\ell + \Delta) \leq (1+\eps') \lambda_i  \ ,
\end{eqnarray*}
which yields that $\eps_\ell \in (-\eps',\eps')$ for every $\ell$ in the peak segment $[\tilde{\ell}_1,\tilde{\ell}_1+ K_{0}]$. Consequently, Claim~\ref{clm:inflation} shows that, taking $\pi$ as $\tilde{\pi}'$ and $\eps_\ell$ set as the previous instance alterations, there exists $\tilde{\pi}$ such that 
\begin{align*}
 \tau({\tilde{\pi}}) &\geq (1+\eps')^{-1} \tilde{\tau}({\tilde{\pi}'}) \\ & \geq \frac{(1-\eps') \tau({\pi^{(2)}}) - (1-\eps')\eps' \tau(\pi)}{1+\eps'} \\ & \geq \frac{(1-\eps')^3 \tau({\pi}) - (1-\eps')\eps' \tau(\pi)}{1+\eps'} \\ &\geq \left((1-\eps')^4 - (1-\eps')^2\eps'\right) \tau(\pi) \ , 
\end{align*}
where the first inequality is by \Cref{clm:inflation}, the second inequality is by \Cref{clm:achieve-target}, and the third inequality uses the properties of $\pi^{(2)}$. It can be verified that if $\eps' \leq 1-(1-\eps)^{1/4}$, we have $\tau(\tilde{\pi}) \geq (1-\eps)\tau(\pi)$, as desired. 

Finally, we handle the case that $\tilde{\pi}$'s match rate of a customer type $j$ is higher than $\tau_j(\pi)$. In this case, we modify $\tilde{\pi}$ as follows: whenever $\tilde{\pi}$ decides to match a type-$j$ customer, we make this match with probability $\frac{\tau_j(\pi)}{\tau_j(\tilde{\pi})}$ and with probability $1-\frac{\tau_j(\pi)}{\tau_j(\tilde{\pi})}$, we do not make the match but artificially discard the corresponding supplier. It is straightforward to see that this modification makes $\tilde{\pi}$'s match rate of type-$j$ customers exactly equal to $\tau_j(\pi)$. Since we preserve the guarantee of $\tau_j(\tilde{\pi}) \geq (1-\eps)\tau_j(\pi)$ for every $j \in [m]$, the required properties are satisfied and the proof is complete. 

\subsubsection{Proof of \Cref{clm:calibration}.}\label{prf:calibration}
We seek a simple distribution $\tilde{\rho}$ that approximates $\rho$, up to our desired accuracy level, with a simple structure that allows only for ${a}_\ell^{\tilde{\rho}} \in \{(1+\eps)^{-1}, 1, 1+\eps\}$ if $\ell \in [\tilde{\ell}_1, \tilde{\ell}_1 + K_0]$, where we abuse the notation by having ${a}_\ell^{\tilde{\rho}} = \frac{\tilde{\rho}_\ell}{\tilde{\rho}_{\ell - 1}}$. Since ${a}^{\tilde{\rho}}_\ell$ is identical to $a_\ell^{\pi^{(2)}}$ for $\ell$ in regions left and right, our design must be aware of only these two quantities: (i) the total mass in the peak region and also (ii) the ratio between mass at $\tilde{\ell}_2$ and $\tilde{\ell}_1$ (the counterpart of $\zeta = \frac{{\rho}_{\tilde{\ell}_2}}{{\rho}_{\tilde{\ell}_1}}$ in our simple distribution). To this end, let ${p} = \frac{\rho}{\rho_{\tilde{\ell}_1}}$ be the scaled version of $\rho$, e.g., $p_{\tilde{\ell}_1} = 1$. We also use the shorthand $h = \eps'  \frac{\tau(\pi)}{\tau_{\max}}$. We now proceed with considering different cases, based on the value of $\rho_v$ for different regions $v$. 

\paragraph{\bf Case 1: $\rho_{\rm left}, \rho_{\rm right} \geq h$.} 
We further split this case into two subcases and begin with the more intricate one. 

{\it Subcase 1.1: $p_{\rm peak} \geq \frac{1}{\eps'}$.} First, note that by definition of $\tilde{\ell}_1$, $\rho_{\ell}$ is increasing for $\ell \leq \tilde{\ell}_1$, and we have \begin{align}
    h \leq \rho_{\rm left} \leq \left \lceil \log_{(1+\eps')}\left(\frac{1}{\eps'^2}\right)  \right \rceil \cdot \rho_{\tilde{\ell}_1} \ , \notag
\end{align} which implies 
\begin{align}\label{ineq:zeta-bound}
    \zeta = \frac{\rho_{\tilde{\ell}_2}}{\rho_{\tilde{\ell}_1}} \leq \frac{1}{\rho_{\tilde{\ell}_1}} \leq \frac{1}{h\kappa} \ ,
\end{align} where we define $\kappa = \lceil \log_{(1+\eps')}(\frac{1}{\eps'^2}) \rceil$. Now, we construct a $K_0$-simple distribution $\tilde{\rho}$ that approximates $\rho$ by observing $\frac{\tilde{\rho}_{\tilde{\ell}_2}}{\tilde{\rho}_{\tilde{\ell}_2}} \approx \zeta$. To ease the exposition, we operate on the non-normalized (scaled) version of $\tilde{\rho}$, called $\tilde{p}$, that has $\tilde{p}_{\tilde{\ell}_1} = 1$. Consequently, we let $\tilde{\rho}$ be $\frac{\tilde{p}}{\sum_\ell \tilde{p}_\ell}$. Specifically, we construct $\tilde{p}$ such that it satisfies
\begin{enumerate}[(a)]
    \item $\tilde{p}_{\ell} = p_{\ell}$ for every $\ell < \tilde{\ell}_1$.
    \item $(1-\eps'){p_{\rm peak}} \leq \tilde{p}_{\rm peak} \leq p_{\rm peak}$.
    \item $\frac{\zeta}{1+\eps'} \leq {\tilde{p}_{\tilde{\ell}_1 + K_0}} \leq \zeta$. 
    \item $\frac{\tilde{p}_{\ell}}{\tilde{p}_{\ell - 1}} \in \left \{\frac{1}{1+\eps'}, 1, 1+\eps'\right \}$ for every $\tilde{\ell}_1 < \ell < \tilde{\ell}_1+K_0$. 
    \item $\frac{\tilde{p}_\ell}{\tilde{p}_{\ell'}} = \frac{{p}_\ell}{{p}_{\ell'}}$ for every $\ell, \ell' \geq \tilde{\ell}_1 + K_0$.
\end{enumerate}

In other words, our construction approximates the mass in the peak region while preserving the ratio $\frac{\tilde{\rho}_{\tilde{\ell}_1 + K_0}}{\tilde{\rho}_{\tilde{\ell}_1}} \approx \frac{\rho_{\tilde{\ell}_2}}{\rho_{\tilde{\ell}_1}}$ and leaves the ratios in  left and right regions the same. We can straightforwardly ensure properties (a) and (e), that are in the definition of simple distributions \eqref{eq:birth-death-target}, by following $a_\ell$. ‌Thus, the main intricacy is having properties (b) and (c) with the constraints of property (d). 

Before explaining the details of our construction, we note that properties (a)-(e) imply that $\tilde{\rho}$ satisfies our desired approximation guarantees. Indeed, we get 
\[   p_{\rm left} + (1-\eps') {p}_{\rm peak} + \frac{p_{\rm right}}{1+\eps'}  \leq \tilde{p}_{\rm left} + \tilde{p}_{\rm peak} + \tilde{p}_{\rm right} \leq p_{\rm left} + p_{\rm peak} + p_{\rm right} \ , \] which entails, for $v \in \{{\rm left}, {\rm peak}, {\rm right}\}$, that $(1-\eps')\rho_v \leq \tilde{\rho}_v \leq \frac{\rho_v}{1-\eps'}$.

We now discuss how to construct $\tilde{p}$. For now, we assume $\zeta \geq 1$ since the proof for $\zeta \leq 1$ is similar and will be discussed later. To design  First, let $\underline{l} = \lfloor \log_{1+\eps'}(\zeta) \rfloor \geq 0$ and $\bar{l}$ be the largest value of $l \in {\bb N}^+$ that satisfies 
\[ g(\bar{l}, \underline{l}) \triangleq \sum_{k = 0}^{\underline{l}-1} (1+\eps')^k + 2\sum_{k = 0}^{l-1} (1+\eps')^{\underline{l} + k} + (1+\eps')^{\underline{l} + l} \leq p_{\rm peak} \ , \] if any such $l$ exists. If not, set $\bar{l} = 0$ and define $g(0, \underline{l}) \triangleq \sum_{k=0}^{\underline{l}} (1+\eps')^k$. Intuitively, $g(\bar{l}, \underline{l})$ is equal to the total mass in the peak region for the construction where we have $\tilde{p}_{\tilde{\ell}_1 + k} = (1+\eps')^k$ for $0 \leq k \leq \underline{l} + \bar{l}$ and $\tilde{p}_{\tilde{\ell}_1 + \underline{l} + \bar{l} + k} = (1+\eps')^{\underline{l} + \bar{l} - k}$ for $1 \leq k \leq \bar{l}$:

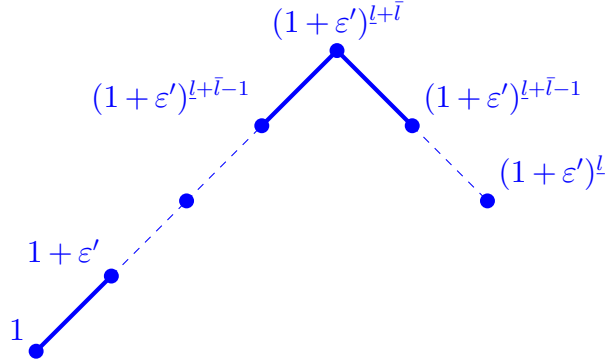
\begin{figure}[H]
\centering
\begin{tikzpicture}[scale=1, transform shape]
\coordinate (O) at (0,0);
\coordinate (A) at (1,1);
\coordinate (B) at (2,2);
\coordinate (C) at (3,3);
\coordinate (P) at (4,4);
\coordinate (D) at (5,3);
\coordinate (E) at (6,2);

\draw[ultra thick, blue] (O) -- (A);
\draw[dashed, blue] (A) -- (B);
\draw[dashed, blue] (B) -- (C);
\draw[ultra thick, blue] (C) -- (P); 
\draw[dashed, blue] (D) -- (E);
\draw[ultra thick, blue] (P) -- (D);

\node at (O) [circle, fill=blue, inner sep=2pt] {};
\node at (A) [circle, fill=blue, inner sep=2pt] {};
\node at (B) [circle, fill=blue, inner sep=2pt] {};
\node at (C) [circle, fill=blue, inner sep=2pt] {};
\node at (D) [circle, fill=blue, inner sep=2pt] {};
\node at (E) [circle, fill=blue, inner sep=2pt] {};
\node at (P) [circle, fill=blue, inner sep=2pt] {};

\node at (O) [above left, blue] {$1$};
\node at (A) [above left, blue] {$1+\eps'$};
\node at (C) [above left, blue] {$(1+\eps')^{\underline{l} + \bar{l}-1}$};
\node at (D) [above right, blue] {$(1+\eps')^{\underline{l} + \bar{l}-1}$}; 
\node at (P) [above, blue] {$(1+\eps')^{\underline{l} + \bar{l}}$};
\node at (E) [above right, blue] {$(1+\eps')^{\underline{l}}$};
\end{tikzpicture}
\caption{The construction corresponding to $\boldsymbol{g(\bar{l}, \underline{l})}$.}
\label{fig:g_construction}
\end{figure}
The above $(2\underline{l} + \bar{l})$-simple distribution may not satisfy our desiderata and we may need to modify it. To proceed with our construction, we first need to establish the property that $g(\bar{l}, \underline{l}) \leq p_{\rm peak}$. If $\bar{l} > 0$, this property is satisfied by design, otherwise, the following claim proves this. 
\begin{claim}
    $g(0, \underline{l}) \leq p_{\rm peak}$. 
\end{claim}
\begin{proof}[ sketch]
    The definitions of $\ell^{\rm pk}$ and $\underline{l}$ imply that $p_{\ell^{\rm pk}} \geq \zeta \geq (1+\eps')^{\underline{l}}$. Then, for every $0 \leq k \leq \underline{l}$, we have $p_{\ell^{\rm pk} - k} \geq (1+\eps')^{\underline{l} - k}$, which immediately proves the claim. 
\end{proof}

Having $g(\bar{l}, \underline{l}) \leq p_{\rm peak}$, we modify the construction so that it has a total mass in the peak region of at least $\frac{p_{\rm peak}}{1+\eps'}$. Algorithm \ref{alg:dst-approx} explains this construction. 

\begin{algorithm}[H]
    \caption{Construction of an approximate simple distribution}
    \label{alg:dst-approx}
    \begin{algorithmic}[1]
        \STATE {\it Inputs and definitions:} $\bar{l}, \underline{l} \geq 0$ defined above. For $t = \bar{l} + \underline{l} + 1$, $C_t$ is the construction in \Cref{fig:g_construction} and $S_t = g(\bar{l}, \underline{l})$ is the peak region total mass. Similarly, for $t \leq \bar{l} + \underline{l}$, we inductively construct simple distributions $C_t$ with corresponding peak region mass $S_t$. Thus, $C_t(\ell)$ is the mass, in state $\ell$, under $C_t$, e.g., $C_{\bar{l} + \underline{l} + 1}(\tilde{\ell}_1) = 1$. 
        \FOR{every $t = \bar{l} + \underline{l}, \ldots, 0$}
            \STATE Define $l_t$ as
                \begin{align}\label{eq:added-steps}
                    l_t = \max \left\{{l \in {\bb N}} \; \left | \; S_{t+1} + l  (1+\eps')^t \leq p_{\rm peak}  \right. \right\} \ .
                \end{align}
            \STATE Let $C_t$ be the construction that is identical to $C_{t+1}$, except in that we add $l_t$ flat steps (i.e., $C_t(\ell+1)/C_t(\ell) = 1$) at state $\tilde{\ell}_1 + t$, and shift $C_{t+1}$ to the right, i.e., 
            \begin{align*}
        C_t(\ell) = \begin{cases}
            C_{t+1}(\ell) & \text{if } \ell \leq \tilde{\ell}_1 + t-1 \ , \\ 
            C_{t+1}(\tilde{\ell}_1 + t) & \text{if } \tilde{\ell}_1 + t \leq \ell \leq \tilde{\ell}_1 + t + l_t \ , \\ C_{t+1}(\ell - l_t) & \text{if } \ell > \tilde{\ell}_1 + t + l_t \ . 
        \end{cases}
   \end{align*}
        (See Figure \ref{fig:ssc-construction} for a visual explanation of this step.)
        \ENDFOR
        \STATE {\it Output:} Return $C_0$ as the desired simple distribution. 
        \end{algorithmic}
\end{algorithm}

\begin{figure}[H]
    \centering
    \hspace{-2cm}
\begin{minipage}[b]{0.45\textwidth}
\centering
\begin{tikzpicture}[scale=0.85, transform shape]
\coordinate (O) at (0,0);
\coordinate (A) at (1,1);
\coordinate (B) at (2,2);
\coordinate (C) at (3,3);
\coordinate (P) at (4,4);
\coordinate (D) at (5,3);
\coordinate (E) at (6,2);

\draw[dashed, blue] (O) -- (A);
\draw[ultra thick, blue] (A) -- (B);
\draw[ultra thick, blue] (B) -- (C);
\draw[dashed, blue] (C) -- (P); 
\draw[dashed, blue] (D) -- (E);
\draw[dashed, blue] (P) -- (D);

\node at (O) [circle, fill=blue, inner sep=2pt] {};
\node at (A) [circle, fill=blue, inner sep=2pt] {};
\node at (B) [circle, fill=blue, inner sep=2pt] {};
\node at (C) [circle, fill=blue, inner sep=2pt] {};
\node at (D) [circle, fill=blue, inner sep=2pt] {};
\node at (E) [circle, fill=blue, inner sep=2pt] {};
\node at (P) [circle, fill=blue, inner sep=2pt] {};

\node at (O) [above left, blue] {$1$};
\node at (B) [above left, blue] {$(1+\eps')^t$};
\node at (P) [above, blue] {$(1+\eps')^{\underline{\ell} + \bar{\ell}}$};
\node at (E) [above right, blue] {$(1+\eps')^{\underline{\ell}}$};
\end{tikzpicture}
\end{minipage}
\raisebox{8\height}{\color{blue} \Large$\quad \quad \Longrightarrow$}
\begin{minipage}[b]{0.45\textwidth}
\centering
\begin{tikzpicture}[scale=0.85, transform shape]
\coordinate (O) at (0,0);
\coordinate (A) at (1,1);
\coordinate (B) at (2,2);
\coordinate (BP) at (3,2);
\coordinate (BPP) at (4,2);
\coordinate (C) at (5,3);
\coordinate (P) at (6,4);
\coordinate (D) at (7,3);
\coordinate (E) at (8,2);

\draw [decorate, decoration={brace, amplitude=5pt, raise=2mm}, blue, yshift=2pt]
(B) -- (BPP) node [above, midway, yshift=5mm] {\color{blue} ${\ell_t}$};

\draw[dashed, blue] (O) -- (A);
\draw[ultra thick, blue] (A) -- (B);
\draw[ultra thick, blue] (B) -- (BP);
\draw[ultra thick, blue] (BP) -- (BPP);
\draw[ultra thick, blue] (BPP) -- (C);
\draw[dashed, blue] (C) -- (P); 
\draw[dashed, blue] (D) -- (E);
\draw[dashed, blue] (P) -- (D);

\node at (O) [circle, fill=blue, inner sep=2pt] {};
\node at (A) [circle, fill=blue, inner sep=2pt] {};
\node at (B) [circle, fill=blue, inner sep=2pt] {};
\node at (BP) [circle, fill=blue, inner sep=2pt] {};
\node at (BPP) [circle, fill=blue, inner sep=2pt] {};
\node at (C) [circle, fill=blue, inner sep=2pt] {};
\node at (D) [circle, fill=blue, inner sep=2pt] {};
\node at (E) [circle, fill=blue, inner sep=2pt] {};
\node at (P) [circle, fill=blue, inner sep=2pt] {};

\node at (O) [above left, blue] {$1$};
\node at (B) [above left, blue] {$(1+\eps')^t$};
\node at (P) [above, blue] {$(1+\eps')^{\underline{\ell} + \bar{\ell}}$};
\node at (E) [above right, blue] {$(1+\eps')^{\underline{\ell}}$};
\end{tikzpicture}
\end{minipage}
\caption{The construction $\boldsymbol{C_{t}}$ (right) from $\boldsymbol{C_{t+1}}$ (left).}
\label{fig:ssc-construction}
\end{figure}
Using \Cref{alg:dst-approx}, we claim that the construction at the last iteration, $C_0$, is our desired distribution satisfying properties (a)-(d) above and a polynomial upper bound on $K_0$, the size of peak region in $C_0$. Clearly, our desired properties (a), (d), and (e) are satisfied. Moreover, since $C_0$ puts a mass of $(1+\eps')^{\underline{l}}$ on the last state of the peak region, the definition of $\underline{l}$ entails that $\frac{\zeta}{1+\eps'} \leq \tilde{p}_{\ell + K_0} \leq \zeta$, which fulfils property (c). It remains to prove property (b). Our design ensures that $\tilde{p}_{\rm peak} \leq p_{\rm peak}$ and we now prove $\tilde{p}_{\rm peak} \geq (1-\eps'){p_{\rm peak}}$. Indeed, the definition of $l_0$ for the last step $t = 0$ implies that $S_0 + 1 > p_{\rm peak}$. Hence, using the subcase hypothesis $1 \leq \eps' p_{\rm peak}$ gives $S_0 + \eps' p_{\rm peak} > p_{\rm peak}$, which is the desired statement. 

We now establish the upper bound for $K_0$. First, we bound the additions in each step of \Cref{alg:dst-approx}:
\begin{claim}\label{clm:lt_bound}
    For every $t = \bar{l} + \underline{l}, \ldots, 0$, we have $l_t \leq 2+\eps'$. 
\end{claim}
\begin{proof}
    We first consider $t = \bar{l} + \underline{l}$. By definition of $\bar{l}$, we have $g(\bar{l}, \underline{l}) + (1+\eps')^{t + 1} + (1+\eps')^{t} > p_{\rm peak}$. Since $(1+\eps')^t \cdot (2+\eps') > (1+\eps')^{t + 1} + (1+\eps')^{t}$ and $g(\bar{l}, \underline{l}) = S_{t+1}$, definition \eqref{eq:added-steps} implies that $l_t \leq 2+\eps'$. 
    For $t < \bar{l} + \underline{l}$, we have $S_{t+1} + (1+\eps')^{t+1} > p_{\rm peak}$, which immediately gives $l_t \leq \frac{(1+\eps')^{t+1}}{(1+\eps')^t} = 1+\eps' \leq 2+\eps'$.
\end{proof}
By \Cref{clm:lt_bound}, we get $K_0 \leq (\underline{l} + \bar{l} + 1)(1 + 2+\eps') + \bar{l} \leq (\underline{l} + \bar{l} + 1)(4+\eps') $. To obtain a polynomial upper bound on $K_0$, we recall inequality \eqref{ineq:zeta-bound} to argue that $p_{\rm peak} \leq \frac{1}{\rho_{\tilde{\ell}_1}} \leq \frac{1}{h \kappa}$ and use the fact that $\underline{l} + \bar{l} \leq \log_{1+\eps'} p_{\rm peak}$. Therefore, we get \[ K_0 \leq \left(\log_{1+\eps'} \left(\frac{\tau_{\max}}{\kappa\eps' \tau(\pi)}\right) + 1\right) \cdot (4+\eps') \ , \] which is a polynomial bound for every fixed $\eps'$. Note that if $\zeta \leq 1$, a similar design works by using $\tilde{\ell}_2$ as reference with $\tilde{p}_{\tilde{\ell}_2} = 1$ and considering the ratio of probabilities in the peak region with a reverse construction. Since the argument is analogous, we avoid repetition and move to the next subcase.

{\it Subcase 1.2: $p_{\rm peak} \leq \frac{1}{\eps'}$.} First, note that the subcase hypothesis implies $\ell^{\rm pk} - \tilde{\ell}_1 \leq \frac{1}{\eps'}$ since $p_{\ell} \geq p_{\tilde{\ell}_1}$ for every $\tilde{\ell}_1 \leq \ell \leq \ell^{\rm pk}$. Recall the argument in subcase 1.1. We can make the exact same argument but using $\tilde{\ell}_2$ as our reference point, i.e., we can use another scaled version of $\rho$, called $p'$, in which $p'_{\tilde{\ell}_2} = 1$. Then, we can use the construction of simple distributions in Algorithm \ref{alg:dst-approx}. If $p'_{\rm peak} \geq \frac{1}{\eps'}$, using subcase 1.1 completes the proof. Otherwise, similar to before, we must have $\tilde{\ell}_2 - \ell^{\rm pk} \leq \frac{1}{\eps'}$. Combining it with the earlier bound, we obtain $\tilde{\ell}_2 - \tilde{\ell}_1 \leq \frac{2}{\eps'}$, which violates our initial assumption that $\tilde{\ell}_2 - \tilde{\ell}_1 \geq K \geq \frac{2}{\eps'}$ and thus, this case cannot occur.

\paragraph{\bf Case 2: $\rho_{\rm left} < h \leq \rho_{\rm right}$.} The idea here is to artificially add mass to the left region so that it has a probability of exactly $h$ while we keep the ratios $\frac{\tilde{\rho}_\ell}{\tilde{\rho}_{\ell - 1}}$ within peak and right regions the same. Clearly, this operation observes the required accuracy level. Consequently, we can use the construction in case 1 to obtain the desired simple distribution $\tilde{\rho}$ that satisfies
\begin{align*}
    (1-\eps')(\rho_v - h) \leq \tilde{\rho}_v \leq \frac{\rho_v + h}{1-\eps'}
\end{align*}
for every $v \in \{{\rm left}, {\rm peak}, {\rm right}\}$.

To increase the probability of the left region, we scale down probabilities for every $\ell \geq \tilde{\ell}_1$ by a factor $(1-(h-\rho_{\rm left}))$ and scale up the probabilities for every $\ell < \tilde{\ell}_1$ such that the sum of probabilities is one. The policy that induces this modified distribution is identical to the original policy with extra enforced abandonments at state $\tilde{\ell}_1$. Now, we can apply the argument in case 1 and the combination of two approximations gives our desired guarantees. 

\paragraph{\bf Case 3: $\rho_{\rm right} < h \leq \rho_{\rm left}$ or $\rho_{\rm left}, \rho_{\rm right} < h$.} The proof for the former is similar to case 2 and thus omitted. The latter is trivial since we can ignore left and right regions and focus only on the peak region. 


\qedsymbol

\subsection{Tentative extension to networks with $n \geq 1$}  \label{subsec:tentative}

For $n>1$, a natural generalization of $DLP(\bar{\ell})$ is as follows
\begin{align}
(\widetilde{DLP}(\bar{\ell})) \quad &&\min_{\boldsymbol{x}} &&& \sum_{i=1}^n\sum_{\ell \leq \bar{\ell}} \sum_{\MyAtop{M \subseteq [m]:}{j\in M}}  \gamma_j  c_{j}  x_{i,M}^{\ell} \notag \\
&&\text{s.t.}  &&& \left(x_{i,M}^\ell\right)_{M, \ell}
\in {\cal B}_i(\bar{\ell}) \ ,   && \forall i \in [n] \notag\\ 
&&  &&&\sum_{i=1}^n\sum_{1 \leq \ell \leq \bar{\ell}}\sum_{M \subseteq [m]} \gamma(M)  x_{i,M}^\ell \geq \tau^* \ . && \label{cons:multi_dlp_throughput} \\
&&  &&&\sum_{i=1}^n\sum_{1 \leq \ell \leq \bar{\ell}}\sum_{\MyAtop{M \subseteq [m]:}{j\in M}}  x_{i,M}^\ell \leq 1 \ . && \forall j \in [m]\label{cons:no-contention} \ ,
\end{align} 
where the intuition around the variables $x_{i,M}^\ell$ is similar for every supplier type $i$. Constraint~\eqref{cons:no-contention} adds a matching constraint on the customer side---indicating that an arriving customer cannot be served by more than one supplier on average. A simple extension of Lemma~\ref{prop:dlp1_benchmark} shows that this LP is a valid relaxation of the general dynamic matching problem. Moreover, using a similar primal-dual analysis, we can provide an FPTAS for solving the extended $DLP(+\infty)$.

\begin{proposition}\label{prop:multi_dlp_fptas}
There exists an FPTAS for approximating $\widetilde{DLP}(+\infty)$. Specifically, for each $\eps \in (0, 1)$, it is possible to compute a $(1+\eps)$-factor of the optimal solution of $\widetilde{DLP}(+\infty)$ in time $\poly(|{\cal I}|, \log(\frac{1}{\tau^*}), \frac{1}{\eps})$, where $|{\cal I}|$ is the size of the input.
\end{proposition}
The algorithm is an extension of our FPTAS in the single-queue case. A nested property holds for active matching sets in the optimal dual solution, as a direct analog to Section~\ref{ssec:dual}. However, the dual formulation now represents weakly-coupled MDPs, each corresponding to a different queue $i \in [n]$ of suppliers. The match reduced cost is adjusted to reflect the shadow price of the added supplier constraints~\eqref{cons:no-contention}. Considering $K$-bounded policies as in Lemma~\ref{lem:truncate_distrib}, we solve $\widetilde{DLP}(K)$ using an efficient separation oracle for the dual formulation. One subtle difference is that the nested family of matching sets per supplier is not known in advance due to the shadow price adjustment from the new constraints. However, such constraints can be separated in polynomial time. We formally discuss this FPTAS below.

While Proposition~\ref{prop:multi_dlp_fptas} gives an FPTAS for $\widetilde{DLP}(+\infty)$, this LP is only a relaxation of the dynamic matching problem when $n > 1$. The counter-example in~\cite{kessel2022stationary}, App. C still shows a constant gap with the optimum, so this LP is too coarse to develop an approximation scheme. \Cref{sec:constant_ptas} refines our LP approximation for networks of suppliers.

\subsubsection*{An FPTAS for solving $\boldsymbol{\widetilde{DLP}}$.} Below, we briefly explain how to use our $DLP$ machinery to devise an FPTAS for the extended LP. 

\paragraph{Dual of relaxed $\widetilde{DLP}$.} To be able to use duality for this infinite-dimensional LP, we use a similar relaxation as in \eqref{ineq:flow_balance_prime} that does not change the objective value, and write the dual LP:
\begin{align}
    \max_{\boldsymbol{\alpha}, \boldsymbol{\beta}, \boldsymbol{\delta}, \theta} \quad & \sum_{i \in [n]}\alpha_i - \sum_{j \in [m]} \beta_j + \theta \tau^* &&   \notag \\
    \text{s.t.} \quad 
    & -\lambda  \delta^{\ell+1}_i + 
    \Bigl(\ell + \sum_{j \in M} \gamma_j\Bigr)  \delta^{\ell}_i + \alpha_i  \leq \sum_{j \in M} \gamma_j   \left(c_{i,j} - \theta + \frac{\beta_j}{\gamma_j}\right) \ , && \forall i \in [n], \forall M \subseteq [m], \forall \ell \in \mathbb{N}^+ \label{cons:multi_dual_bellman} \\ & -\lambda  \delta^1_i + \alpha_i \leq 0 \ , && \forall i \in [n] \label{cons:multi_delta1} \\
    & \delta_i^\ell \leq 0 \ , && \forall i \in [n], \forall \ell \in {\bb N}^+ \label{cons:delta_pos} \\ 
    & \beta_j \leq 0 \ , && \forall j \in [m]  \label{cons:beta_pos} \\
    & \theta \geq 0 \ . \label{cons:theta_pos}
\end{align}
This LP is very similar to the dual of single-queue LP, except in that the reduced cost of a match includes the term $\frac{\beta_j}{\gamma_j}$ which is the shadow price corresponding to constraint \eqref{cons:no-contention}. In fact, for any fixed $\boldsymbol{\beta}$ and $\theta$, the optimal value of $\alpha_i$-s are decoupled. It is straightforward to verify that \Cref{lem:dual-threshold-increasing} holds verbatim for each $i$, provided that the reduced costs are replaced with $c'_{i,j} = c_{i,j} - \theta + \frac{\beta_j}{\gamma_j}$. 

\paragraph{Polynomial truncation.} Recall that a key component in efficiently approximating $DLP$ is \Cref{lem:truncate_distrib}, which demonstrates that polynomially bounded policies can closely approximate the optimal policy. Since this lemma holds only for single-queue instances, we now combine it with the weak coupling structure to efficiently approximate the dual of $\widetilde{DLP}$. 

Any optimal solution $\boldsymbol{x}$ of $\widetilde{DLP}$ induces $n$ single-queue monotone policies $\pi_1, \dots, \pi_n$, where each $\pi_i$ matches only type-$i$ suppliers and is monotone by (extended) \Cref{lem:dual-threshold-increasing}. While each policy $\pi_1, \dots, \pi_n$ is individually feasible in the actual stochastic system, their superposition may not be, as the contention constraint \eqref{cons:no-contention} holds only in expectation. Moreover, we have $\widetilde{DLP}^* = \sum_{i \in [n]} c(\pi_i)$, where $c(\pi_i)$ denotes the expected average cost rate of $\pi_i$ in the single-queue instance for queue $i$.

Let $S$ represent the set of supplier types $i$ that satisfy $\tau(\pi_i) \geq {\eps \tau^*}/(2n)$. Clearly, $\sum_{i \in S} \tau(\pi_i) \geq (1 - \eps/2)\tau^*$. Furthermore, by \Cref{lem:truncate_distrib}, there exists a $O(\frac{1}{\eps}\log(\frac{\tau_{\max}}{\tau(\pi_i)}))$-bounded policy $\tilde{\pi}_i$ with $(1-\eps/2)\tau(\pi_i) \leq \tau(\tilde{\pi}_i) \leq \tau(\pi_i)$, for each $i \in S$. We can thus provide a feasible solution ${\boldsymbol{\tilde{x}}}$ for $\widetilde{DLP}(K)$ with $K = O(\frac{1}{\eps}\log(\frac{n\tau_{\max}}{\eps \tau^*}))$ such that it has an objective value less than $\widetilde{DLP}^*$ and satisfies \eqref{cons:multi_dlp_throughput} up to an $\eps$-factor: For $i \in S$, let $\tilde{x}_{i,M}^\ell$ be the stationary probability that, under policy $\tilde{\pi}_i$, there are $\ell$ type-$i$ suppliers available and the policy commits to matching customers from set $M$ to a type-$i$ supplier. For $i \not\in S$, we use a policy that makes no matches, and define $\tilde{x}_{i,\emptyset}^\ell$ as the stationary probability of being in state $\ell$. All other entries of $\boldsymbol{\tilde{x}}$ are set to 0. 

It is easy to see that $\boldsymbol{\tilde{x}}$ satisfies the throughput constraint \eqref{cons:multi_dlp_throughput} up to an $\eps$-factor. The other feasibility constraints are immediate and we can verify that the cost (i.e. $\widetilde{DLP}$'s value) under $\boldsymbol{\tilde{x}}$ is smaller than the one under $\boldsymbol{x}$. Therefore, $\widetilde{DLP}(K)$ (with the adjusted $(1-\eps)\tau^*$ throughput constraint), that includes a polynomial range of queue lengths, is feasible. However, the issue of having exponentially many matching sets $M \subseteq [m]$ remains. Next, we solve this issue by devising an efficient separation orcale using the extension of \Cref{lem:dual-threshold-increasing}.



\paragraph{Efficient separation oracle for the dual of $\widetilde{DLP}(K)$.}
Our separation oracle takes as input a candidate $\boldsymbol{\alpha}, \boldsymbol{\beta}, \boldsymbol{\delta}, \theta$, and certifies either all constraints \eqref{cons:multi_dual_bellman}-\eqref{cons:theta_pos} are satisfied, or returns a separating hyperplane corresponding to one of those violated constraints. 

 The only non-trivial constraint is \eqref{cons:multi_dual_bellman} where the challenge is that there are exponentially many $M$ subsets that correspond to every fixed $i \in [n]$ and $\ell \in [\bar{\ell}]_0$. It is immediate to see that \Cref{lem:dual-threshold-increasing} holds for $\widetilde{DLP}(K)$ too. Thus, by \Cref{lem:dual-threshold-increasing}, we only need to consider $M \subseteq [m]$ if it satisfies $\hat{M}_i^\ell \subseteq M \subseteq \hat{M}_i^\ell \cup \tilde{M}_i^\ell$. To see this, we define \[\textsf{slack}_i(M, \ell) = \sum_{j \in M} \gamma_j   \left(c_{i,j} - \theta + \frac{\beta_j}{\gamma_j}\right) - \left(  - \lambda_i  \delta_i^{\ell+1} + 
    \Bigl(\ell + \sum_{j \in M} \gamma_j\Bigr)  \delta_i^{\ell} + \alpha_i \right) \] to be the difference between RHS and LHS of constraint \eqref{cons:multi_dual_bellman} for $M, \ell, i$. The structure of the constraint implies that if $M \subseteq [m]$ satisfies $\hat{M}_i^\ell \subseteq M \subseteq \hat{M}_i^\ell \cup \tilde{M}_i^\ell$ and $M' \subseteq [m]$ does not satisfy $\hat{M}_i^\ell \subseteq M' \subseteq \hat{M}_i^\ell \cup \tilde{M}_i^\ell$, we have $\textsf{slack}_i(M, \ell) < \textsf{slack}_i(M', \ell)$. Moreover, every $M \subseteq [m]$ that satisfies $\hat{M}_i^\ell \subseteq M \subseteq \hat{M}_i^\ell \cup \tilde{M}_i^\ell$, yields the same slack value. Therefore, we can efficiently construct $\hat{M}_i^\ell$ for every value of $\ell \in [\bar{\ell}], i\in [n]$ and certify $\textsf{slack}_i(\hat{M}^\ell, \ell) \geq 0$. If any constraint is violated, we return the corresponding hyperplane. Otherwise, we have found a feasible point for the dual of $\widetilde{DLP}(K)$. In conclusion, we have constructed an efficient separation oracle. 

    Note that each iteration of the ellipsoid method requires $O(nm\bar{\ell})$ operations. Since the number of variables of the dual is $O(n\bar{\ell})$, the ellipsoid algorithm runs for $O(n^6\bar{\ell}^6\log(nm\bar{\ell}) |{\cal I}|)$ iterations \citep{bertsimas1997introduction}. Hence, the time complexity of solving the dual is $O(mn^7\bar{\ell}^7\log(nm\bar{\ell})|{\cal I}|)$. 

\paragraph{Solving $\widetilde{DLP}(K)$ from its dual.} 
Suppose that we have access to a solution $\boldsymbol{{\delta}}, \boldsymbol{{\alpha}}, \boldsymbol{\beta}, {\theta}$ which is optimal for the dual LP. First, it is clear by complementary slackness that if the dual constraint \eqref{cons:multi_dual_bellman} is loose for some $i \in [n], \ell \in [\bar{\ell}], M \subseteq [m]$, we must have $x_{i, M}^\ell = 0$. However, the other direction is not necessarily true: a tight constraint does not necessarily imply $x_{i,M}^\ell > 0$. Thus, it is not clear which subsets $M \subseteq [m]$ with $\hat{M}_i^\ell \subseteq M \subseteq \hat{M}_i^\ell \cup \tilde{M}_i^\ell$ in \Cref{lem:dual-threshold-increasing} are active, i.e., $x_{i, M}^\ell > 0$. 


Nevertheless, the dual LP is has \(O(n\bar{\ell})\) variables, corresponding to the constraints considered during the ellipsoid method. Thus, an optimal basic feasible solution of the dual is determined by a polynomial-sized subset of the tight constraints. These can be identified by pruning redundant constraints from the working set generated by the ellipsoid method. In fact, it suffices to restrict attention to the subset of tight constraints necessary to define the primal BFS, which can be extracted from the working set of constraints used during the ellipsoid method. Consequently, we solve the primal LP by considering only the constraints corresponding to this subset. By complementary slackness, this restricted primal LP yields a feasible and optimal dual solution. In conclusion, this restricted primal LP has a polynomial number of constraints and variables, making it solvable in polynomial time. Consequently, both the primal and dual LPs can be solved in \(O(mn^7\bar{\ell}^7\log(nm\bar{\ell})|{\cal I}|)\) time.

\section{Additional Material from Section \ref{sec:constant_ptas}}
\subsection{Proof of \Cref{lem:NLP_feasible}}\phantomsection\label{prf:NLP_feasible}

Consider a stationary policy ${\pi}$ that satisfies ${\tau}({\pi}) \geq \tau^*$. Let $\bar{\ell} = \frac{1}{\eps}(\frac{1}{\delta})^{\kappa} + \{u\}^+ + 1$  for $u = \lceil (\log \frac{1}{\eps})^{-1} \cdot \log \frac{n^2}{(1-\eps)\eps^2 \delta^\kappa } \rceil$ where we recall that $\kappa \leq \min\{1/\eps, n\}+1$ is the constant used in the definition of short and long queues. We recall that a policy is ${\ell}$-bounded for a set of queues $S$ if it discards any new supplier of a type in $S$ who arrive when their queue length is already at ${\ell}$. 

We prove that there exists a near-optimal policy that is $\bar{\ell}$-bounded with respect to short queues. First, we consider the case where $\tau^*$ is large:

\begin{claim}\label{clm:large_tau}
    Suppose $\tau^* > \eps^2/n$. There exists a policy $\pi^{\rm trunc}$ that is $\bar{\ell}$-bounded for \underline{short} queues and satisfies $\tau(\pi^{\rm trunc}) \geq (1-\eps) \cdot \tau^*$ and $c(\pi^{\rm trunc}) \leq c(\pi)$. 
\end{claim}
\begin{proof}
    
We define ${\tilde{\pi}}$ to be the policy that mimics ${\pi}$, except that it discards any short supplier of type $i\in {\cal S}^{\rm short}$ who arrive when queue $i$'s length is at least $\bar{\ell}$. However, when the queue length is at least $\bar{\ell}$, $\tilde{\pi}$ simulates a fake stream of type-$i$ suppliers arriving with rate $\lambda_i$. Going forward, $\tilde{\pi}$ treats these fake suppliers similarly to real ones. Namely, $\tilde{\pi}$ implements the same matches as ${\pi}$. 
    
Note that the steady-state distributions and match rates of $\tilde{\pi}$ and $\pi$ are identical when we count the matches of fake suppliers too. It suffices to show that fake suppliers constitute only an $\eps$ fraction of all matches. Since the queue lengths under ${\pi}$ in the steady-state (denoted by ${\cal L}^{{\pi}} (\infty)$) are stochastically dominated by queue lengths under a policy that never makes any matches (denoted by ${\cal L}^{\bot}(\infty)$), we observe that 
\begin{eqnarray}
{\tau}({\pi})-{\tau}(\tilde{\pi}) &\leq& \sum_{i\in {\cal S}^{\rm vs}}\lambda_i \pr{{\cal L}^{{\pi}}_i(\infty) \geq \bar{\ell} } \notag \\
&\leq& \sum_{i\in {\cal S}^{\rm vs}}\lambda_i \pr{{\cal L}^{\bot}_i(\infty) \geq \bar{\ell} } \notag \\
&\leq& n \cdot \frac{\tau^*}{\eps} \cdot \max_{i \in {\cal S}^{\rm vs}} \pr{{\cal L}_i^{\bot} (\infty) \geq \bar{\ell}}  \notag \\
&\leq& n \cdot \frac{\tau^*}{\eps} \cdot  \eps^u \frac{\eps}{1-\eps} \notag  \\
&\leq& {\eps \tau^*} \ , \notag
\end{eqnarray}
where the third inequality follows from having $\lambda_i \leq \frac{\tau^*}{\eps}$ and $n_s \leq n$. The fourth inequality uses the fact that ${\cal L}_i^{\bot}(\infty)$ is a Poisson distribution with parameter $\lambda_i$; a simple calculation shows that the stationary distribution of ${\cal L}^{\bot}(\infty)$ decays geometrically beyond $\frac{1}{\eps}\lambda_i$. Concretely, for $\ell > \frac{1}{\eps}\lambda_i$, we have
\[
    \frac{ \pr{{\cal L}^{\bot}_i(\infty) = \ell+1}}{ \pr{{\cal L}^{\bot}_i(\infty) = \ell}} = \frac{\lambda_i}{\ell + 1} \leq \frac{\lambda_i}{1/\eps \cdot \lambda_i} = \eps  \ ,
\]
which implies 
\[
     \pr{{\cal L}^{\bot}_i(\infty) \geq \bar{\ell}} \leq  \pr{{\cal L}^{\bot}_i(\infty) = \frac{\lambda_i}{\eps}} \cdot \eps^{\{u\}^+} \cdot (\eps + \eps^2 + \cdots) \leq \eps^u \frac{\eps}{1-\eps} \ .
\]
Finally, the last inequality uses the definition of $u$ along with
\begin{align}
    n \cdot \lambda_i \cdot \eps^u \frac{\eps}{1-\eps} \leq \eps\tau^* \Leftrightarrow \eps^u \leq (1-\eps)\tau^* /(\lambda_in) \Leftrightarrow u \geq \frac{\log((1-\eps)\tau^* /(\lambda_in) )}{\log(\eps)} = \frac{\log(\frac{n\lambda_i}{(1-\eps) \tau^*})}{\log(\frac{1}{\eps})} \ . \label{ineq:u_deriv}
\end{align} Therefore, observing $\lambda_i \leq \frac{1}{\delta^\kappa}$ and $\tau^* \geq {\eps^2}/{n}$, we have 
\[
    u = \frac{\log \frac{n^2}{\eps^2 \delta^ \kappa (1-\eps)}}{\log \frac{1}{\eps}} \geq \frac{\log(\frac{n\lambda_i}{(1-\eps) \tau^*})}{\log(\frac{1}{\eps})} \ ,
\]
as desired. 

Notice that if we exclude fake suppliers from our queues, each short queue has a length of at most $\bar{\ell}$ at every point in time. Hence, $\tilde{\pi}$ is a policy that that is $\bar{\ell}$-bounded for short queues by and satisfies $\tau(\tilde{\pi}) \ge (1-\eps)\tau^*$ and $c(\tilde{\pi}) \leq c^*$. 
\end{proof}

On the other hand, when $\tau^*$ is small, we can in fact show that there exists a near-optimal policy which is $\bar{\ell}$-bounded for all queues. The proof in Appendix~\ref{prf:small_tau} follows from a technical argument that leverages the structure of the problem when $\tau^*$ is small. 
\begin{claim}\label{prop:small_tau}
    Suppose $\tau^* \leq \eps^2/n$. There exists a policy $\pi^{\rm trunc}$ that is $\bar{\ell}$-bounded for \underline{all} queues and satisfies $\tau(\pi^{\rm trunc}) \geq (1-\eps) \cdot \tau^*$ and $c(\pi^{\rm trunc}) \leq c(\pi)$. 
\end{claim}

By Claims \ref{clm:large_tau} and \ref{prop:small_tau}, there exists a stationary policy $\pi^{\rm trunc}$ that is $\bar{\ell}$-bounded for short queues and satisfies $\tau(\pi^{\rm trunc}) \geq (1-\eps)\tau^*$ and $c(\pi^{\rm trunc}) \leq c(\pi)$. We now prove that the stationary distribution induced by ${\pi}^{\rm trunc}$ satisfies the constraints of the $NLP(\bar{\ell)}$ and attains a cost rate of at most $c(\pi)$; this statement proves that $NLP(\bar{\ell})^* \leq c(\pi)$. In the following, we use the convention that ${\bs{M}} \in {\cal D}({\cal S}^{\rm short})$ is an assignment of $[m]$ into $n_s$ matching sets of short queues. Similarly, ${\bs{M'}} \in {\cal D}({\cal S})$ is an assignment into $n$ matching sets of all~queues. 

Consider the Markov chain consisting of all queues and let $z_{\bs{M'}}^{\boldsymbol{\ell'}}$ be the stationary probability that the state of the queues is $\boldsymbol{\ell'} \in \nat^{n}$ and $\pi^{\rm trunc}$ decides on ${\bs{M'}} \in {\cal D}({\cal S})$ as its {current} matching sets. 
Now, we define restriction of ${\bs{M'}}$ and $\boldsymbol{\ell'}$ to short queues: ${\bf M'_{|_s}}$ and $\boldsymbol{\ell'_{|_s}}$ are, respectively, the matching sets and states of short queues according to ${\bs{M'}}$ and $\boldsymbol{\ell'}$. Consequently, we define the marginals $(x_{\bs{M}}^{\boldsymbol{\ell}})$ for $\boldsymbol{\ell} \in [\bar{\ell}]_0^{n_s}$ and ${\bs{M}} \in {\cal D}({\cal S}^{\rm short})$, corresponding to short queues, to be
\[
x_{\bs{M}}^{\boldsymbol{\ell}} = \sum_{\substack{{\bs{M'}} \in {\cal D}({\cal S}):\\ {\bs{M}} = {\bf M'_{|_s}}}} \sum_{\substack{\boldsymbol{\ell'} \in \nat^{n}: \\ \boldsymbol{\ell'_{|_s}} = \boldsymbol{\ell}}} z_{\bs{M'}}^{\boldsymbol{\ell'}} \ .
\]

Similarly, we define the marginals corresponding to long queues. However, we use the static (independent of queue lengths) marginals. Namely, for every $i \in {\cal S}^{\rm long}, j \in [m]$, we let 
\[
y_{i,j} = \sum_{\substack{{\bs{M'}} \in {\cal D}({\cal S}):\\ j \in M'_i}} \sum_{\substack{\boldsymbol{\ell'} \in \nat^{m}}} z_{\bs{M'}}^{\boldsymbol{\ell'}} \ .
\]
Next, we verify $(x_{\bs{M}}^{\ell})_{{\bs{M}}, \ell} \in {\cal B}({\cal S}^{\rm short}, \bar{\ell})$. Consider some $\ell \in [\bar{\ell}]_0^{n_s}$ and let $L$ be the set of all states $\boldsymbol{\ell'} \in \nat^{n}$ such that $\boldsymbol{\ell'_{|_s}} = \ell$. It is well known that for any Markov chain and a set of states $\chi$, the stationary probability satisfies this property: the inflow to $\chi$ is equal to the outflow; the global balance condition can be viewed as its special case where $\chi$ consists of a single state. Thus, we invoke this fact for the Markov chain of all queues to derive the inflow to $\chi = L$:
\begin{align*}
& \sum_{\substack{\boldsymbol{\ell'} \in \nat^n: \\ \boldsymbol{\ell'_{|_s}} = \boldsymbol{\ell}}} \sum_{{\bs{M'}} \in {\cal D}({\cal S})} \left( \sum_{\substack{i \in {\cal S}^{\rm short}:\\ \ell'_i \geq 1}} \lambda_i   z_{\bs{M'}}^{{\boldsymbol{\ell}} -e_i} + \sum_{i \in {\cal S}^{\rm short}} (\gamma(M_i) +  (\ell_i+1))  z_{\bs{M'}}^{{\boldsymbol{\ell}} + e_i} \right) \\  & \quad = \sum_{\substack{{\bs{M}} \in {\cal D}({\cal S}^{\rm short})}}  \sum_{\substack{i \in {\cal S}^{\rm short}:\\ \ell'_i \geq 1}} \lambda_i \left( \sum_{\substack{{\bs{M'}} \in {\cal D}({\cal S}):\\ {\bf M'_{|s}} = {\bs{M}}}} \sum_{\substack{\boldsymbol{\ell'} \in \nat^n: \\ \boldsymbol{\ell'_{|_s}} = \boldsymbol{\ell}}}   z_{\bs{M'}}^{{\boldsymbol{\ell}} -e_i}\right)  \\ & \quad \quad + \; \sum_{i \in {\cal S}^{\rm short}} (\gamma(M_i) +  (\ell_i+1)) \left( \sum_{\substack{{\bs{M'}} \in {\cal D}({\cal S}):\\ {\bf M'_{|s}} = {\bs{M}}}} \sum_{\substack{\boldsymbol{\ell'} \in \nat^n: \\ \boldsymbol{\ell'_{|_s}} = \boldsymbol{\ell}}}  z_{\bs{M'}}^{{\boldsymbol{\ell}} + e_i} \right)  \\ & \quad = \sum_{{\bs{M}} \in {\cal D}({\cal S}^{\rm short})} \left( \sum_{\substack{i \in L:\\ \ell_i \geq 1}} \lambda_i   x_{\bs{M}}^{{\boldsymbol{\ell}} -e_i} + \sum_{i \in L} (\gamma(M_i) +  (\ell_i+1))  x_{\bs{M}}^{{\boldsymbol{\ell}} + e_i} \right) \ ,
\end{align*}
which is the LHS in equation \eqref{eq:multivariate_flow_balance}. We can similarly show that the outflow from $\chi = L$ is equal to the RHS of expression \eqref{eq:multivariate_flow_balance}:
\begin{align*}
& \sum_{\substack{\boldsymbol{\ell'} \in \nat^n: \\ \boldsymbol{\ell'_{|_s}} = \boldsymbol{\ell}}} \sum_{{\bs{M'}} \in {\cal D}({\cal S})} z_{\bs{M'}}^{{\boldsymbol{\ell}}} \left( \sum_{\substack{i \in {\cal S}^{\rm short}}} \lambda_i \cdot \mathbbm{I}[\ell_i < \bar{\ell}] + \sum_{i \in {\cal S}^{\rm short}} (\gamma(M_i) \cdot \mathbbm{I}[\ell_i \geq 1] + \ \ell_i)  \right) \\ 
& \quad = \sum_{\substack{{\bs{M}} \in {\cal D}({\cal S}^{\rm short})}} \, \sum_{\substack{\boldsymbol{\ell'} \in \nat^n: \\ \boldsymbol{\ell'_{|_s}} = \boldsymbol{\ell}}}   \, \sum_{\substack{{\bs{M'}} \in {\cal D}({\cal S}):\\ {\bf M'_{|s}} = {\bs{M}}}} z_{\bs{M'}}^{{\boldsymbol{\ell}}}  \left( \sum_{\substack{i \in {\cal S}^{\rm short}}} \lambda_i \cdot \mathbbm{I}[\ell_i < \bar{\ell}] + \sum_{i \in {\cal S}^{\rm short}} (\gamma(M_i) \cdot \mathbbm{I}[\ell_i \geq 1] + \ \ell_i)  \right)  \\ & \quad = \sum_{\substack{{\bs{M}} \in {\cal D}({\cal S}^{\rm short})}} \, \left( \sum_{\substack{i \in {\cal S}^{\rm short}}} \lambda_i \cdot \mathbbm{I}[\ell_i < \bar{\ell}] + \sum_{i \in {\cal S}^{\rm short}} (\gamma(M_i) \cdot \mathbbm{I}[\ell_i \geq 1] + \ \ell_i)  \right) \cdot  \sum_{\substack{\boldsymbol{\ell'} \in \nat^n: \\ \boldsymbol{\ell'_{|_s}} = \boldsymbol{\ell}}}   \, \sum_{\substack{{\bs{M'}} \in {\cal D}({\cal S}):\\ {\bf M'_{|s}} = {\bs{M}}}} z_{\bs{M'}}^{{\boldsymbol{\ell}}}   \\ & \quad = \sum_{{\bs{M}} \in {\cal D}(S)} x_{\bs{M}}^{\boldsymbol{\ell}}  \left( \sum_{i \in S} \sum_{j \in M_i} \gamma_j \cdot \mathbb{I}[\ell_i \geq 1] + \lVert{\boldsymbol{\ell}}\rVert_1 + \sum_{i \in S} \lambda_i \cdot {\bb I}[\ell_i < \bar{\ell}] \right) \ ,
\end{align*}
which is the RHS of equation \eqref{eq:multivariate_flow_balance}. Given that entries of $\boldsymbol{z}$ sum up to 1, by the law of total probability, we infer that $(x_{\bs{M}}^{\ell})_{{\bs{M}}, \ell} \in {\cal B}({\cal S}^{\rm short}, \bar{\ell})$. 

Note that the average match rate to queue $i \in {\cal S}^{\rm long}$ cannot be more than $\lambda_i$, which translates to 
\[
    \sum_{j \in {\cal C}}\sum_{\substack{{\bs{M'}} \in {\cal D}({\cal S}):\\ j \in M'_i}} \sum_{\substack{\boldsymbol{\ell'} \in \nat^{m}}} \gamma_j \cdot z_{\bs{M'}}^{\boldsymbol{\ell'}} = \sum_{j \in {\cal C}} \gamma_j \sum_{\substack{{\bs{M'}} \in {\cal D}({\cal S}):\\ j \in M'_i}} \sum_{\substack{\boldsymbol{\ell'} \in \nat^{m}}} z_{\bs{M'}}^{\boldsymbol{\ell'}} = \sum_{j \in {\cal C}} \gamma_j y_{i,j} \leq \lambda_i \ ,
\] which coincides with constraint \eqref{ineq:abundant_capacity}. Similarly, a type-$j$ customer can be matched with probability at most~1:
\[
    \sum_{i \in {\cal S}} \sum_{\substack{{\bs{M'}} \in {\cal D}({\cal S}):\\ j \in M'_i}} \sum_{\substack{\boldsymbol{\ell'} \in \nat^{m}}}  z_{\bs{M'}}^{\boldsymbol{\ell'}} = \sum_{i \in {\cal S}^{\rm short}} \sum_{\substack{{\bs{M'}} \in {\cal D}({\cal S}):\\ j \in M'_i}} \sum_{\substack{\boldsymbol{\ell'} \in \nat^{m}}}  z_{\bs{M'}}^{\boldsymbol{\ell'}} + \sum_{i \in {\cal S}^{\rm long}} y_{i,j} = \sum_{\MyAtop{i \in {\cal S}^{\rm short}}{{\bs{M}} \in {\cal D}({\cal S}^{\rm short})}}  \sum_{\substack{j \in M_i\\{\boldsymbol{\ell}: \ell_i \geq 1}}} x_{\bs{M}}^{\boldsymbol{\ell}} + \sum_{i\in {\cal S}^{\rm long}} y_{i,j} \leq 1 \ ,
\] as in constraint \eqref{ineq:NLP-contention-constraint}. Moreover, it is easy to see that having $\tau(\pi_g) \geq (1-\eps)\tau^*$ is equivalent to
\[
\sum_{\MyAtop{i \in {\cal S}^{\rm short}}{{\bs{M}} \in {\cal D}({\cal S}^{\rm short})}}  \sum_{\substack{j \in M_i\\{\boldsymbol{\ell}: \ell_i \geq 1}}}\gamma(M_i) \cdot x_{\bs{M}}^{\boldsymbol{\ell}} + \sum_{i \in {\cal S}^{\rm long}} \sum_{j \in {\cal C}}\gamma_{j}  y_{i, j} \geq (1-\eps) \tau^* \ .
\] In conclusion, we infer that $(x_{\bs{M}}^{\ell})_{{\bs{M}}, \ell}$ and $(y_{i,j})_{i,j}$ satisfy the constraints of $NLP(\bar{\ell})$.  Furthermore, because the objective of the LP, under $(x_{\bs{M}}^{\ell})_{{\bs{M}}, \ell}$ and $(y_{i,j})_{i,j}$, coincides with $c({\pi}^{\rm trunc})$, and ${\pi}^{\rm trunc}$ makes fewer matches than ${\pi}$, we have $NLP(\bar{\ell})^* \leq {c}({\pi})$, as desired. 

Finally, we argue about the computational complexity of solving $(NLP(\bar{\ell}))$ by presenting and analyzing a separation oracle. The separation oracle uses brute force and examines all the constraints. Since the number of partitions is at most $m^n$ and we have $\bar{\ell}^{n_s}$ different states for short queues, each iteration of the ellipsoid method has a running time of $O(\bar{\ell}^n  m^n)$. Since the number of variables of $NLP(\bar{\ell})$ is also $O(\bar{\ell}^n  m^n)$, the total computational complexity of solving $(NLP(\bar{\ell}))$ is $\poly(\bar{\ell}^n  m^n \cdot |{\cal I}|)$, where $|{\cal I}|$ is size of the input \citep{bertsimas1997introduction}.

\subsection{Proof of \Cref{prop:small_tau}}\phantomsection\label{prf:small_tau}
We begin with a classification of queues. 
    \begin{definition}
        A short supplier type is {\em infrequent} if $\lambda_i \leq \frac{\tau^*}{\eps}$ and {\em frequent} otherwise. We denote the infrequent and frequent queues by $\vv$ and $\sss$, respectively. Note that some short queues may be frequent according to this classification.
    \end{definition}

We construct a near-optimal policy $\bar{\pi}$ that is $\bar{\ell}$-bounded for all queues. At a high level, we decompose the dynamic matching instance into two smaller subinstances: one containing only infrequent queues, and the other containing only frequent queues. The infrequent queues can be bounded with an argument similar to the one in \Cref{clm:large_tau}. Consequently, the key step is to show that a simple greedy policy is near-optimal when handling only the frequent queues. We then show how to implement the policies of both subinstances simultaneously with only minimal loss. 

\paragraph{\bf Decoupling of infrequent and frequent queues.} 

Let $\tau_{\rm in}^{{\pi}}$ be $\pi$'s throughput rate from matches to infrequent queues and $\tau_{\rm fr}^{\pi} = \tau(\pi) - \tau_{\rm in}^{\pi}$. Analogously, define $c_{\rm in}^{\pi}$ and $c_{\rm fr}^{\pi}$ for the cost rates from matches to infrequent and frequent queues, respectively. Consider a dynamic matching problem ${\cal I}_{\rm fr}$ with ${\cal C}$ and only $\sss$ as the suppliers, with throughput target $\tau_{\rm fr}$. Similarly, define ${\cal I}_{\rm in}$ to be a dynamic matching problem with ${\cal C}$ and only $\vv$ as the suppliers, with throughput target $\tau_{\rm in}$. Let $\pi^*_{\rm fr}, \pi_{\rm in}^*$ be, respectively, the stationary optimal policies of ${\cal I}_{\rm fr}, {\cal I}_{\rm in}$. It thus follows that $c(\pi_{\rm fr}^*) \leq c_{\rm fr}^{{\pi}}$ and $c(\pi_{\rm in}^*) \leq c_{\rm in}^{{\pi}}$. 

\paragraph{\bf Policy for $\bs{{\cal I}_{\rm in}}$.} An argument identical to the one for \Cref{clm:large_tau}---where we mimicked $\pi$ while bounding the queues---allows us to prove that bounding infrequent queues by $\bar{\ell}$ decreases the throughput by at must $\eps \tau^*$. Specifically, we can plug in the assumption $\lambda_i \leq \frac{\tau^*}{\eps}$ in inequality \eqref{ineq:u_deriv} to upper bound the loss of throughput. Consequently, we denote by ${\pi'}$ the policy that is $\bar{\ell}$-bounded for infrequent queues and satisfies $\tau({\pi'}) \geq (1-\eps) \cdot \tau^*$ and $c({\pi'}) \leq c(\pi)$.

For ${\cal I}_{\rm in}$, we devise policy $\pi_{\rm in}$ as follows: we simulate the frequent queues by having fake frequent supplier arrivals. Then, we imitate policy $\pi'$ on the collection of infrequent queues and fake frequent queues. Since $\pi_{\rm in}$'s match rates to infrequent queues is identical to those of $\pi'$, it is clear that we have $\tau(\pi_{\rm in}) \geq (1-\eps)\tau_{\rm in}^{\pi}$ and $c(\pi_{\rm in}) \leq c_{\rm in}^{\pi}$. Also, by construction, $\pi_{\rm in}$ is $\bar{\ell}$-bounded for infrequent queues. 

\paragraph{\bf Policy for $\bs{{\cal I}_{\rm fr}}$.} 
First, we construct a greedy policy $\pi_g$ is near optimal for ${\cal I}_{\rm fr}$. Next, we show that we can modify $\pi_g$ to obtain an $\bar{\ell}$-bounded policy ${\pi}_{\rm fr}$ with small loss against $\pi^*_{\rm fr}$.  

Our policy $\pi_g$ uses the following \emph{static} LP relaxation:
\begin{align}
	\nonumber  \min_{{\bs z}} \quad &  \sum_{i \in \sss} \sum_{j \in {\cal C}} c_{i,j} \cdot z_{i,j} && \tag{SLP} \label{TLPon} \\
	\textrm{s.t.} 
	 \quad &   \sum_{i \in H} z_{i,j} \le \gamma_j \cdot \left(1 - \exp\left(-\sum_{i \in H}  \lambda_i \right)\right) \ , && \forall H  \subseteq \sss, \forall j \in {\cal C} \label{eqn:tightOnlineFlow}
     \\
      & \sum_{i \in \sss} \sum_{j \in {\cal C}} z_{i,j} \ge \tau_s^* \ , \label{ineq:tlp_thr}
      \\
 & z_{i,j} \ge 0 \ . && \forall i \in \sss, \forall j \in {\cal C} \notag
\end{align}
 It is easy to see that \eqref{TLPon} is a valid relaxation for the optimal policy of $\tilde{\cal I}$, which also implies $SLP^* \leq  c(\pi_{\rm fr}^*)$; see \citet[Claim 2.1]{amanihamedani2024improved} for a proof.

\paragraph{Properties of \eqref{TLPon}'s optimal solution.} The greediness of $\pi_g$ stems from a greedy structure of \eqref{TLPon}'s optimal solution, as explained next. Since the set function $H \to 1-\exp(-\sum_{i \in H} \lambda_i)$ is submodular, \eqref{TLPon} is related to optimizing a linear function over a polymatroid.\footnote{Although the feasible region of \eqref{TLPon} is not strictly a polymatroid because of constraint \eqref{ineq:tlp_thr}, Lagrangifying this constraint yields a dual LP that is a linear optimization over a polymatroid. In fact, the optimal Lagrange multiplier is equal to $c_{\sigma_k(L_k), k}$.} Thus, the optimal solution is obtained via a greedy algorithm \cite[Cor. 44.3a]{schrijver2003combinatorial}. Specifically, by rearranging customer types, there exists an index $k \in [m]$ such that for each $j \leq k$, there exist a number $L_j \in [|\sss|]$ and a permutation of $\sss$, called $\sigma_j(\cdot)$, such that:

\begin{itemize}
    \item for every (i) $1 \leq l \leq L_j$ if $j < k$, or (ii) $1 \leq l < L_k$ if $j = k$, we have\footnote{If $l = 1$, the empty sum is defined to be 0.} \begin{align} x_{\sigma_j(l), j} = \gamma_j \left( \exp \left( -\sum_{i \in \{\sigma_j(1), \ldots, \sigma_j(l-1)\}} \lambda_i \right) - \exp \left( -\sum_{i \in \{\sigma_j(1), \ldots, \sigma_j(l)\}} \lambda_i \right) \right) \ , \label{eq:polymatroid_sol} \end{align}
    \item and for every (i) $l > L_j$ if $j \leq k$, or (ii) $l \in [|\sss|]$ if $j > k$, we have $x_{\sigma_j(l), j} = 0$.
\end{itemize}
For $j = k$ and $l = L_k$, it is possible that $x_{\sigma_k(L_k), k}$ is smaller than the right-hand side of expression \eqref{eq:polymatroid_sol}, due to the threshold constraint \eqref{ineq:tlp_thr}.

We also remark that the greedy structure described above implies an upper bound on the total arrival rate of customer types $j \in [k]$. This bound will help us analyze the downward pressure on frequent queues due to the matches of $\pi_g$. In the following claim, we define $\gamma_k' = \frac{x_{\sigma_k(L_k), k}}{1-\exp(-\lambda_{\sigma_k(L_k)})}$ (proof in Appendix \ref{prf:downpressure}).
\begin{claim}\label{clm:downpressure}
    If $L_k > 1$, we have $\sum_{j = 1}^k \gamma_j \leq 2\eps$. Otherwise if $L_k = 1$, we have $\sum_{j = 1}^{k-1} \gamma_j + \gamma_k' \leq 3\eps$. 
\end{claim}

\paragraph{Construction of $\pi_g$.} We now construct $\pi_g$ using an optimal greedy solution of \eqref{TLPon}. For every queue $i \in \sss$, let $P_i(t)$ be the number of \emph{present} suppliers, i.e., how many suppliers have arrived not yet abandoned at time $t$, regardless of whether or not they have been matched. Recall the definition of $k$ for the solution of \eqref{TLPon}. If $j > k$, we do not match the customer since $L_j = 0$. If $j \leq k$, we find the highest ranked queue, according to $\sigma_j$, that is non-empty of present suppliers, i.e., we find
\begin{align}
    q^* = \min \left \{ 1 \leq q \leq L_j: P_{\sigma_j(q)}(t) > 0 \right \} \ , \label{expr:q_vals}
\end{align}
Then, we match the customer to queue $\sigma_j(q^*)$ if there is an available supplier in that queue who has not been previously matched; we use the notation $\pi_g^j(t)$ to denote $\sigma_j(q^*)$ at time $t$. If no value of $q$ satisfies the conditions in expression \eqref{expr:q_vals}, the customer leaves unmatched and we let $\pi_g^j(t) = \emptyset$.

We now provide some intuition for why $\pi_g$ is nearly optimal. Consider a frequent queue $i$ and a customer of type $j < k$ arriving at time $t \geq 0$. Suppose $i = \sigma_j(l)$ for some $l \leq L_j$. The probability that $\pi_g^j(t) = i$---that is, the customer is routed to queue~$i$---is exactly
\[
\exp\left( -\sum_{i' \in \{\sigma_j(1), \ldots, \sigma_j(l-1)\}} \lambda_{i'} \right).
\] Now, by \Cref{clm:downpressure}, the depletion rate of queue~$i$ at length~$\ell$ is at most $\ell + 3\eps$. This means that the queue stochastically dominates one that depletes at rate $\ell(1 + 3\eps)$. In such a system, the probability that the queue is non-empty is $1 - \exp(-\frac{\lambda_i}{1 + 3\eps})$. Assuming for the moment that queue~$i$ evolves independently of the other queues, the probability that the customer is successfully matched to queue $i$ is the product of the two probabilities above:
\[
\left(1 - \exp\left(-\frac{\lambda_i}{1 + 3\eps}\right)\right) \cdot \exp\left( -\sum_{i' \in \{\sigma_j(1), \ldots, \sigma_j(l-1)\}} \lambda_{i'} \right).
\]
This match probability yields a match rate that is within a $(1 - O(\eps))$ factor of the LP solution in~\eqref{eq:polymatroid_sol}. However, in reality, the state of queue~$i$ may be correlated with other queues, since their states affect the depletion dynamics of queue $i$. This dependence complicates the analysis. To address this, we introduce a randomized discarding mechanism via fake customers that makes the frequent queues evolve independently. We formalize this idea when discussing our policy $\bar{\pi}$ for the original dynamic matching instance.  

\paragraph{Policy $\pi_{\rm fr}$ and bounding frequent queues.} Lastly, we bound the frequent queues by $\bar{\ell}$, while incurring only an $O(\eps)$ loss in throughput. The policy $\pi_{\mathrm{fr}}$ mimics every decision made by $\pi_g$; Specifically, when a type-$j$ customer arrives at time $t$, policy $\pi_{\mathrm{fr}}$ attempts to match them to queue $\pi_g^j(t)$, provided there is at least one available supplier of type $i$. To enforce the $\bar{\ell}$-bounded property, $\pi_{\mathrm{fr}}$ discards any newly arriving supplier of type $i$ if queue~$i$ already has $\bar{\ell}$ suppliers. We will later show that this truncation affects the match rates by only an $O(\eps)$-factor.

\paragraph{\bf Putting everything together: an $\bs{\bar{\ell}}$-bounded policy $\bs{\bar{\pi}}$.} 

We now describe how to (almost) implement both $\pi_{\rm in}$ and $\pi_{\rm fr}$ simultaneously. At a high level, our approach aims to mimic the decisions made by each policy, including both matches and also supplier discards (to keep the queues bounded). A key challenge arises when an arriving customer must be matched under both policies. In such cases, we prioritize matches to infrequent queues and follow the decision of $\pi_{\rm in}$. Because infrequent queues are rarely non-empty, this prioritization has a negligible impact on the match rates to frequent queues.

Moreover, to simplify the analysis of correlations among queues, we introduce \textit{fake} customers that lead to discarding of some suppliers. Our goal is to make the evolution of frequent queues independent from each other and from that of infrequent queues. For every $j \in {\cal C}$, we create some fake customer arrival processes so that there exists a type-$j$ customer arrival process with rate $\gamma_j$ for each frequent queue, as well as one for the collection of infrequent queues. Furthermore, our construction is such that these different type-$j$ customer streams are independent.

\paragraph{Fake customers.} Let $n_f = |\sss|$ be the number of frequent queues. Consider an arbitrary $j \in {\cal C}$ and simulate $\ns$ Poisson processes of fake type-$j$ customers labeled from 2 to $\ns+1$, each arriving with rate $\gamma_j$. Let $N(\cdot)$ be the counting process of the aggregate type $j$ arrivals which has rate $(\ns+1) \gamma_j$ (with $\ns$ fake streams and one real stream). Consider an arrival time $t$ with $N(t) = N^-(t) + 1$, where $N^-(t) = \lim\limits_{\tau \to t^-}N(\tau)$. Moreover, let $f(t)$ be the label of the customer arriving at time $t$. For example, if there is a real type-$j$ arrival at time $t$, we have $f(t) = 1$. By construction, for every arrival time $t$, we have 
    \begin{align}\label{eq:f-probability}
    \pr{f(t) = l \; \left| \; N(t) = N^-(t) + 1\right.} = \frac{1}{\ns + 1} \quad \text{for every } l \in [\ns + 1] \ .
    \end{align}
    
    For any policy $\tilde{\pi}$, let $\tilde{\pi}^j(t)$ denote the queue to which $\tilde{\pi}$ would match a type-$j$ customer arriving at time $t$; we set $\tilde{\pi}^j(t) = \emptyset$ if such a customer would not be matched by $\tilde{\pi}$. This decision is made at the beginning of time $t$, so $\tilde{\pi}^j(t)$ is well-defined regardless of whether a type-$j$ customer actually arrives at that time. Intuitively, we want the probability that the next arriving type-$j$ customer is routed to a frequent queue $i\in \sss$ to be exactly $\frac{1}{\ns + 1}$ which means this queue sees type-$j$ customers arriving with rate $\gamma_j$. 
    
    To obtain this property, we define a time-indexed family of permutations $\left(\psi^j_\tau\right)_{\tau \geq 0}$. For each time $t \geq 0$, the permutation $\psi^j_t(\cdot)$ is constructed at follows:
    \begin{itemize}
  \item {If \(\pi_{\rm in}^j(t) \neq \emptyset\):}
  \begin{itemize}
    \item Set \(\psi_t^j(1) = \pi_{\rm in}^j(t)\).
    \item For \(1 < l \leq n_f+1\), select \(\psi_t^j(l)\) uniformly at random from $\sss \setminus \{\psi_t^j(2), \dots, \psi_t^j(l-1)\}$.
    
  \end{itemize}
  
  \item {Else if \(\pi_{\rm fr}^j(t) \neq \emptyset\):}
  \begin{itemize}
    \item Set \(\psi_t^j(1) = \pi_{\rm fr}^j(t)\).
    \item For \(1 < l \leq n_f\), select \(\psi_t^j(l)\) uniformly at random from $\sss \setminus \{\pi_{\rm fr}^j(t), \psi_t^j(2), \dots, \psi_t^j(l-1)\}$.
    \item Set \(\psi_t^j(n_f + 1) = \emptyset\).
  \end{itemize}
  
  \item {Else:}
  \begin{itemize}
    \item For \(1 \leq l \leq n_f\), select \(\psi_t^j(l)\) uniformly at random from $\sss \setminus \{\psi_t^j(2), \dots, \psi_t^j(l-1)\}$.
    \item Set \(\psi_t^j(n_f + 1) = \emptyset\).
  \end{itemize}
\end{itemize}
This permutation maps the label of an arriving customer to a queue. Consequently, if a type-$j$ customer with label $l$ arrives at time $t$, we route it to queue $\psi_t^j(l)$. As a result, real customers are routed to $\pi_{\rm in}^j(t)$ or $\pi_{\rm fr}^j(t)$, prioritizing infrequent queues. In contrast, fake customers are routed to a queue chosen at random from the set of frequent queues. For every frequent queue $i$, there exists a unique position $l$ in the permutation such that $i = \psi_t^j(l)$.

    

\paragraph{Matching and discarding decisions.} Consider a time $t$ where a type-$j$ customer with label $l$ arrives. If the customer is routed to infrequent queues and is real (i.e., $l = 1$), we match it to the specific infrequent queue assigned by $\pi_{\rm in}$. Now suppose instead that the customer is routed to a frequent queue $i$, and let $i = \sigma_j(q)$ for some $q \in [|\sss|]$. Recall the definition of $k$ for the solution of \eqref{TLPon}. If $j < k$ or if $j = k$ and $q < L_k$, we match the customer to a supplier in queue $i$, if possible. In other words, we first check if there is a type-$i$ supplier that is present and has not been matched/discarded before. If so, we match it to the customer, if $l=1$, and otherwise, discard the supplier if $l > 1$. If there are no type-$i$ suppliers available, the customer leaves unmatched. On the other hand, if $j = k$ and $q = L_k$, we define $\rho = \frac{x_{i,k}}{\gamma_k(1-\exp(-\lambda_{i}))}$ and match the customer to queue $i$ (if possible) with conditional probability $\rho$. With the remaining probability $1-\rho$, the customer departs unmatched.\footnote{Note that $\rho \leq 1$ by constraint \eqref{eqn:tightOnlineFlow}.}

\paragraph{\bf Analysis of $\bs{\bar{\pi}}$.} Note that $\bar{\pi}$ makes exactly the same decisions as $\pi_{\rm in}$. Consequently, the match rates to infrequent queues are identical, and $\bar{\pi}$'s throughput and cost contributions from infrequent queues are precisely $\tau(\pi_{\rm in})$ and $c(\pi_{\rm in})$, respectively. It remains, then, to analyze the match rates associated with the frequent queues. We begin by establishing that frequent queues evolve independently—both from one another and from the infrequent queues. Then, using \Cref{clm:downpressure}, we derive a lower bound on the probability that a frequent queue is non-empty at the time of a customer arrival. Combining these observations, we characterize the match rates under $\bar{\pi}$.

\paragraph{Independence of queues.}
Our policy routes a type-$j$ customer, arriving at time $t$ with label $f(t)$, to queue $\psi^j_t(f(t))$. The important observation is that $\psi^j_t(f(t))$ is independent of ${\cal F}_t^{-} = \bigcup\limits_{\tau < t} {\cal F}_\tau$, where ${\cal F}_\tau$ is the canonical filtration of the system. This fact, formally stated below, enables us to use the Poisson thinning property. The proof appears in Appendix \ref{app:label-independence}.
    \begin{claim}\label{clm:label-independence}
        For any arrival time $t$ with $N(t) = N^-(t) + 1$, $\psi_t^j(f(t))$ is independent of ${\cal F}_t^-$.
    \end{claim}
For a frequent queue $i$, $\{\psi_t^j(f(t)) = i\}$ is equivalent to $\{f(t) = (\psi_t^j)^{-1}(i)\}$. Thus, by Claim \ref{clm:label-independence} and the Poisson thinning property, we infer that each queue has an independent stream of arrivals of type $j$ with rate $\gamma_j$. Furthermore, the matching/discarding decisions are made independently from other queues. Therefore, as discussed earlier, we conclude that frequent queues evolve independently from each other and from infrequent queues. 

\paragraph{Availability probability.} We now analyze the probability of matching an arriving type-$j$ customer to queue $i$ at time $t \geq 0$. Suppose that $j < k$ and $i = \sigma_j(q)$. If $q > L_j$, we never make the match. On the other hand, we match the customer to queue $i$ if we have (i) ${\cal L}_i^{\bar{\pi}} > 0$, (ii) $P_{i'}(t) = 0$ for every $i' \in \{\sigma_{j}(1), \cdots, \sigma_j(q-1)\}$, and (iii) the customer is not matched to infrequent queues by $\pi_{\rm in}$, i.e., $\pi_{\rm in}^j(t) = \emptyset$. Note that since frequent queues are independent from each other and infrequent queues, conditions (i)-(iii) are independent. Condition (ii) occurs with probability $ \exp  (- \sum_{i' \in \{\sigma_{j}(1), \cdots, \sigma_j(q-1)\} }\lambda_{i'} )$, since the number of present suppliers in a queue is a Poisson distribution, independent from other queues. On the other hand, the next claim (proof in Appendix \ref{prf:ltilde_avail}) analyzes the probability of ${\cal L}^{\bar{\pi}}_i(t) > 0$. 
\begin{claim}\label{clm:ltilde_avail} We have
    \[
        \pr{{\cal L}_i^{\bar{\pi}}(t) > 0}  > \frac{1-2\eps}{1+3\eps} \cdot \left ( 1 - \exp\left (-{\lambda_i}\right )\right ) \ .
    \]
\end{claim}
Lastly, the following claim (proof in Appendix \ref{prf:vshort_match}) upper bounds the probability of condition (iii) holding.

\begin{claim}\label{clm:vshort_match}
    Policy $\pi_{\rm in}$ matches a customer arriving at time $t \ge 0$ with probability at most $\eps$. In other words, we have $\pr{{\pi_{\rm in}^j}(t) = \emptyset} \geq 1-\eps$. 
\end{claim}
\paragraph{Analyzing match rates.} Recall that $\bar{\pi}$ is completely mimicking $\pi_{\rm in}$ for matches to infrequent queues, so the match rates to infrequent queues are the same. On the other hand, with Claims~\ref{clm:ltilde_avail} and \ref{clm:vshort_match}, we can lower bound the match rate between queue $i \in \sss$ and customer type $j$. This piece of proof is very similar to the last part of \Cref{thm:constant_ptas}'s proof in \Cref{ssec:pm_analysis}. We discretize time in small steps $\Delta t$ and argue that
\begin{align}
    \tau_{i,j}(\bar{\pi}) &= \liminf\limits_{T \to \infty} \frac{1}{T} \cdot \ex{M_{i,j}(0, T)} \notag \\ &\geq \liminf\limits_{k \to \infty} \frac{1}{k \Delta t} \cdot \left(  \ex{M_{i,j}(0, k\Delta t)} - \gamma_j \Delta t \right) \ , \label{ineq:trunc_T_vs_kdelta}
\end{align}
where setting $k = \lceil T / \Delta t \rceil$, we have $\ex{M_{i,j}(0, k \Delta t)}  \leq \ex{M_{i,j}(0, T)} + \gamma_j \Delta t$, since $M_{i,j}(t_1, t_2) \leq A_j(t_1, t_2)$ and $A_j(t_1, t_2)$ is a Poisson process with rate $\gamma_j$. The right hand side of inequality \eqref{ineq:trunc_T_vs_kdelta} is equal to
\begin{align}
    &\liminf\limits_{k \to \infty} \frac{1}{k \Delta t} \cdot \ex{M_{i,j}(0, k\Delta t)} \notag \\
    & \; = \liminf\limits_{k \to \infty} \frac{1}{k \Delta t} \cdot {\sum_{r = 0}^{k-1} \ex{M_{i,j}(r\Delta t, (r+1)\Delta t)}} \notag
    \\ & \; \geq \liminf\limits_{k \to \infty} \frac{1}{k \Delta t} \cdot \sum_{r = 0}^{k-1} \pr{A_j(r\Delta t, (r+1) \Delta t) = 1} \ex{\left. M_{i,j}(r\Delta t, (r+1)\Delta t)\;  \right | \; A_j(r\Delta t, (r+1) \Delta t) = 1} \ , \label{ineq:trunc_one_arrival_bound}
\end{align}
where the inequality considers, as a lower bound, the case where exactly one customer arrives in each step. Observe that, conditioned on exactly one type-$j$ customer arriving within the interval $[r\Delta t,, (r+1)\Delta t)$, this customer is matched to a type-$i$ supplier provided that upon its arrival, we have (i) $\tilde{\cal L}_i(t) > 0$, (ii) $Q_{i'}(t) = 0$ for every $i' \in \{\sigma_{j}(1), \cdots, \sigma_j(q-1)\}$, and (iii) $\bar{\pi}^j(t) = \emptyset$. Let $t_{\mathfrak{c}}$ be the arrival time of that customer, denoted by $\mathfrak{c}$. By Claims~\ref{clm:ltilde_avail} and \ref{clm:vshort_match} and the independence of conditions (i)-(iii), the RHS of inequality \eqref{ineq:trunc_one_arrival_bound} is at least 
\begin{align}
    & \liminf\limits_{k \to \infty} \frac{1}{k \Delta t} \cdot \sum_{r = 0}^{k-1} \pr{A_j(r\Delta t, (r+1) \Delta t) = 1}  \cdot \pr{{\cal L}_i^{\bar{\pi}}(t) > 0} \cdot \pr{\sum_{i' \in \{\sigma_{j}(1), \cdots, \sigma_j(q-1)\}} P_{i'}(t) = 0} \cdot \pr{\pi_{\rm in}^j(t) = \emptyset} \notag
    \\ & \; \geq \liminf\limits_{k \to \infty} \frac{1}{k \Delta t} \cdot \sum_{r = 0}^{k-1} e^{-\gamma_j \Delta t}\gamma_j \Delta t \cdot \frac{1-2\eps}{1+3\eps} \cdot \left ( 1 - \exp\left (-{\lambda_i}\right )\right ) \cdot \exp  \left(- \sum_{i' \in \{\sigma_{j}(1), \cdots, \sigma_j(q-1)\} }\lambda_{i'} \right ) \cdot (1-\eps) \notag
    \\ & \; \geq (1-6\eps) \cdot \left ( 1 - \exp\left (-{\lambda_i}\right )\right ) \cdot \exp  \left(- \sum_{i' \in \{\sigma_{j}(1), \cdots, \sigma_j(q-1)\} }\lambda_{i'} \right ) \cdot \liminf\limits_{k \to \infty} \frac{1}{k \Delta t} \cdot \sum_{r = 0}^{k-1} e^{-\gamma_j \Delta t}\gamma_j\Delta t \ .  \notag
\end{align}
Now, letting $\Delta t \to 0$ yields 
\[
    \tau_{i,j}(\bar{\pi}) \geq \gamma_j \cdot (1-6\eps) \cdot \left ( 1 - \exp\left (-{\lambda_i}\right )\right ) \cdot \exp  \left(- \sum_{i' \in \{\sigma_{j}(1), \cdots, \sigma_j(q-1)\} }\lambda_{i'} \right ) \ .
\]
On the other hand, a necessary condition for matching $j$ to queue $i$ at time $t$ is that $P_{i'}(t) = 0$ for every $i' \in \{\sigma_{j}(1), \cdots, \sigma_j(q-1)\}$ but $P_i(t) > 0$, which occurs with probability $(1-\exp(-\lambda_i))\cdot \exp  (- \sum_{i' \in \{\sigma_{j}(1), \cdots, \sigma_j(q-1)\} }\lambda_{i'} )$, which immediately implies \[\tau_{i,j}(\bar{\pi}) \leq \gamma_j (1-\exp(-\lambda_i))\cdot \exp \left (- \sum_{i' \in \{\sigma_{j}(1), \cdots, \sigma_j(q-1)\} }\lambda_{i'} \right ) \ .\] 
Recalling the LP solution \eqref{eq:polymatroid_sol}, we deduce that $\tau_{i,j}(\bar{\pi}) = (1-6\eps)z_{i,j}$. Moreover, recall that we had assumed $j < k$. If we have $j = k$ and $i = \sigma_j(q)$ with $q < L_k$, the argument is identical. However, if $j = k$ and $i = \sigma_j(1)$ with $L_k = 1$, we should use $\gamma_k'$ above due to the thinning of type-$k$ customer arrivals. Since the argument is identical but only with a different benchmark match rate, we omit the details. 

Repeating the above argument for every $i \in \sss$ and $j \in {\cal C}$, and recalling that we are mimicking $\pi_{\rm in}$ for infrequent queues, we obtain
\[
    \tau(\bar{\pi}) \geq (1-6\eps) \tau_{\rm fr} + \tau_{\rm in} \geq (1-6\eps)\tau^* \quad \text{and} \quad c(\bar{\pi}) \leq  c(\pi_{\rm fr}^*) +  c_{\rm in}^\pi \leq c_{\rm fr}^\pi +  c_{\rm in}^\pi = c(\pi) \ ,
\]
where we recall that $\pi$ was the policy we started with, that satisfied $\tau(\pi) \geq \tau^*$. In conclusion, considering the optimal policy $\pi$ implies that policy $\bar{\pi}$ is a near-optimal policy that is $\bar{\ell}$-bounded for all queues. \qedsymb

\subsubsection{Proof of \Cref{clm:downpressure}.}\label{prf:downpressure}
        For every $j \leq k-1$, we have 
        \[
            z_{\sigma_j(1), j} = \gamma_j (1-\exp(-\lambda_{\sigma_j(1)})) \geq \gamma_j (1-\exp(-\tau^*/\eps)) \geq \gamma_j \cdot \frac{\tau^*}{2\eps} \ ,
        \]
        where we used the fact that $\tau^* \leq \eps^2$. Observing
        \[
            \frac{\tau^*}{2\eps} \cdot \sum_{j=1}^{k-1} \gamma_j \leq \sum_{j=1}^{k-1}z_{\sigma_j(1), j}  \leq \tau^*
        \]
        shows that $\sum_{j=1}^{k-1}\gamma_j \leq 2\eps$. As for online type $k$, note that if $L_k > 1$, we can strengthen the above inequality to get $\sum_{j = 1}^k \gamma_k \leq 2\eps$. Therefore, assume $L_k = 1$. We argue
    \[
        \gamma_k' = \frac{z_{\sigma_k(1), k}}{1-\exp(-\lambda_{\sigma_k(1)})} \leq \frac{\tau^*}{\tau^*/\eps} = \eps  \ ,
    \]
    where we used the fact that $\lambda_{\sigma_k(1)} \geq \frac{\tau^*}{\eps}$. The proof is now complete. 
\qedsymb

\subsubsection{Proof of \Cref{clm:label-independence}.}\label{app:label-independence}
        Let $A \in {\cal F}_t^-$ be an event measurable with respect to the history before time $t$. Moreover, let $S \subseteq [n]$ be a set of supplier types. We can argue
        \begin{align*}
            \pr{\psi_t^j(f(t)) \in S, A} &= \int_{H \in {\cal F}_t^-} \pr{\psi_t^j(f(t)) \in S, A \; | \; H} \cdot F(dH) \\ &\stackrel{(1)}{=} \int_{H \in {\cal F}_t^-} \pr{\psi_t^j(f(t)) \in S \; | \; H} \cdot \pr{A \; | \; H} \cdot F(dH) \\ &\stackrel{(2)}{=} \int_{H \in {\cal F}_t^-} \pr{\psi_t^j(f(t)) \in S \; | \; H} \cdot {\bb I}[A \; | \; H] \cdot F(dH) \\ &\stackrel{(3)}{=} \int_{H \in {\cal F}_t^-} \pr{f(t) \in (\psi_t^j)^{-1}(S) \; | \; H} \cdot {\bb I}[A \; | \; H] \cdot F(dH)  \\ &\stackrel{(4)}{=} \int_{H \in {\cal F}_t^-} \frac{|S|}{\ns+1} \cdot {\bb I}[A \; | \; H] \cdot F(dH) \\ & = \frac{|S|}{\ns+1} \cdot \pr{A} \\ & = \pr{\psi_t^j(f(t)) \in S} \cdot \pr{A} \ ,
        \end{align*}
        where (1) and (2) follow from $A \in {\cal F}_t^-$, (3) follows from $\psi_t^j \in {\cal F}_t^-$, and (4) is just by the construction of $f(t)$ in \eqref{eq:f-probability}. The proof of the claim is complete. 
        \qedsymb

\subsubsection{Proof of \Cref{clm:ltilde_avail}.}\label{prf:ltilde_avail}
    Recall the following basic result in probability theory:
    \begin{claim}\label{claim:stationarydistbirthdeath}
    For a birth-death process with constant birth rate $\lambda$ and linear death rate $\mu_\ell = \ell$, the stationary distribution is Poisson with parameter $\lambda$.
\end{claim}
    The death rate of ${\cal L}_i^{\bar{\pi}}$ at state $\ell$ is at most $\ell + 3\eps < \ell(1+3\eps)$. Thus, if ${\cal L}_i^{\bar{\pi}}$ was not truncated at $\bar{\ell}$, we had $\pr{{\cal L}_i^{\bar{\pi}}(t) > 0} \geq 1-\exp(-\lambda_i/(1+3\eps))$. By truncating at $\bar{\ell}$, the probability mass above $\bar{\ell}$ is uniformly spread across states below $\bar{\ell}$. Thus, letting $\tilde{\cal L}$ be the non-truncated version of ${\cal L}_i^{\bar{\pi}}$, we argue
    \begin{align*}
        \pr{\tilde{\cal L}(t) > \bar{\ell}} \leq \pr{\tilde{\cal L}(t) = \bar{\ell}} \cdot (\eps + \eps^{2} + \cdots) = \pr{\tilde{\cal L}(t) = \bar{\ell}} \cdot \frac{\eps}{1-\eps} \leq 2\eps\pr{\tilde{\cal L}(t) = \bar{\ell}}
    \end{align*}
    for every $\eps\leq1/2$. Therefore, because $\pr{\tilde{\cal L}(t) = \bar{\ell}} \leq \pr{{\cal L}_i^{\bar{\pi}}(t) = \bar{\ell}} \leq \pr{{\cal L}_i^{\bar{\pi}}(t) >0}$, after the truncation we have
    \[
        \pr{\tilde{\cal L}_i(t) > 0} \geq (1-2\eps) \cdot \left(1-\exp\left(-\frac{\lambda_i}{1+3\eps} \right) \right) \geq \frac{1-2\eps}{1+3\eps} \cdot \left(1-\exp\left(-{\lambda_i}\right) \right)   \ ,
    \] where the second inequality follows from concavity of $x \to 1-\exp(-x)$.
\qedsymb
\subsubsection{Proof of \Cref{clm:vshort_match}.}\label{prf:vshort_match}
    A necessary condition for $\pi_v^*$ matching an arriving customer is that there exists a infrequent supplier that is present, i.e., has arrived and not yet abandoned at time $t$. By the Poisson superposition property, the number of present infrequent suppliers forms a Poisson distribution with parameter $\sum_{i \in \vv}\lambda_i$, which is zero with probability
    \[
        \exp\left (-\sum_{i \in \vv}\lambda_i \right ) \geq \exp\left (-n \cdot \frac{\tau^*}{\eps} \right ) \geq \exp(-\eps) \geq 1-\eps \ ,
    \] where the second inequality follows from $\tau^* \leq \eps^2/n$.
    \qedsymb

\subsection{Proof of \Cref{lem:short_convergence}}\label{prf:short_convergence}
$\left\{{\cal L}^{\pi^{\rm pr}}_{\rm short}(t)\right\}_{t \geq 0}$ is governed by the following intensity matrix, for every pair of distinct states $\boldsymbol{\ell}, \boldsymbol{\ell}' \in [\bar{\ell}]_0^{n_s}$:
\begin{align*}
    {\cal Q}_{\boldsymbol{\ell}, \boldsymbol{\ell}'} = 
    \begin{cases}
        \lambda_i & \text{ if } \boldsymbol{\ell}' = \boldsymbol{\ell} + e_i \ , \\ \sum\limits_{{\bs{M}} \in {\cal D}({\cal S}^{\rm short})}\frac{x_{\bs{M}}^{\boldsymbol{\ell}}}{\sum_{{\bs{M'}} \in {\cal D}({\cal S}^{\rm short})} x_{\bs{M'}}^{\boldsymbol{\ell}}} \cdot \left(\gamma(M_i) + (\ell_i + 1)\right) & \text{ if } \boldsymbol{\ell}' = \boldsymbol{\ell} - e_i  \ , \\ 0 & \text{ otherwise} \ ,
    \end{cases}
\end{align*} and ${\cal Q}_{\boldsymbol{\ell}, \boldsymbol{\ell}} = -\sum_{\substack{\boldsymbol{\ell}' \in [\bar{\ell}]_0^{n_s}: \\ \boldsymbol{\ell'} \neq \boldsymbol{\ell}}} {\cal Q}_{\boldsymbol{\ell}, \boldsymbol{\ell'}}$ for every $\boldsymbol{\ell} \in [\bar{\ell}]_0^{n_s}$. Hence, since the probability measure \[\left(\sum_{{\bs{M}} \in {\cal D}({\cal S}^{\rm short})} x_{\bs{M}}^{\boldsymbol{\ell}}\right)_{\boldsymbol{\ell} \in [\bar{\ell}]_0^{n_s}}\] satisfies the global balance equations \eqref{eq:multivariate_flow_balance} with the intensity matrix ${\cal Q}$, and that there exists a unique stationary distribution satisfying the global balance equations \citep[Thm. 12.25]{kallenberg1997foundations}, we infer that
\begin{align}
    \pr{{\cal L}_{\rm short}^{\pi^{\rm pr}}(t) = \boldsymbol{\ell}} \overset{t \to \infty}{\longrightarrow} \pr{{\cal L}_{\rm short}^{\pi^{\rm pr}}(\infty) = \boldsymbol{\ell}} = \sum_{{\bs{M}} \in {\cal D}({\cal S}^{\rm short})} x_{\bs{M}}^{\boldsymbol{\ell}} \notag \ ,
\end{align} as desired. 



\subsection{Proof of \Cref{lem:abundant_empty_prob}}\label{prf:abundant_empty_prob}

The proof begins with coupling ${\cal L}_i^{\pi^{\rm pr}}(t)$ with an auxiliary queue $\tilde{\cal L}_i(t)$. Consequently, we analyze $\tilde{\cal L}_i(t)$ and complete the proof. 

\paragraph{Coupling with $\tilde{\cal L}_i(t)$.} Initialize the queue $\tilde{\cal L}_i(t)$ to be 0, and let it be incremented whenever a type-$i$ supplier arrives, similarly to ${\cal L}^{\pi^{\rm pr}}_i(t)$. To explain the decrements of $\tilde{\cal L}_i(t)$, we need a few definitions. For every customer type $j$, we create a duplicate (fake) customer type $j'$ arriving with rate $\gamma_j$, where the counting process is denoted by $N_j^{\rm fake}(\cdot)$. For every time $t' \in [0,t]$, let $\bs{M}(t')$ be the assignment associated with ${\cal L}_{\rm short}^{\pi^{\rm pr}}(t')$. An alternative view of \Cref{alg:priority_matching} is that this matching set is drawn at the start of time $t'$, before observing the arrivals at that time. In other words, letting ${\cal F}_{t'}$ be the canonical filtration corresponding to the Markov chain, we have that $\bs{M}(t')$ is ${\cal F}_{t'}^-$-measurable. 

We now define the \emph{patched} arrival process of type $j$, called $\overline{N}_j(\cdot)$. Let ${N}_j(\cdot)$ be the counting process for real customer arrivals. Then, for a time $t' \in [0,t]$, we have $\overline{N}_j(t') = \lim_{\tau \to t'^-} \overline{N}_j(\tau) + 1$ if and only if one of the conditions (i) ${N}_j(t') = \lim_{\tau \to t'^-} {N}_j(\tau) + 1$ and $j \not \in \bs{M}(t')$, or (ii) ${N}^{\rm fake}_j(t') = \lim_{\tau \to t'^-} {N}^{\rm fake}_j(\tau) + 1$ and $j  \in \bs{M}(t')$, holds. In other words, there is a patched arrival if either a real customer arrives that is not matched to short queues, or a fake customer arrives when $j \in \bs{M}(t')$. Next, we define the \emph{complementary} arrival process to be the complement of the patched process. In other words, letting $N^c_j(\cdot)$ be the complementary counting process of, we have ${N}^c_j(t') = \lim_{\tau \to t'^-} {N}^c_j(\tau) + 1$ if and only if one of the conditions (i) ${N}_j(t') = \lim_{\tau \to t'^-} {N}_j(\tau) + 1$ and $j \in \bs{M}(t')$, or (ii) ${N}^{\rm fake}_j(t') = \lim_{\tau \to t'^-} {N}^{\rm fake}_j(\tau) + 1$ and $j \not \in \bs{M}(t')$, holds. By using the Poisson thinning property, we can show that the patched and complementary process are independent. 


\begin{claim}\label{clm:process_independence}
    The patched and complementary process are independent and identically distributed.
\end{claim} 

\Cref{clm:process_independence} is proved at Appendix \ref{prf:process_independence}. We now explain how $\tilde{\cal L}_i$ is decremented. $\tilde{\cal L}_i$ is decremented---if it is positive---at time $t'$ if there is {a} patched arrival of type $j$ at that time, and one of the conditions (i) $j \in {\cal C}^{\rm ct}$, or (ii) $j \in \overline{{\cal C}^{\rm ct}}$ and the draw in Line \ref{lin:abundant_scheduling} is $i^{\rm long} = i$, holds.\footnote{Although $i^{\rm long}$ was originally designed to be drawn upon real customer arrivals, we now do so upon fake customer arrivals too.} Furthermore, $\tilde{\cal L}_i$ is also decremented if some type-$i$ supplier abandons the system without being matched. Now, note that whenever some real type-$i$ supplier is matched or it abandons, which means that the real queue ${\cal L}_i^{\pi^{\rm pr}}$ is depleted by one, $\tilde{\cal L}_i$ is also decremented (if it is positive). Since the increments of $\tilde{\cal L}_i$ and ${\cal L}_i^{\pi^{\rm pr}}$ are the same, we infer that $\tilde{\cal L}_i(t') \leq {\cal L}_i^{\pi^{\rm pr}}(t')$ for every $t' \in [0,t]$ almost surely. Therefore, it now suffices to upper bound $\prpar{\tilde{\cal L}_i(t) = 0 \; | \; {\cal L}^{\pi^{\rm pr}}_{\rm short}(t) = \bs{\ell}}$.

Note that the evolution of the short queues is fully determined by their arrivals, abandonments, and the complementary processes. $\tilde{\cal L}_i$, however, is fully determined by its own arrivals and abandonments, as well as the patched processes. Therefore, since $\tilde{\cal L}_i$ and short queues are determined by independent processes, their evolution is independent too. Thus, it now suffices to upper bound $\prpar{{\tilde{\cal L}_i(t) = 0}}$.

    \indent\paragraph{Analysis of $\tilde{\cal L}_i$.} 
    Note that $\tilde{\cal L}_i$ is simply an M/M/1 queue (with abandonments) that is incremented with rate $\lambda_i$ and decremented with rate 
    \begin{align}
        & (1-\eps) \cdot \sum_{j \in \overline{{\cal C}^{\rm ct}}} \gamma_j y_{i,j} + \sum_{j \in {\cal C}^{\rm ct}} \gamma_j \notag \ .
    \end{align} 
    We now upper bound $\sum_{j \in {\cal C}^{\rm ct}} \gamma_j$. Note that the polytope ${\cal B}({\cal S}^{\rm short}, \bar{\ell})$ with constraints \eqref{eq:multivariate_flow_balance} describes the stochastic system for short queues and thus it satisfies the flow balance equation. Therefore, in particular, we have 
    \begin{align}
        \sum_{j \in {\cal C}^{\rm ct}} \gamma_j  \sum\limits_{\substack{{\bs{M}} \in {\cal D}({\cal S}^{\rm short}):\\{j\in {\bs{M}}}}}  \sum\limits_{\substack{{\boldsymbol{\ell}} \in [\bar{\ell}]^{n_s}}} x_{\bs{M}}^{\boldsymbol{\ell}} \leq \sum_{i \in {\cal S}^{\rm short}}  \lambda_i \ ,  \label{ineq:short_match_flow}
    \end{align} which means that the average match rate of contentious customer types to short queues is at most the arrival rate of short supplier types. Using the definition of contentious customer and short supplier types, we have
    \[ 
        {\eps} \cdot \sum_{j \in {\cal C}^{\rm ct}} \gamma_j \leq \sum_{i \in {\cal S}^{\rm short}} \lambda_i \leq n \cdot \frac{1}{\delta^\kappa} \ .
    \] Consequently, we can upper bound rate at which $\tilde{\cal L}_i$ is decremented to be at most
    \begin{align}
        & (1-\eps) \cdot \sum_{j \in \overline{{\cal C}^{\rm ct}}} \gamma_j y_{i,j} + \sum_{j \in {\cal C}^{\rm ct}} \gamma_j \notag \\ & \; \leq  (1-\eps)\lambda_i + \frac{n}{\eps \delta^\kappa} \notag \\ & \; \leq (1-\eps)\lambda_i + \frac{n\delta}{\eps}  \lambda_i \notag \\ & \; \leq \lambda_i \ , \label{ineq:downpressure_ub}
    \end{align} 
    where we use the Poisson superposition property and constraint \eqref{ineq:abundant_capacity} along with having $\delta \leq \frac{\eps^2}{n}$. Now, we assume that the decrement rate is exactly $\lambda_i$, since it only increases the probability of $\tilde{\cal L}_i(t)$ being 0. 

    We now use Little's law for $\tilde{\cal L}_i$ to write
    \begin{align}\label{eq:fast-queue-stability}
     \lambda_i = \lambda_i  \pr{\tilde{\cal L}_i(t) \geq 1} + \ex{\tilde{\cal L}_i(t)} \ ,    
    \end{align}
      which follows from the intuition that any arriving supplier either abandons the system or is matched. Using the drift method in \Cref{clm:exp_queue_ub} (in Appendix \ref{app:exp_queue_ub}), we obtain the upper bound $\expar{\tilde{\cal L}(t)} \leq \sqrt{{\lambda_i}}$. Consequently, equation \eqref{eq:fast-queue-stability} implies \[ \pr{\tilde{\cal L}_i(t) \geq 1} \geq 1 - \frac{1}{\sqrt{\lambda_i}}  \ . \] Recalling $\lambda_i \geq (\frac{1}{\delta})^{\kappa+1}$, we have $\frac{1}{\sqrt{\lambda_i}} \leq {\eps}$ if $\delta \leq {\eps^2}$, which completes the proof. 

\subsubsection{Proof of \Cref{clm:process_independence}.}\label{prf:process_independence} We provide a construction of real and fake customer arrivals that shows that the patched and complementary processes are independent. Concretely, consider some $j \in {\cal C}$ and an aggregate arrival process of rate $2\gamma_j$. Upon the arrival of a customer, we look up whether or not $j \in \bs{M}(t)$. If $j \notin {\bs M}(t)$, we flip a fair coin to pick a number in $\{1,2\}$. If the coin comes 1, we label the customer as real, and fake otherwise. If $j \in \bs{M}(t)$, the labeling is reversed. Namely, if the coin comes 1, we label the customer as fake, and real otherwise. Figure~\ref{fig:real_fake_arrival} illustrates this process.

    Note that the patched process corresponds to type-1 arrivals, while the complementary process corresponds to type-2 arrivals. Consequently, each customer is independently classified as either real or fake with probability $1/2$, regardless of past arrivals. By the Poisson thinning property, these processes are independent, and the claim is proved.
    
\begin{figure}[H]
    \centering
    \begin{tikzpicture}[scale=1, every node/.style={scale=1}]
\node at (-2,2) {real};
\node at (-2,1) {fake};
\node[circle, draw, inner sep=1pt] (r1) at (-1,2) {1};
\node[circle, draw, inner sep=1pt] (f2) at (-1,1) {2};

\node at (-2,-1) {fake};
\node at (-2,-2) {real};
\node[circle, draw, inner sep=1pt] (f1) at (-1,-1) {1};
\node[circle, draw, inner sep=1pt] (r2) at (-1,-2) {2};

\draw[thick,decorate,decoration={brace}] (-0.7,2.2) -- (-0.7,0.8);
\draw[thick,decorate,decoration={brace}] (-0.7,-0.8) -- (-0.7,-2.2);

\draw[thick] (2,0) circle(0.3) node[right]{};

\draw[thick,->] (3.5,0) -- (2.3,0) node[midway,above] {$2\gamma_j$};
\draw[thick,->] (1.7,0.2) -- (0,1.5) node[midway,above] {$ $};
\draw[thick,->] (1.7,-0.2) -- (0,-1.5) node[midway,below] {$ $};

\node[right] at (0.5,2) {$j\notin {\bs M}(t)$};
\node[right] at (0.5,-2) {$j\in \bs{M}(t)$};
\end{tikzpicture}
\caption{A schematic visualization of real and fake customer arrival process}
    \label{fig:real_fake_arrival}
\end{figure}

\subsection{Proof of \Cref{lem:scheduler_bounded}}
\phantomsection
\label{prf:scheduler_bounded}

We argue that ${\cal V}_{i,j}$ is an $M/M/1$ queue, modulated by the Markov chain corresponding to short queues ${\cal L}_{\rm short}^{\pi^{\rm pr}}$. That is, when ${\cal L}_{\rm short}^{\pi^{\rm pr}} = \bs{\ell}$ for some $\bs{\ell} \in [\bar{\ell}]_0^{n_s}$, the increments of ${\cal V}_{i,j}$ form a Poisson process with parameter $\vartheta(\bs{\ell})$ and each decrement takes an $Exp(\phi(\bs{\ell}))$ amount of time, where $Exp(\phi(\bs{\ell}))$ is an exponential random variable independent from previous decrements.

A primary tool for studying the stability of queueing systems and analyzing their steady-state behavior is the drift method and the Foster-Lyapunov theorem. Using an extension of the drift method, \cite{grosof2024analysis} derive upper bounds for the mean queue length of a Markov-modulated $M/M/1$ queue. However, since our goal is only to show a bounded expected queue length, we rely on the results of \cite{neuts1978further,neuts1978m, neuts1994matrix} who showed that the steady-state distribution of the queue length follows a matrix-geometric solution. First, we state the problem definition of \cite{neuts1994matrix} and present the steady-state distribution. Next, we show that ${\cal V}_{i,j}$ is a Markov-modulated $M/M/1$ queue that satisfies our desired stability conditions. Lastly, we conclude that the virtual buffer has bounded expectation at all times. 

\cite{neuts1994matrix} considers a Markovian environment that is described by a $k$-state irreducible Markov process with generator ${\cal Q}$. When this Markov process is at state $z$, $1 \leq z \leq k$, jobs arrive to a single server queue according to a Poisson process with rate $\zeta(z)$, and the service rate is $\sigma(z)$. When the Markov process ${\cal Q}$ transitions from state $z$ to $z'$, the arrival and service rates of the queue instantaneously change to $\zeta({z'})$ and $\sigma(z')$. Moreover, it is assumed that at least one $\zeta(\cdot)$ and at least one $\sigma(\cdot)$ is positive. 

Let $\bs{\theta}$ be the stationary distribution vector of ${\cal Q}$. A necessary and sufficient condition for the queue to be stable is that
\begin{align}
    \sum_{z = 1}^k \theta_z \cdot \zeta(z) < \sum_{z = 1}^k \theta(z) \cdot \sigma(z) \ . \label{ineq:modul_stab}
\end{align}
Moreover, let the matrix $R$ be the minimum solution of equation
\[
    R^2 \text{diag}({\bs \sigma}) + R({\cal Q} - \text{diag}({\bs \zeta} + {\bs \sigma})) + \text{diag}({\bs \zeta}) = 0 \ .
\] Then, when the stability condition \eqref{ineq:modul_stab} holds, the spectral radius of $R$ is less than 1. The following result pins down the stationary distribution of the queue length. 
\begin{theorem}[c.f. \cite{neuts1978m} Thm. 5 and \cite{neuts1994matrix} Thm 6.2.1]\label{thm:neuts}
    Under the stability condition \eqref{ineq:modul_stab}, the stationary distribution of the Markov-modulated $M/M/1$ queue is given by
    \[
        \alpha(\ell) = {\bs \theta}^\top (I - R) R^\ell \mathbf{1} \ , \quad \text{for} \quad \ell \geq 0  \ , 
    \]
    where $\mathbf{1}_{k \times 1}$ is a vector of size $k$ that has 1 entries.  
\end{theorem}
\Cref{thm:neuts} immediately implies a bounded expected queue length:
\begin{align}
    \sum_{\ell \geq 0} \ell \alpha(\ell) = \sum_{q = 1}^\infty \sum_{\ell = q}^\infty \alpha(\ell) = \sum_{q = 1}^\infty \sum_{\ell = q}^\infty  {\bs \theta}^\top (I - R) R^\ell \mathbf{1} &= \sum_{q = 1}^\infty {\bs \theta}^\top (I - R) \left( \sum_{\ell = q}^\infty   R^\ell\right) \mathbf{1}  =  \sum_{q = 1}^\infty {\bs \theta}^\top (I - R) (I-R)^{-1} R^q \mathbf{1} \notag \\ &= {\bs \theta}^\top \left(\sum_{q = 1}^\infty R^q \right)   \mathbf{1} = {\bs \theta}^\top R(I-R)^{-1}  \mathbf{1}  < \infty  \ , \label{eq:modul_expec}
\end{align}
where we used the fact that $R$ has a spectral radius less than 1.

We now use the above result to prove \Cref{lem:scheduler_bounded}. The queue ${\cal V}_{i,j}$ is incremented if and only if a customer $\mathfrak{c}$ of type $j$ arrives and Line \ref{lin:task_increment} is reached with $i^{\rm long} = i$. In order for this line to be reached, (i) $\mathfrak{c}$ must be included in the assignment of short queues, drawn in Line \ref{lin:s_def}, and (ii) we must have $i^{\rm long} = i$. Condition (ii) happens, independently of (i), with probability $(1-\eps)y_{i,j}$. Therefore, the increment rate of ${\cal V}_{i,j}$ only depends on the state of short queues; let $\vartheta({\bs \ell})$ be the increment rate when the state of short queues is ${\bs \ell} \in [\bar{\ell}]_0^{n_s}$. Letting $\bs{M}^{\rm short}$ be the assignment in Line \ref{lin:s_def} upon arrival of $\mathfrak{c}$, condition (i) holds with probability
    \begin{align}
        \pr{j \in \bs{M}^{\rm short}} &= \sum_{\boldsymbol{\ell} \in [\bar{\ell}]_0^{n_s}} \pr{{\cal L}_{\rm short}^{\pi^{\rm pr}}(t) = \boldsymbol{\ell}} \cdot \frac{\sum\limits_{\substack{{\bs{M}} \in {\cal D}({\cal S}^{\rm short}) : j \in {\bs{M}}}} x_{\bs{M}}^{\boldsymbol{\ell}}}{\sum\limits_{{\bs{M}} \in {\cal D}({\cal S}^{\rm short})} x_{\bs{M}}^{\boldsymbol{\ell}}} \notag \\ &= \underbrace{\sum_{\boldsymbol{\ell} \in [\bar{\ell}]_0^{n_s}}  {\sum\limits_{\substack{{\bs{M}} \in {\cal D}({\cal S}^{\rm short}) : \\ j \in {\bs{M}}}} x_{\bs{M}}^{\boldsymbol{\ell}}}}_{:= x_{{\rm short}, j}} \ . \label{eq:xshort}
    \end{align}    
    Since type-$j$ customers arrive with rate $\gamma_j$, the expected increment rate of ${\cal V}_{i,j}$, at every time $t \geq 0$, is 
    \begin{align}
        \sum_{{\bs \ell} \in [\bar{\ell}]_0^{n_s}} \pr{{\cal L}_{\rm short}(t) = \boldsymbol{\ell}} \cdot \vartheta({\bs \ell}) = \gamma_j \cdot (1-\eps)y_{i,j} \cdot x_{{\rm short},j} \ . \label{expr:inc_prob}
    \end{align}
    On the other hand, if condition (i) above does \emph{not} hold and we have $i^{\rm long} = \bot$ and $i^\dagger = i$, the virtual buffer ${\cal V}_{i,j}$ is decremented. Clearly, these events are independent. Similar to the increment rate, the decrement rate of ${\cal V}_{i,j}$ only depends on the state of short queues and we use $\phi({\bs \ell})$ to denote the decrement rate when the short queues are at state ${\bs \ell} \in [\bar{\ell}]_0^{n_s}$.\footnote{We could also have some decrements that are due to the fact that $i^{\rm long} \neq \bot$ but $\ell_{i^{\rm long} = 0}$; we ignore those decrements to simplify the analysis while retaining the stability of the queue.} Therefore, in light of equality \eqref{eq:xshort}, when $\mathfrak{c}$ arrives, the probability that \Cref{alg:priority_matching} reaches Line \ref{lin:abundant_contentious_decrement} is at least
    \begin{align}
        (1-x_{{\rm short}, j}) \cdot \frac{y_{i, j}}{\sum_{i'\in {\cal S}^{\rm long}} y_{i',j}} \cdot \left(1-\sum_{i \in {\cal S}^{\rm long}} (1-\eps)y_{i,j}\right) \ . \label{expr:decrement_prob}
    \end{align}
    Now, recall inequality \eqref{ineq:NLP-contention-constraint} $1-x_{{\rm short}, j} \geq \sum_{i' \in {\cal S}^{\rm long}} y_{i',j}$. It also holds, by a rearrangement, that $1-\sum_{i \in {\cal S}^{\rm long}} (1-\eps)y_{i,j} \geq x_{{\rm short}, j}$. Combining these two inequalities with expression \eqref{expr:decrement_prob} shows that the expected decrement rate of ${\cal V}_{i,j}$ (when ${\cal V}_{i,j} > 0$) is at least 
    \begin{align}
        \sum_{{\bs \ell} \in [\bar{\ell}]_0^{n_s}} \pr{{\cal L}_{\rm short}(t) = \boldsymbol{\ell}} \cdot \phi({\bs \ell}) = \gamma_j \cdot y_{i,j} \cdot x_{{\rm short},j} \ .  \label{expr:dec_prob}
    \end{align}
    Given the expected increment and decrement rates \eqref{expr:inc_prob} and \eqref{expr:dec_prob}, the virtual buffer ${\cal V}_{i,j}$ satisfies the stability condition \eqref{ineq:modul_stab}, and thus, by the bound $\eqref{eq:modul_expec}$, it has a bounded steady-state expectation which proves the lemma.

\subsection{Proof of \Cref{lem:ineq:line_oncu_reach}}\label{prf:lem_non_cont}
For convenience, let $j(t) = j$. The proof is essentially immediate from the definition of non-contentious types. By \Cref{lem:short_convergence} and \eqref{eq:short_stationary_dist}, the probability of $\{i^{\rm short}(t) \neq \bot\}$ is

     \[ 
        \sum\limits_{\boldsymbol{\ell} \in [\bar{\ell}]_0^{n_s}} \pr{{\cal L}_{\rm short}^{\pi^{\rm pr}}(t) = \boldsymbol{\ell}} \cdot \frac{\sum_{{\bs{M}} \in {\cal D}({\cal S}^{\rm short}): j \in {\bs{M}}}x_{\bs{M}}^{\boldsymbol{\ell}}}{\sum_{{\bs{M}} \in {\cal D}({\cal S}^{\rm short})} x_{\bs{M}}^{\boldsymbol{\ell}}} = \sum_{\MyAtop{{\bs{M}} \in {\cal D}({\cal S}^{\rm short}):}{j\in {\bs{M}}}} \sum_{{{\boldsymbol{\ell}} \in [\bar{\ell}]_0^{n_s}}} x_{\bs{M}}^{\boldsymbol{\ell}} \leq {\eps} \ ,
     \] 
     where the inequality follows from the definition of $\overline{{\cal C}^{\rm ct}}$ in equation~\eqref{ineq:contentious}. 

\subsection{Upper bound for the length of an M/M/1 queue with equal arrival and service rate, in the presence of abandonments}\label{app:exp_queue_ub}
\begin{claim}\label{clm:exp_queue_ub}
    Consider an independent M/M/1 queue $\tilde{\cal L}$ where suppliers arrive with rate $\lambda$ and are matched to customers with rate $\lambda$. Furthermore, each unmatched supplier abandons the queue with rate 1. Then, we can obtain the following upper bound for the steady-state queue length:
    \[ \ex{\tilde{\cal L}(\infty)} \leq\sqrt{{\lambda}} \ . \]
\end{claim}
{\bf Proof.} The idea is to use the fact that drift conditions imply moment bounds:
\begin{proposition}[\cite{hajek2015random}, Prop. 6.17] \label{prop:drift_moment}
                Consider an irreducible continuous-time Markov process $X$ on a countable set ${S}$, with generator matrix $Q$. If $V$ is a function on ${S}$, $QV$ represents the drift vector: $QV(v) = \sum_{u \in {S}, u \neq v} q_{v u} (V(u)-V(v))$. Suppose $V$, $f$, and $g$ are nonnegative functions defined on ${S}$ such that 
            \begin{align}
                QV(v) \leq -f(v) + g(v) \quad \forall v \in {S} \ .  \notag             
            \end{align}
            In addition, suppose for some $\varepsilon > 0$, that $C$ defined by $C = \{ v : f(v) < g(v) + \varepsilon \}$ is finite. Then, $X$ is positive recurrent and $\ex{f(X(\infty))} \leq \ex{g(X(\infty)}$. 
\end{proposition}
We use \Cref{prop:drift_moment} with $X(t) = \tilde{\cal L}(t)$ and $V(\ell) = \ell^2$. We argue that
\begin{align*}
    QV(\ell) &= \lambda((\ell+1)^2 - \ell^2) + (\lambda + \ell)((\ell - 1)^2 - \ell^2) - \lambda \mathbb{I}\{\ell = 0\} \\ &= -2\ell^2 + \ell + 2\lambda - \lambda \mathbb{I}\{\ell = 0\} \ .
\end{align*}
Letting $f(\ell) = 2\ell^2 +  \lambda \mathbb{I}\{\ell = 0\}$ and $g(\ell) = \ell + 2\lambda$, we can write
\begin{align}\label{eq:moment_result}
    \ex{2\tilde{\cal L}^2(\infty) + \lambda \cdot \mathbb{I}\left\{\tilde{\cal L}(\infty)=0\right\}} \leq \ex{\tilde{\cal L}(\infty) + 2\lambda} \ . 
\end{align} 
Observe that the flow balance equation implies \[ \lambda \cdot \pr{\tilde{\cal L}(\infty) = 0} = \lambda - \lambda \cdot \pr{\tilde{\cal L}(\infty) \geq 1} = \ex{\tilde{\cal L}(\infty)} \ . \] Plugging this expression into equation \eqref{eq:moment_result} yields
$\expar{2\tilde{\cal L}^2(\infty)} \leq 2\lambda$, which leads to $\expar{\tilde{\cal L}(\infty)} \leq \sqrt{\lambda}$ by the Jensen's inequality.


\section{Euclidean Matching Policy: A Formal Description}\phantomsection\label{app:euclidean}

\subsection{Preliminaries: near-optimal non-crossing policies}\phantomsection\label{ssec:non-crossing}
We construct a {\em random grid} ${\cal G}$, i.e., a partition of the space $[0,1]^d$ into cells, and then show that we can restrict attention to the matches that do not cross the cells' boundaries while losing only a small fraction of throughput. 

Specifically, we set the length of the cells as ${\eta} = \frac{16d{c}^*}{\eps^2 \tau^*}$ and uniformly draw a shift $\boldsymbol{\delta} \sim U[0,\eta]^d$. Next, we construct the grid ${\cal G} = \{C_1,\ldots,C_g\}$ where each cell is of the form $[\delta_1 + t_1\eta) \times \cdots \times [\delta_d + t_d\eta)$ for $t_k \in \{0, 1, \cdots, K\}$ and $K \leq \lceil \frac{1}{\eta} \rceil \leq \lceil\frac{\eps^2 \tau^*}{16dc^*}\rceil$. Intuitively, our construction ensures that matches across two different cells do not constitute a large fraction of the throughput of a near-optimal policy. To formalize this idea, we introduce the notion of non-crossing policies.

\begin{definition}[Non-crossing Policy] Policy $\pi$ is said to be non-crossing with respect to grid $\mathcal{G}$ if no match occurs between locations that belong to different cells. 
\end{definition}
The next claim establishes that we can focus on non-crossing policies when approximating our bi-criteria dynamic matching problem. The proof is in Appendix \ref{prf:noncrossing-policy}.
\begin{lemma}\label{lem:noncrossing-policy}
If the cost-throughput target $(c^*,\tau^*)$ is attainable, then there exists a non-crossing policy $\pi^{\rm nc}$ with respect to grid ${\mathcal G}$ such that $c({\pi^{\rm nc}}) \leq c^*$ with probability 1 and $\expartwo{\mathcal{G}}{\tau({\pi^{\rm nc}})} \geq (1-\eps) \tau^*$.
\end{lemma}

Following \Cref{lem:noncrossing-policy}, we focus on constructing a non-crossing policy $\pi^{\rm nc}$  such that $c({\pi^{\rm nc}}) \leq c^*$ and $\expar{\tau({\pi^{\rm nc}})} \geq (1-\eps) \tau^*$. However, for any fixed realization of the random grid ${\mathcal G}$, it is unclear how we should adjust the throughput rate target; indeed, \Cref{lem:noncrossing-policy} ensures that there exists $\pi^{nc}$ such that $\extwo{\cal G}{\pi^{\rm nc}} \geq (1-\eps)\tau^*$ but this does not say we should set a target $\tau(\pi^{nc}) \approx (1-\eps)\tau^*$ for a fixed ${\cal G}$. 

To transcribe \Cref{lem:noncrossing-policy} algorithmically, we need to find a non-crossing policy that achieves the cost-throughput target $(c^*, \tau_{\mathcal G})$, where $\tau_{\mathcal G}$ is the largest throughput rate such that $(c^*, \tau_{\mathcal G})$ is attainable by a non-crossing policy with respect to ${\cal G}$. Although $\tau_{\mathcal G}$ is not known a priori, if we have an approximation scheme for any cost-throughput target, then it suffices to call this algorithm with a target $(c^*,\hat{\tau}_{\cal G})$, where $\hat{\tau}_{\cal G}$ is a $(1+\eps)$-underestimate of $\tau_{\cal G}$. For example, $\hat{\tau}_{\cal G}$ can be determined through binary search.  Therefore, in the remainder of this section, we assume that we have access to $\tau_{\mathcal G}$ and our objective is to find a non-crossing policy for the target $(c^*, \tau_{\mathcal G})$. Note that we do not consider values of $\tau_{\cal G}$ that are smaller than $\eps \tau^*$; this would affect $\extwo{\cal G}{\tau(\pi^{\rm nc})}$ by at most an $\eps$-factor.

\subsection{Decomposition into local-cell instances}\phantomsection\label{ssec:decomposition_lp}
Our grid reduces the problem to \emph{local-cell} instances where we can use the priority rounding policy. Therefore, we should determine the throughput rate that our policy $\pi^{\rm nc}$ satisfies within each cell such that the overall cost is less than $c^*$. This task, which resembles a knapsack problem, is performed by a \emph{decomposition} LP. 

\paragraph{Local-cell instances.} Having partitioned the space into different cells $C_1, \ldots, {C}_{g}$, we introduce a corresponding notion of local-cell instances. 
\begin{definition}[local-cell instance] \label{def:local_cell}
    For each cell ${C}$ and throughput target $t$, we define $\phi({C},t)$ as the minimum attainable cost rate by a policy, called $\Phi(C, t)$, that only serves customers in ${C}$ using suppliers therein subject to the minimum throughput rate constraint $\tau({\Phi(C, t))} \geq~t$. \footnote{If $t$ is not attainable, $\phi(C, t)$ is defined to be $\infty$.} Relatedly, $\hat{\Phi}(C, t)$ is an approximation algorithm that satisfies $\tau({\hat{\Phi}(C, t))} \geq (1-\eps/2)t$ and $c({\hat{\Phi}(C, t))} \leq \phi(C,t) + \eps c^* \frac{t}{2\tau^*}$.
\end{definition} 
The approximation algorithm $\hat{\Phi}(C, t)$ is obtained via a clustering of close types inside $C$ which leads to a constant number of types and allows us to use the matching policy $\pi^{\rm pr}$ for constant-size networks. This idea will be formalized in the next section. For the rest of this section, however, we assume access to such an approximation scheme.

Although there are $g = \Theta((\frac{\tau^*\eps^2}{dc^*})^d)$ different cells, many of these cells are equivalent---this occurs when they have the same relative position of suppliers and customers. Therefore, from a computational perspective, it suffices to consider the collection of $g_{\rm unq}$ unique cells ${\bf C}_1, \ldots, {\bf C}_{g_{\rm unq}}$, where each cell type ${\bf C}_u$ is repeated $r_u$ times among all the $g$ cells, i.e., $\sum_{u = 1}^{g_{\rm unq}} r_u = g$. Since there are at most $n + m$ cells that contain at least one supplier or customer type, we have $g_{\rm unq} \leq n + m$. This observation helps us with the cost-throughput knapsack problem, as explained next.

\paragraph{The cost-throughput knapsack problem and decomposition LP.} 
A non-crossing policy decides the throughput rate from matches made inside every cell. Therefore, our goal is to find the throughput rate for each cell such that the aggregate cost rate is at most $c^*$. This goal can be interpreted as a knapsack problem where our goal is to maximize the the throughput rate subject to the constraint on total cost. 

To find an approximate solution for this knapsack problem, we use a discretization of the throughputs for local-cell instances by powers of $(1+\eps)$. This operation, when done for every cell, must not decrease throughput or increase cost by more than a factor of $\eps$. Specifically, let 
\[D = \{0\} \cup \left\{\frac{\eps \tau_{\cal G}}{2g} \cdot (1+\eps)^k \; \left| \; 0 \leq k \leq \left\lceil\log\left(\frac{2g}{\eps}\right)  \cdot \log^{-1}(1+\eps) \right \rceil \right\}\right. \ . \] 
Going forward, we construct local-cell instances in which the throughput rate targets are values in $D$. To ease the notation, let $D$ be represented as $D = \{\zeta_1, \ldots, \zeta_{|D|}\}$.

Consequently, we are ready to formulate our {\em decomposition LP} that solves our discretized knapsack problem, i.e., assigns throughput rate targets in $D$ to the cells. Later, we show how to convert the corresponding solution into a feasible matching policy with the desired performance guarantees.
 
We consider the set of all pairs $(u, k) \in [g_{\rm unq}] \times [|D|]$ that correspond to unique cells and throughput targets in the discretized set. Then, we define the cost vector $\boldsymbol{z} = (z(u,k))_{(u, k) \in {[g_{\rm unq}] \times [|D|]}}$ where $z(u,k) = c(\hat{\Phi}({\bf C}_{u}, \zeta_k))$, recalling that $\hat{\Phi}(\cdot, \cdot)$ is an $\eps$-approximation of $\Phi(\cdot, \cdot)$. Our LP decision variables $\boldsymbol{x} = \left(x(u, k)\right)_{(u, k) \in {[g_{\rm unq}] \times [|D|]}}$ count, in coordinates $(u,k)$, the number of cells similar to ${\bf C}_u$ that get assigned to the throughput rate target $\zeta_k$. The goal is to optimally allocate the cost-throughput target across different local-cell instances:
 \begin{align*}
  \quad & \min_{\boldsymbol{x}}  && \boldsymbol{z} \cdot \boldsymbol{x} \notag \\
& \text { s.t. }  && \sum_{k = 1}^{|D|} x(u, k) \leq r_u \ , & \forall u \in [g_{\rm unq}]    \\
&  && \sum_{k = 1}^{|D|} \sum_{u =1}^{g_{\rm unq}} x(u, k) \cdot \zeta_k \geq (1-\eps)\tau_{\mathcal{G}} \ ,   \\
&  && \boldsymbol{x} \geq 0 \ . \notag 
\end{align*}
The next lemma shows that the decomposition LP finds a near-optimal throughput target assignment to cells. The proof is in Appendix \ref{prf:gluing-feasibility}.
\begin{lemma}\label{lem:gluing-feasibility}
    The decomposition LP is feasible and has a value at most $(1+\eps/2) c^*$. Moreover, it is poly-sized and hence solvable in time $\poly(\max\{\log^d(\frac{g}{\eps}),1\} \cdot |{\cal I}|)$, where $|{\cal I}|$ is the size of the input. 
\end{lemma}

\paragraph{\bf Converting an LP solution into our Euclidean matching policy $\boldsymbol{\pi^{\rm approx}}$.} It remains to explain how the output of the LP can lead to a matching policy achieving the desired cost-throughput target. 

After solving the decomposition LP, we convert the fractional solution $x(u,k)_{u \in [g_{\rm unq}], k \in [|D|]}$ to a non-crossing policy $\pi^{\rm approx}$. In the LP solution, $x(u,k)$ should be interpreted as the number of cells equivalent to ${\bf C}_u$ that are assigned to throughput target $\zeta_k$. However, some $x(u,k)$ may be fractional. As a ``rounding'', we assign all cells of type ${\bf C}_u$ to the target $\tau_u = \frac{1}{r_u}\sum_{k = 1}^D x(u,k) \zeta_k$, i.e., the weighted mean throughput target. This rounding is motivated by the fact that the local-cell optimal cost function $\phi(C, t)$ is convex in $t$, as formalized by the next lemma (see the proof in Appendix \ref{prf:phi_convexity}).
\begin{lemma}\label{lem:phi_convexity}
    For any cell $C$ and attainable throughput target $t$, $\phi(C,t)$ is convex in $t$.
\end{lemma}

Recall that we are designing a non-crossing policy. Therefore, it suffices to let each local-cell evolve and operate independently using its own matching policy and throughput target.  Namely, we let $\tilde{\tau}_1, \cdots, \tilde{\tau}_g$ be the throughput targets from the decomposition LP, i.e., we have $\tilde{\tau}_k = \tau_u$ if $C_k$ is equivalent to ${\bf C}_u$. Then, our policy $\pi^{\rm approx}$ is the concatenation of $\hat{\Phi}(C_k, \tilde{\tau}_k)$ for every $k \in [g]$. In other words, upon the arrival of a customer of type $j$ with $l^{j,C} \in C_k$, $\pi^{\rm approx}$ follows the same decision as $\hat{\Phi}(C_k, \tilde{\tau}_k)$ either by matching that customer with a supplier in $C_k$ or by choosing not to match. The following lemma analyzes this policy and shows that it approximately obtains the cost-throughput target. The proof is deferred to Appendix \ref{prf:approx_policy_perf}.
\begin{lemma}\label{lem:approx_policy_perf}
    The policy $\pi^{\rm approx}$ satisfies $\tau(\pi^{\rm approx}) \geq (1-\eps)\tau_{\cal G}$ and $c(\pi^{\rm approx}) \leq (1+\eps)c^*$. 
\end{lemma}

So far, we have presented an efficient policy $\pi^{\rm approx}$ that approximates the optimal non-crossing policy, given that it has access to an efficient approximation scheme $\hat{\Phi}(\cdot,\cdot)$ for the local-cell instances. Nevertheless, the design of these local-cell approximation schemes is an important question that we address in the next section. 

\subsection{Clustering of Supplier and Customer Types: Construction of $\boldsymbol{\hat{\Phi}(C, t)}$}\phantomsection\label{ssec:clustering} 
To start, we exploit the structure of the local-cell instances to cluster all supplier and customers into a small (constant in terms of $1/\eps$) number of types. This property stems from the small dimension of each cell, relative to the targeted cost rate, enabling us to round each location to cluster types in a pre-defined grid.

Specifically, within each cell $C$ with side length $\eta = \frac{16dc^*}{\eps^2\tau^{*}}$, we construct an inner grid of side length $\eta_{\rm inner} = \frac{{c}^* \eps}{\tau^*\sqrt{d}}$. Concretely, first, suppose without loss of generality that $C = [0, \eta]^d$. Then, we define the inner grid as the set of locations 
\[ \mathcal{G}_{\rm inner} = \left\{\sum_{s=1}^d k_s \boldsymbol{e}_s
\; \left| \; 0 \leq k_s \leq \left \lceil \frac{16{d}^{\nicefrac{3}{2}}}{\eps^3}  \right \rceil \; \forall s \in [d] \right\}\right. \ , \]  where $\boldsymbol{e}_s$ is a unit vector in ${\bb R}^d$ with 1 in its $s$-th coordinate. Consequently, for every supplier or customer type at location $l \in C$, we will consider a modified location, which is the closest point of $\mathcal{G}_{\rm inner}$, as its new location. We use ${\cal I}_C$ and $\tilde{\cal I}_C$ to refer to the original and modified instances (restricted to cell $C$), respectively. Here, the set of customer and supplier types is precisely the set of clustered locations ${\cal G}_{inner}$ and the arrival rates of each location is the combined rate of original types rounded to that cluster. Importantly, $\tilde{\cal I}_C$ is an instance that has a constant (for a fixed $\eps$) number of supplier and customer types. If a throughput $\tau$ is attainable in ${\cal I}_C$, it is also attainable in $\tilde{\cal I}_C$. In the next lemma (see proof in Appendix \ref{prf:new-instance}), we show how a policy in one instance (original or modified) leads to a policy in the other instance with an approximately equal cost and the same throughput. We use $\tilde{c}(\cdot)$ and $\tilde{\tau}(\cdot)$ to refer to the expected average cost and throughput rates with respect to the new instance $\tilde{\cal I}_{C}$.

\begin{lemma}\label{lem:new-instance}
       For any policy $\pi$ ($\tilde{\pi}$) with respect to ${\cal I}_C$ ($\tilde{\cal I}_C$), there exists a policy $\tilde{\pi}$ ($\pi$) for $\tilde{\cal I}_C$ (${\cal I}_C$) such that $\tilde{c}(\tilde{\pi}) \leq {c}({\pi})+ \eps c^* \frac{\tau(\pi)}{4\tau^*}$ (${c}({\pi}) \leq \tilde{c}(\tilde{\pi})+ \eps c^* \frac{\tilde{\tau}(\tilde{\pi})}{4\tau^*}$) and $\tilde{\tau}(\tilde{\pi}) = {\tau}({\pi})$. 
\end{lemma}

\paragraph{Construction of $\hat{\Phi}(C, t)$.} One consequence of \Cref{lem:new-instance} is that we can efficiently construct the approximation algorithm $\hat{\Phi}(C, t)$ defined in \Cref{ssec:decomposition_lp}. Indeed, given any attainable throughput rate $t$, we can perform the clustering to obtain the modified instance $\tilde{\cal I}_C$ and use $\pi^{\rm pr}$ (policy in \Cref{sec:constant_ptas}, with $n = O(\frac{\sqrt{d}}{\eps})^d$) to obtain $\tilde{\pi}_{C}$ such that $\tilde{\tau}(\tilde{\pi}) \geq (1-\eps/2)t$ and $\tilde{c}(\tilde{\pi}_C) \leq \phi(\tilde{\cal I}_C, t)$. Then, we can construct $\hat{\Phi}(C, t)$ from \Cref{lem:new-instance} given $\tilde{\pi}_C$. The proof of this lemma devises $\hat{\Phi}(C, t)$ by mimicing $\tilde{\pi}_C$: $\hat{\Phi}(C, t)$ matches a customer to a supplier if and only their corresponding supplier and customer are matched by $\tilde{\pi}_C$ in $\tilde{\cal I}_C$. Consequently, by \Cref{lem:new-instance}, we have
\[ 
    c(\hat{\Phi}(C, t)) \leq \tilde{c}(\tilde{\pi}_C) + \eps c^* \frac{t}{4\tau^*} \leq \phi(\tilde{\cal I}_C, t) + \eps c^* \frac{t}{4\tau^*} \leq \phi({\cal I}_C, t) + \eps c^* \frac{t}{2\tau^*} \  .
\]
Since we also have $\tau(\hat{\Phi}(C, t)) \geq (1-\eps/2)t$, this construction satisfies our desired properties in \Cref{def:local_cell}. 

In light of this result, we now conclude this section by a runtime analysis of our Euclidean matching policy $\pi^{\rm approx}$. 

\paragraph{Runtime analysis of $\pi^{\rm approx}$.} Recall that $\pi^{\rm approx}$ first solves the decomposition LP in time $\poly((n+m) \cdot \log^d(\frac{g}{\eps}) \cdot |{\cal I}|)$, given access to the cost vector $\boldsymbol{z}$ where $z(u, k) = c(\hat{\Phi}({\bf C}_u, \zeta_k)$. To obtain this cost vector, however, we must use the local-cell approximation $\hat{\Phi}(\cdot, \cdot)$ for $|D| \cdot g_{\rm unq}$ times. 

Our local-cell approximation uses the $\pi^{\rm pr}$ policy, where $n = (16{d}^{\nicefrac{3}{2}}/\eps^3)^d$ (for the local-cell problem) and the throughput target is at least $\tau_{\cal G} / g$. Recall that we are not considering values of $\tau_{\cal G}$ due to their small effect on achieving a target throughput of $\tau^*$. Therefore, the runtime of $\pi^{\rm approx}$ is $\poly(\max\{\log(\frac{\tau^*\eps}{dc^*}),1\} \cdot (\frac{2d}{\eps})^{\frac{d}{\eps}(\frac{d}{\eps})^{4d}}  \cdot |{\cal I}|)$.

\section{Additional Proofs from Section \ref{sec:euclidean} and Appendix \ref{app:euclidean}}

\subsection{Proof of \Cref{lem:noncrossing-policy}}\label{prf:noncrossing-policy}
    For any fixed deterministic policy $\pi$ that attains the cost-throughput target $(c^*, \tau^*)$, we introduce its non-crossing counterpart policy $\pi^{\rm nc}$ by foregoing any match that occurs between a supplier and a customer located in different cells of $\mathcal{G}$. However, $\pi^{\rm nc}$  still discards any supplier matched by $\pi$ to ensure that state of the system (i.e., queue lengths) evolve identically to those in $\pi$. It is thus clear that $c({\pi^{\rm nc}}) \leq c^*$. What remains to be shown is that $\pi^{\rm nc}$ achieves an expected throughput rate of $(1+\eps)^{-1}\tau^*$. To this end, we distinguish between two types of matches, depending on whether or not their cost is larger than $\bar{c} = \frac{4c^*}{\tau^*\eps}$, where we note that $\bar{c} < \eta$. We lower bound the contributions of these matches to $\pi$'s expected throughput rate separately:

    \paragraph{Case 1: Crossing matches of cost at least $ \bar{c}$:} If the expected average match rate of edges of cost at least $ \frac{4{c}^*}{\tau^*\eps}$ under $\pi$ were to exceed $\frac{\tau^*\eps}{2}$, then $c(\pi) \geq \bar{c} \cdot \frac{\tau^*\eps}{2} = 2c^*$, which contradicts the claim hypothesis that $c(\pi) \leq c^*$. Hence, foregoing such crossing matches reduces the throughput rate by at most $\frac{\tau^*\eps}{2}$. 
        
    \paragraph{Case 2: Crossing matches of cost smaller than $\bar{c}$:} Consider any match under policy $\pi$ between a supplier of type $i$ and a customer of type $j$ such that $c_{i, j} < \bar{c}$. We now consider the $k$-th coordinate, where $1 \leq k \leq d$. Without loss of generality, we assume that ${l}^{i, S}_k \leq {l}^{j, C}_k$. It is clear that ${l}^{j, C}_k - {l}^{i, S}_k \leq \bar{c}$. Let $t'$ be the largest integer such that $t' \eta \leq l^{i, S}_k$. The grid ${\cal G}$ crosses $\boldsymbol{l}^{i, S}$ and $\boldsymbol{l}^{j, C}$ in the $k$-th dimension if and only if $(\delta_k + t'\eta) \in [l_k^{i, S} - \eta, l_k^{j,C} - \eta] \cup [l_k^{i, S}, l_k^{j,C}]$, which occurs with probability exactly $\frac{2(l_k^{j,C} - l_k^{i,S})}{\eta} \leq \frac{2\bar{c}}{\eta} = \frac{\eps}{2d}$ due to the fact that $\delta_k \sim U[0, \eta]$. Hence, the probability that the grid crosses $\boldsymbol{l}^{i, S}$ and $\boldsymbol{l}^{j, C}$ in some dimension, i.e., that $\boldsymbol{l}^{i, S}$ and $\boldsymbol{l}^{j, C}$ lie in different cells, is at most $d \cdot \frac{\eps}{2d} = \frac{\eps}{2}$ by the union bound. Therefore, foregoing such crossing matches reduces the expected throughput rate by at most $\frac{\tau^* \eps}{2}$. 

Combining cases (1) and (2) yields $\expartwo{\mathcal{G}}{\tau({\pi^{\rm nc}})} \geq \tau(\pi) - 2 \cdot \frac{\tau^* \eps}{2} \geq  (1-\eps) \tau^*$, as desired. 

\subsection{Proof of \Cref{lem:gluing-feasibility}}
\phantomsection
\label{prf:gluing-feasibility}
    Consider the optimal non-crossing policy $\pi^{\rm nc}$. This policy assigns the vector $\tau_1, \cdots, \tau_{g}$ of optimal throughput rates to the cells ${C}_1, \ldots, {C}_{g}$ and incurs a matching cost of $\sum_{s = 1}^g \phi(C_s, \tau_s)$ since different cells evolve independently. We now explain how to obtain a feasible solution for the decomposition LP from $\pi^{\rm nc}$.
    For any $1 \leq s \leq g$, if $\tau_s < \frac{\eps \tau_{\cal G}}{2g}$, we round the throughput rate to $0$. Otherwise, there is a unique $\vartheta_s \in D$ such that $\vartheta_s \leq \tau_s \leq (1+\eps)\vartheta_s$. Then, we define a vector $x(u, k)_{u \in [g_{\rm unq}, k \in [|D|]]}$ where $x(u,k)$ is the number of cells in ${C}_1, \ldots, {C}_{g}$, equivalent to ${\bf C}_u$ whose rounded throughput target is equal to $\zeta_k$. Note that the combined decrease in throughput from rounding the throughput targets $\tau_s < \frac{\eps \tau_{\cal G}}{2g}$ to zero is at most $g \cdot \frac{\eps \tau_{\cal G}}{2g} = \eps \tau_{\cal G}/2$. Hence, we obtain 
    \[ \sum_{k = 1}^{|D|} \sum_{u =1}^{g_{\rm unq}} x(u, k) \cdot \zeta_k \geq \sum_{s = 1}^{g} \frac{\tau_s}{1+\eps} - \frac{\eps \tau_{\cal G}}{2} = \frac{\tau_{\cal G}}{1+\eps}- \frac{\eps \tau_{\cal G}}{2} \geq (1-\eps)\tau_{\cal G} \ , \] 
    which shows that $\boldsymbol{x}$ is feasible. Lastly, since we have $\vartheta_s \leq \tau_s$ for every $s \in [g]$ by construction, using the approximation $\hat{\Phi}(\cdot,\cdot)$ implies that the objective value is at most equal to 
    \[  \sum_{k = 1}^{|D|} \sum_{u =1}^{g_{\rm unq}} x(t, k)  z(u,k) = \sum_{s =1}^{g} c(\hat{\Phi}({C}_s, \vartheta_s)) \leq \sum_{s=1}^{g}{\phi}({C}_s, \tau_s) + \sum_{s=1}^g \eps c^* \frac{\tau_s}{2\tau^*} = \sum_{s=1}^{g}{\phi}({C}_s, \tau_s) + \frac{\eps c^*}{2} \ ,
    \] which completes the proof.

\subsection{Proof of \Cref{lem:phi_convexity}}
\phantomsection
\label{prf:phi_convexity}
    To prove the lemma, it suffices to show that for any throughput targets $t_1, t_2 \geq 0$ and probabilities $p_1, p_2 \geq 0$ such that $p_1 t_1 + p_2 t_2$ is attainable and $p_1 + p_2 = 1$, we have $\phi(C, p_1 t_1 + p_2 t_2) \leq p_1\phi(C, t_1) + p_2 \phi(C, t_2)$. 
    This inequality is established by considering the policy 
    \[ \pi =  \begin{cases}
        \Phi(C, t_1) & w.p. \; p_1 \ , \\ \Phi(C, t_2) & w.p. \; p_2 \ .
    \end{cases} \]
    Clearly, $\pi$ has an expected throughput rate $p_1 t_1 + p_2 t_2$ and a cost rate of at most $p_1 \phi(C, t_1) + p_2 \phi(C, t_2)$. Since the optimal policy for throughput rate target $p_1 t_1 + p_2 t_2$, $\Phi(C, p_1 t_1 + p_2 t_2)$, does not have a higher cost rate, the proof is complete. 

\subsection{Proof of \Cref{lem:approx_policy_perf}}\label{prf:approx_policy_perf}    
 Since different local-cells evolve independently under $\pi^{\rm approx}$, the resulting total throughput rate satisfies
\begin{eqnarray}\label{ineq:rounded_throughput}
\tau(\pi^{\rm approx}) = 
\sum_{u=1}^{g_{\rm unq}} r_u \tau(\hat{\Phi}({\bf C}_u,\tau_u)) \geq  (1-\eps/2) \cdot \left(\sum_{u=1}^{g_{\rm unq}} r_u \tau_u\right) \geq (1-\eps)\tau_{\cal G} \ , \notag
\end{eqnarray}
where the second inequality follows from \Cref{lem:gluing-feasibility}. Now, we analyze the expected average cost rate
\begin{eqnarray}
\sum_{u=1}^{g_{\rm unq}} r_u c(\hat{\Phi}({\bf C}_u,\tau_u)) &\leq& \sum_{u=1}^{g_{\rm unq}} r_u \phi({\bf C}_u,\tau_u) + \sum_{u=1}^{g_{\rm unq}} r_u \eps c^* \frac{\tau_u}{2\tau^*} \nonumber\\
&\leq& \frac{\eps c^*}{2} + \sum_{u=1}^{g_{\rm unq}} x(u,k) \phi({\bf C}_u,\zeta_k) \nonumber \\
&\leq& (1+\eps) {c^*} \label{ineq:rounded_cost} \ ,
\end{eqnarray}
where the first inequality follows from the definition of $\hat{\Phi}(\cdot, \cdot)$, the second inequality proceeds from the convexity of $\phi({C},\tau)$ in $\tau$ (\Cref{lem:phi_convexity}), and the third inequality follows from \Cref{lem:gluing-feasibility}.

\subsection{Proof of \Cref{lem:new-instance}}\label{prf:new-instance}
    We explain how to construct $\tilde{\pi}$, given ${\pi}$. The other direction is identical and thus omitted.
    
    Policy $\tilde{\pi}$ mimics ${\pi}$ by using the actual location of every supplier or customer instead of their modified location, i.e., it considers a transformation of $\tilde{\cal I}_C$ to ${\cal I}_C$ and operates with respect to ${\cal I}_C$. Then, $\tilde{\pi}$ always makes the same decision as ${\pi}$. It is immediate from our construction that $\tau(\pi) = \tilde{\tau}(\tilde{\pi})$. Furthermore, no match in $\tilde{\cal I}_C$ can be costlier than its corresponding match in ${\cal I}_C$ by more than $\eta_{\rm inner} \cdot \frac{\sqrt{d}}{2}$. This observation follows from the fact every location in ${\cal I}_C$ has moved in $\tilde{\cal I}_C$ for a length at most $\eta_{\rm inner} \cdot \frac{\sqrt{d}}{2}$, which occurs along the diagonal of a hypercube with side length $\eta_{\rm inner}$. Then, considering the $\tilde{\pi}$'s match rate, we conclude that $\tilde{c}(\tilde{\pi}) \leq {c}({\pi}) + \eta_{\rm inner} \cdot \frac{\sqrt{d}}{2} \cdot \tilde{\tau}(\tilde{\pi}) = c(\pi) + \eps c^* \frac{\tau(\pi)}{4\tau^*}$, which is the desired statement. 
\end{APPENDICES}

\end{document}